\documentclass{lmcs}

\usepackage{macros}
\usetikzlibrary{arrows,calc,arrows.meta,decorations.pathreplacing}

\usepackage{stmaryrd}
\usepackage{thm-restate}
\usepackage[boxed]{algorithm2e}
\SetAlgorithmName{Scheme}{}{}
\tikzset{j0/.style={draw,text centered,rounded corners=2pt}}

\begin{document}
\title[Decisiveness for Countable MDP\texorpdfstring{\lowercase{s}}{s} and Insights for NPLCS\texorpdfstring{\lowercase{s}}{s} and POMDP\texorpdfstring{\lowercase{s}}{s}]{Decisiveness for Countable MDP\texorpdfstring{\lowercase{s}}{s}\texorpdfstring{\\}{ }and Insights for NPLCS\texorpdfstring{\lowercase{s}}{s} and POMDP\texorpdfstring{\lowercase{s}}{s}}
\titlecomment{Extended version of paper~\cite{BBBFV25}, which appeared as a chapter in \emph{Principles of Formal Quantitative Analysis}.}
\author[N.~Bertrand]{Nathalie Bertrand\lmcsorcid{0000-0002-9957-5394}}[a]
\author[P.~Bouyer]{Patricia Bouyer\lmcsorcid{0000-0002-2823-0911}}[b]
\author[T.~Brihaye]{Thomas Brihaye\lmcsorcid{0000-0001-5763-3130}}[c]
\author[P.~Fournier]{Paulin Fournier}
\author[P.~Vandenhove]{Pierre Vandenhove\lmcsorcid{0000-0001-5834-1068}}[c]
\address{IRISA -- Inria, CNRS, Univ. Rennes, Rennes, France}
\address{Universit\'e Paris-Saclay, CNRS, ENS Paris-Saclay, Laboratoire M\'ethodes Formelles, 91190, Gif-sur-Yvette, France}
\address{UMONS -- Universit\'e de Mons, Mons, Belgium}
\begin{abstract}
Markov chains and Markov decision processes (MDPs) are well-established probabilistic models. While finite Markov models are well-understood, analysing their infinite counterparts remains a significant challenge. Decisiveness has proved to be an elegant property for countable Markov chains: it is general enough to be satisfied by several natural classes of countable Markov chains, and it guarantees the termination of simple qualitative and approximate quantitative model-checking algorithms.

In contrast, existing work on the formal analysis of countable MDPs usually relies on \emph{ad hoc} techniques tailored to specific classes. In this paper, we provide a general framework to analyse countable MDPs (with finitely many enabled actions per state) by extending the notion of decisiveness. Compared to Markov chains, MDPs exhibit an additional nondeterministic choice that can be resolved in an adversarial or cooperative way, leading to multiple natural notions of decisiveness. We show that these notions enable the approximation of supremum reachability and infimum reachability (equivalently, supremum safety) probabilities in countable MDPs using generic approximation schemes.

We then instantiate our generic approach to two concrete classes of models inducing countable MDPs: nondeterministic probabilistic lossy channel systems (NPLCSs) and partially observable MDPs (POMDPs). This leads to concrete algorithms to compute approximations to supremum safety probabilities in each of these classes.
\end{abstract}

\maketitle
\section{Introduction}


Formal methods for systems with random behaviours call for
models with probabilistic aspects, and appropriate automated
verification techniques. \emph{Markov chains} are among the simplest probabilistic models;
the verification of finite-state Markov chains has been thoroughly studied in the literature and is supported by several mature tools, including PRISM~\cite{KNP-cav11} and Storm~\cite{HJKQV22}.

\paragraph{Countable Markov chains.} In some cases,
\emph{finite} Markov chains fall short of providing an appropriate
modelling formalism, and \emph{infinite} Markov chains must be
considered. Two broad approaches are common for the model checking of
infinite-state Markov chains. One option is to focus on Markov chains
generated in a specific way; for instance, when the underlying
transition system is the configuration graph of a lossy channel
system~\cite{IN97,ABRS05}, a pushdown automaton~\cite{KEM06}, or a one-counter
system~\cite{EWY10}. In this case, \emph{ad hoc} model-checking techniques
have been developed for qualitative and quantitative analyses. The
second option is to establish general criteria on infinite Markov
chains that are sufficient for their qualitative and/or
quantitative model checking to be computable.

Abdulla et al.\ explored the latter direction and proposed the elegant
notion of \emph{decisive} Markov chains~\cite{ABM07}. Intuitively,
decisive countably infinite Markov chains share key properties with finite-state Markov chains, such as the property that if a state is always reachable with a
positive probability, then it will almost surely be
reached. Precisely, a Markov chain is decisive (with respect to a
target state $\Goal$, from a given initial state $s_0$) if almost all
runs from $s_0$ either reach $\Goal$ or eventually reach a state from which
$\Goal$ is no longer reachable.
This is convenient to deal
with \emph{reachability objectives}, i.e., the event of reaching a
specified set of states.
Assuming decisiveness, the qualitative model checking of reachability
objectives reduces---as in the finite case---to simple graph
analysis. Moreover, decisiveness is a property that allows for
approximating the probability of reachability objectives within any prescribed additive error, and for sampling trajectories for statistical
model-checking of countably infinite Markov chains~\cite{BBH24}. Classes of decisive Markov chains have been exhibited~\cite{ABM07}, as well as classes in which the decisiveness property is decidable~\cite{FHY23}.
A stronger
property for countable Markov chains is the existence of a
\emph{finite attractor}, i.e., a finite set of states that is reached
almost surely from any state of the Markov chain. Sufficient
conditions for the existence of a finite attractor are given
in~\cite{BBS-ipl05}.

\paragraph{Markov decision processes.} Purely probabilistic models are
too limited to represent features such as the lack of
assumptions regarding scheduling policies or relative speeds (in
concurrent systems), the lack of information regarding values
that have been abstracted away (in abstract models), or the latitude
left for delayed implementation decisions (in early designs). In such
situations, it is not desirable to assume the choices to be resolved
probabilistically, and nondeterminism is needed.
\emph{Markov decision processes} (MDPs) are an extension of Markov chains with nondeterministic
choices; they exhibit both nondeterminism and probabilistic phenomena. In MDPs, the nondeterminism is resolved by a \emph{scheduler}.
For a given event, it is relevant to consider both the infimum and supremum, over all schedulers, of the probability of its occurrence.

As with Markov chains, when considering infinite systems, one can either opt for
\emph{ad hoc} model-checking algorithms for classes of infinite-state
MDPs, or derive generic results under appropriate assumptions.
\begin{itemize}
    \item In the first scenario, examples include MDPs which are generated by lossy
    channel systems~\cite{ABRS05,BBS-acmtocl07}, with nondeterministic
    action choices and probabilistic message losses. To the best of our knowledge,
    only qualitative verification algorithms---based on the
    finite-attractor property---have been developed for NPLCSs.
    In particular, the existence of a scheduler that ensures a
    reachability objective with probability~$1$ (or with positive
    probability) is decidable for lossy channel systems~\cite{BBS-acmtocl07}; however, the existence
    of a scheduler ensuring a B\"uchi objective with positive probability
    is undecidable~\cite{BBS-acmtocl07}.
    More generally, there are also examples of \emph{games} on infinite arenas with underlying tractable models for which
    decidability results exist: among others, recursive concurrent stochastic
    games~\cite{EYlmcs08}, pushdown automata~\cite{BBKO11}, one-counter stochastic games~\cite{BBEfsttcs10,BBEK13}, lossy channel systems~\cite{BS-qapl13,ACMSlmcs14}.
    \item In the second scenario, generic quantitative analyses have been developed for MDPs over general, including uncountable, state spaces under various assumptions.
    More specifically, Tkachev et al.~\cite{TMKA17} develop
    dynamic-programming characterisations for maximal and minimal constrained
    reachability, obtaining error bounds under contractivity conditions. Grover et al.~\cite{GKMW22} give an anytime lower- and upper-bounding algorithm for maximal reachability in uncountable
    MDPs under computability, sampling, Lipschitz-continuity, and absorption
    assumptions.
    General countable MDPs have also been considered with the aim of characterising the value function for various quantitative objectives~\cite{Put94} or identifying the resources (memory requirements, randomness) needed by optimal or $\epsilon$-optimal schedulers~\cite{Orn69,Put94,KMSTW20}.
\end{itemize}

\paragraph{Contributions.}
Our focus is the second scenario above, but with a different approach from the cited work: we study countable MDPs given through finitely described infinite-state models and seek exact, incrementally constructible finite subproblems whose lower and upper bounds yield a certified stopping criterion. \begin{itemize}
    \item We first build on the seminal work on decisive Markov
chains~\cite{ABM07} and explore how the notion of decisiveness can be
extended to Markov decision processes. We propose several notions of
decisiveness, notably \emph{inf-decisiveness} and
\emph{sup-decisiveness}, which differ in how the nondeterminism is resolved (adversarially or cooperatively). These notions are
natural extensions of the existing decisiveness for Markov
chains.
\item Second, we provide approximation schemes for the infimum and
supremum probabilities of reachability objectives. These schemes
provide a non-decreasing sequence of lower bounds, as well as a
non-increasing sequence of upper bounds, for the probability one
wishes to compute.
\item Third, we identify sufficient conditions related to
decisiveness for the lower- and upper-bound sequences to converge towards the same limit,
which is necessary for the schemes to terminate for any given error
margin. We obtain that for inf-decisive (resp.\ sup-decisive) MDPs with finitely many actions enabled at each state, one can approximate
the infimum (resp.\ supremum) reachability probability within any prescribed additive error.
\end{itemize}

We end the paper by instantiating our generic approach to two concrete classes of models inducing countably infinite MDPs of different natures: \emph{nondeterministic probabilistic lossy channel systems} and finite \emph{partially observable MDPs}.
Using decisiveness, we show in both classes that the infimum reachability probabilities can be approximated up to any desired precision.
As we will discuss, existing algorithms often focus on \emph{qualitative} problems (e.g., whether there is a scheduler that reaches a state almost surely) due to the undecidability of many other quantitative problems.
Our quantitative results in these two classes are summarised in Table~\ref{tab:results-intro}.

\begin{table}[tbp]
  \centering
  \renewcommand{\arraystretch}{1.5}
  \begin{tabular}{@{}ccc@{}}
    \toprule
    & NPLCSs & POMDPs\\
    \midrule
    $\Prob^{\inf}(\F \Goal)$ & \textbf{Approximable} (Theorem~\ref{thm:nplcs-inf-decid}) & \textbf{Approximable} (Theorem~\ref{thm:pomdpSafetyDecidable})\\
    $\Prob^{\sup}(\F \Goal)$ & {\renewcommand{\arraystretch}{0.88}\begin{tabular}{@{}c@{}}Open in general\\\textbf{Approximable if non-fleeing}\\(Proposition~\ref{prop:nplcs-sup-nonfleeing})\end{tabular}} & Not approximable~\cite{MHC03}\\
    \bottomrule
  \end{tabular}
  \caption{Results about quantitative reachability problems in NPLCSs and POMDPs. Results from this paper are in bold. To the best of our knowledge, whether $\Prob^{\sup}(\F \Goal)$ can be approximated in NPLCSs is still open for unrestricted NPLCSs. Our result is conditional: we show it can be approximated under a decidable sufficient condition called \emph{non-fleeingness}.}
  \label{tab:results-intro}
\end{table}

For consistency, we mostly discuss \emph{reachability} objectives throughout the paper. However, note that minimising the probability of a reachability objective is equivalent to maximising the probability of the dual \emph{safety} objective (consisting of avoiding a specified set of states).
All results regarding the infimum probability of a reachability objective can therefore be thought of as results about the supremum probability of a safety objective (and vice versa).

See Figure~\ref{fig:summary} for a summary of our main results.

\begin{figure}[tb]
\centering
\begin{tikzpicture}[
    >=Stealth,
    box/.style={rectangle, draw, thick, minimum width=3.5cm, minimum height=0.8cm, align=center, fill=white},
    rounded/.style={rectangle, draw, thick, rounded corners, minimum width=3.5cm, minimum height=1.2cm, align=center, fill=white},
    every node/.append style={font=\small}
]

    \newcommand{\symA}{$\bigstar$}
    \newcommand{\symB}{$\triangle$}

    \node[box] (attr) at (-1.5, 4.5) {$\mathcal{M}$ has a finite\\attractor};
    \node[box] (sem) at (4.5, 4.5) {$\mathcal{M}$ semantically stopping};

    \node[box] (inf) at (-4.5, 3) {$\mathcal{M}$ $\inf$-decisive};
    \node[box] (sup) at (4.5, 3) {$\mathcal{M}$ $\sup$-decisive};

    \node[box] (univ) at (0, 1.5) {$\mathcal{M}$ univ-decisive};

    \node[rounded] (inf_unf) at (-4.5, 0.5) {$\ApproxScheme_\unfold^{\inf}$\\[0.1cm] converges};
    \node[rounded] (sup_unf) at (4.5, 0.5) {$\ApproxScheme_\unfold^{\sup}$\\[0.1cm] converges};

    \node[rounded] (inf_sli) at (-4.5, -2.5) {$\ApproxScheme_{\sliced}^{\inf}$\\[0.1cm] converges};
    \node[rounded] (sup_sli) at (4.5, -2.5) {$\ApproxScheme_{\sliced}^{\sup}$\\[0.1cm] converges};

    \node[box] (nonf) at (0, -4.0) {$\mathcal{M}$ non-fleeing};

    \draw[->, thick] (sem) -- node[right] {Prop.~\ref{prop:coarse}} (sup);

    \draw[->, thick] (attr.south)to[out=270, in=90] node[above=5pt,xshift=14pt] {Prop.~\ref{prop:finiteAttractors}} ([xshift=-0.5cm]univ.north);

    \draw[->, thick] (sup.west) -| ([xshift=0.5cm]univ.north);

    \draw[->, thick] (univ.west) -- ++(-.6, 0) |- (inf.east);

    \draw[->, thick] (inf) -- (inf_unf) node[midway, right] {\symA} node[midway, left] {Thm.~\ref{thm:approx-opt}};
    \draw[->, thick] (sup) -- (sup_unf) node[midway, left] {\symA} node[midway, right] {Thm.~\ref{thm:approx-opt}};

    \draw[->, thick] (sup_unf) --  node[midway, right] {Thm.~\ref{thm:refinement}} (sup_sli);

    \draw[->, thick] ([xshift=-0.5cm]inf_unf.south) to[bend right=30] node[left] {Thm.~\ref{thm:refinement}} ([xshift=-0.5cm]inf_sli.north);
    \draw[->, thick] ([xshift=0.5cm]inf_sli.north) to[bend right=30] node[midway, left] {\symB} node[midway, right] {Thm.~\ref{thm:approxSchemeTwoImpliesOne}} ([xshift=0.5cm]inf_unf.south);

    \draw[ thick] (nonf.north) -- (0, -2.5);
    \draw[->, thick] (0, -2.5) -- (sup_sli.west) node[midway, above] {\symB} node[midway, below] {Thm.~\ref{thm:nonFleeing}};

    \draw[->, thick] (sup.east) -- ++(1.5, 0) |- (nonf.east);

    \node[align=left, anchor=north west, font=\small] at (-6.5, -4.0) {
        \symA \ \ Finite action-branching \\[0.2cm]
        \symB \ \ Finite branching
    };
\end{tikzpicture}
\caption{Summary of our results. The arrows indicate implications, and the symbols $\bigstar$ and $\triangle$ indicate the conditions under which the implications hold. Decisiveness, non-fleeingness, and convergence are with respect to a fixed initial state and target state. Unlabelled arrows follow directly from the definitions of the properties.}
\label{fig:summary}
\end{figure}

\paragraph{Related version.}
This article extends a short version by the same authors published as a chapter in \emph{Principles of Formal Quantitative Analysis}~\cite{BBBFV25}.
Compared to this short version, multiple results have been added (Section~\ref{sec:stepBounded}, Lemma~\ref{lem:finiteEndComponents}, Lemma~\ref{lem:nplcs-nonfleeing}, Proposition~\ref{prop:nplcs-sup-nonfleeing}), and a simplification in the definition of \emph{non-fleeing} (Definition~\ref{def:nonFleeing}) led to cleaner proofs.
We also include all proofs omitted from the short version, as well as new summary tables and figures.

\paragraph{Organisation.}
Section~\ref{sec:preliminaries} introduces the basic definitions and notation used throughout the paper.
Section~\ref{sec:decisiveness} introduces and discusses decisiveness notions for MDPs, as well as general criteria for these notions.
Section~\ref{sec:constructions} presents the generic approximation schemes for infimum and supremum reachability probabilities in countable MDPs.
Section~\ref{sec:convergingSchemes} provides sufficient conditions for the convergence of these schemes.
Section~\ref{sec:applications} presents applications of our generic approach to NPLCSs (infimum and supremum reachability) and POMDPs (infimum reachability).
To keep the presentation focused on the main ideas, we defer some supporting technical details to the appendix.

\paragraph{AI disclosure.}
The statements and proof ideas of Lemmas~\ref{lem:finiteEndComponents} and~\ref{lem:nplcs-nonfleeing} were suggested by GPT-5.5 Pro while considering extensions of the results from the short version~\cite{BBBFV25}.
All the text in this article has been written and checked by the authors, who take full responsibility for the content of the article.

\section{Preliminaries} \label{sec:preliminaries}


\subsection{Markov decision processes}

\begin{defi}
  A \emph{Markov decision process} (MDP) is a tuple
  $\calM = (S,\Act,\proba)$ where $S$ is a countable set of states,
  $\Act$ is a countable set of actions,
  $\proba \colon S \times \Act \times S \to [0,1]$ is a
  probabilistic transition function satisfying
  $\sum_{s' \in S} \proba(s,a,s') \in \{0,1\}$ for all
  $(s,a) \in S \times \Act$.
\end{defi}

An MDP $\calM$ is \emph{finite} if $S$ and $\Act$ are finite. Let $\calM = (S,\Act,\proba)$ be an MDP.
Given $(s,a) \in S \times \Act$, we say that the action~$a$ is \emph{enabled} at state~$s$
whenever $\sum_{s' \in S} \proba(s,a,s') = 1$.
We overload $\Act$ and write $\En(s)$ for the set of actions enabled at $s$.
We assume that each state has at least one enabled action.
A state $s$ is \emph{absorbing} if for all enabled actions $a\in\En(s)$, $\proba(s, a, s) = 1$.
The MDP~$\calM$ is \emph{finitely action-branching} if for every $s\in S$, $\En(s)$ is finite. It is \emph{finitely \probBranching}
if for every $(s,a) \in S \times \Act$, the support of $\proba(s,a,\cdot)$
is finite. It is \emph{finitely branching} if it is both finitely
action-branching and finitely \probBranching.

A \emph{history} (resp.\ \emph{path}) in $\calM$ is an element
$s_0a_1s_1\cdots a_ks_k$ of $S \times (\Act \times S)^*$ (resp.\ $s_0a_1s_1a_2s_2\cdots$ of $S \times (\Act\times S)^\omega$) such that for every
$0\le i < k$ (resp.\ $i \ge 0$), $a_{i+1}$ is enabled at $s_i$ and $\proba(s_i,a_{i+1},s_{i+1})>0$. We write $\Hist(\calM)$ for the set of histories in~$\calM$
and $\Paths(\calM)$ for the set of paths in $\calM$. We define the
\emph{length} of a history $h = s_0a_1s_1 \cdots a_ks_k$ as $k$, and denote
its last state by $\last(h) = s_k$. We write $has$ for the history obtained by concatenating $as$ to $h$.

We consider the $\sigma$-algebra generated by all cylinder sets in
$\Paths(\calM)$: for a history $h \in \Hist(\calM)$, the \emph{cylinder
generated by $h$} is
\[
\Cyl(h) = \{\rho \in \Paths(\calM) \mid h\ \text{is a prefix of}\
\rho\}
\enspace.
\]

\begin{defi}
  A \emph{scheduler} in $\calM$ is a function
  $\sigma \colon \Hist(\calM) \to \Dist(\Act)$ which assigns a
  probability distribution over actions to any history, with the
  constraint that for every $h\in \Hist(\calM)$, the support of
  $\sigma(h)$ is included in $\En(\last(h))$. We write $\Sched(\calM)$ for the set of schedulers
  in~$\calM$.
\end{defi}

Schedulers are sometimes called \emph{strategies} or \emph{policies} in the literature.
We fix a scheduler~$\sigma$ in $\calM$.
If $\sigma$ only depends on the last state of histories,
i.e., if $\last(h) = \last(h')$ implies
$\sigma(h) = \sigma(h')$, then it is called \emph{positional}. If for
every history~$h$, $\sigma(h)$ is a Dirac probability measure, it is called
\emph{pure}. A pure and positional scheduler can alternatively be
described as a function $\sigma \colon S \to \Act$ such that $\sigma(s)\in \En(s)$ for all $s \in S$. We write
$\SchedPurePos(\calM)$ for the set of pure and positional schedulers in~$\calM$, and $\SchedPureMemoryful(\calM)$ for the set of pure
(\emph{a priori} not positional, that is, \emph{history-dependent}) schedulers.

Given a scheduler $\sigma$ in $\calM$ and an initial state
$s_0 \in S$, one can define a probability measure
$\Prob^\sigma_{\calM,s_0}$ on $\Paths(\calM)$ inductively as follows:
\begin{itemize}
\item $\Prob^\sigma_{\calM,s_0}(\Cyl(s_0)) = 1$ and $\Prob^\sigma_{\calM,s_0}(\Cyl(s)) = 0$ for all $s \neq s_0$;
\item if $h = s_0a_1s_1 \cdots a_ks_k \in \Hist(\calM)$ and $ha_{k+1}s_{k+1} \in \Hist(\calM)$, then
  \[
  \Prob^\sigma_{\calM,s_0}(\Cyl(ha_{k+1}s_{k+1})) =
  \Prob^\sigma_{\calM,s_0}(\Cyl(h)) \cdot \sigma(h)(a_{k+1}) \cdot \proba(s_k,a_{k+1},s_{k+1})
  \enspace.
  \]
\end{itemize}

Figure~\ref{fig:ex-mdp} presents an example of a countably infinite MDP, which is finitely
branching. Under a scheduler which always selects $\alpha$, this
yields a random walk~\cite[Section~3.1]{Ser13}.
It is ``diverging'' if $p>\frac{1}{2}$, which entails that the probability $\lambda_p$ of \emph{not} reaching $\Goal$ is positive from $s_1$. In particular, in this case, the infimum probability of
reaching~$\Goal$ from $s_1$ depends on the relative values of $q$ and
$\lambda_p$.

\begin{figure}[tb]
  \centering
  \begin{tikzpicture}
    \path (7,1.8) node[j0] (goalir) {$\Goal$};
    \path (7,-1.8) node[j0] (sinkir) {$\bad$};
    \path (5,1.8) node[j0] (goali) {$\Goal$};
    \path (5,-1.8) node[j0] (sinki) {$\bad$};
    \path (3,1.8) node[j0] (goalil) {$\Goal$};
    \path (3,-1.8) node[j0] (sinkil) {$\bad$};
    \path (1,1.8) node[j0] (goal1) {$\Goal$};
    \path (1,-1.8) node[j0] (sink1) {$\bad$};
    \path (-1,1.8) node[j0] (goal0) {$\Goal$};
    \path (-1,-1.8) node[j0] (sink0) {$\bad$};

    \path (-1,0) node[j0] (s0) {$s_1$};
    \path (1,0) node[j0] (s1) {$s_2$};
    \path (2,0) node (dots) {$\cdots$};
    \path (3,0) node[j0] (sil) {$s_{i - 1}$};
    \path (5,0) node[j0] (si) {$s_i$};
    \path (7,0) node[j0] (sir) {$s_{i + 1}$};
    \path (8,0) node (dots) {$\cdots$};
    \path (-3,0) node[j0] (goal) {$\Goal$};

    \path (s0) edge[-latex',bend left] node[pos=.5,above] {$\alpha,p$} (s1);
    \path (si) edge[-latex',bend left] node[pos=.5,above] {$\alpha,p$} (sir);
    \path (sil) edge[-latex',bend left] node[pos=.5,above] {$\alpha,p$} (si);

    \path (sir) edge[-latex',bend left] node[pos=.5,below] {$\alpha,1 - p$} (si);
    \path (si) edge[-latex',bend left] node[pos=.5,below] {$\alpha,1 - p$} (sil);
    \path (s1) edge[-latex',bend left] node[pos=.5,below] {$\alpha,1 - p$} (s0);

    \path (si) edge[-latex'] node[pos=.75,right] {$\beta,q$} (goali);
    \path (si) edge[-latex'] node[pos=.75,left] {$\beta,1 - q$} (sinki);

    \path (sil) edge[-latex'] node[pos=.75,right] {$\beta,q$} (goalil);
    \path (sil) edge[-latex'] node[pos=.75,left] {$\beta,1 - q$} (sinkil);

    \path (sir) edge[-latex'] node[pos=.75,right] {$\beta,q$} (goalir);
    \path (sir) edge[-latex'] node[pos=.75,left] {$\beta,1 - q$} (sinkir);

    \path (s1) edge[-latex'] node[pos=.75,right] {$\beta,q$} (goal1);
    \path (s1) edge[-latex'] node[pos=.75,left] {$\beta,1 - q$} (sink1);
    \path (s0) edge[-latex'] node[pos=.75,right] {$\beta,q$} (goal0);
    \path (s0) edge[-latex'] node[pos=.75,left] {$\beta,1 - q$} (sink0);

    \path (s0) edge[-latex'] node[pos=.5,below] {$\alpha,1 - p$} (goal);
  \end{tikzpicture}
  \caption{Example of a finitely branching MDP with an infinite state
    space with parameters $p, q\in(0, 1)$. For readability, the absorbing states
    $\Goal$ and $\bad$ are duplicated in the figure.
    Self-loops on absorbing states are omitted.}
  \label{fig:ex-mdp}
\end{figure}

\subsection{Optimum reachability probabilities}
Depending on the application, the nondeterminism in Markov decision
processes can be thought of as adversarial or as cooperative.
For the
probability of a given event, it thus makes sense to consider
both the infimum and supremum probabilities when ranging over
all schedulers.

We describe path properties using the standard \textsf{LTL} operators
eventually~\F and always~\G, and their step-bounded variants \F[\le n] and \G[\le n]. Let $\rho = s_0a_1s_1 \cdots \in \Paths(\calM)$ be a path in~$\calM$. If $\psi$ is a state property, the path property $\F \psi$
holds on $\rho$ if there is some index $k \in \nats$ such that $s_k$
satisfies $\psi$. Given $n\in \nats$, $\F[\le n]\psi$ holds on $\rho$ if
there is some index $k \le n$ such that~$s_k$ satisfies $\psi$.
Dually, $\rho$ satisfies $\G \psi$ if all indices~$k\in \nats$ are
such that~$s_k$ satisfies $\psi$, and~$\rho$ satisfies $\G[\le n]
\psi$ if for all indices~$k \le n$, $s_k$ satisfies~$\psi$. Now, given
a path property $\phi$, we write $\sem{\phi}{\calM,s_0}$ for the set
of paths from $s_0$ in $\calM$ that satisfy $\phi$.

In this paper, we focus on optimising the probability of reachability
objectives, and thus aim at computing or approximating the following
values: given an MDP~$\calM$, an initial state~$s_0$, a set of target states
$\targetset$, and
$\opt \in \{\inf,\sup\}$,
\[
  \Prob^{\opt}_{\calM,s_0}(\F{\targetset}) \eqdef \opt_{\sigma \in \Sched(\calM)}
  \Prob^\sigma_{\calM,s_0}(\F{\targetset}) \enspace.
\]
Without loss of generality, one can assume that $\targetset$ consists of a
single absorbing state, which we denote $\Goal$ in the sequel.

\begin{rem}
  The literature often considers \emph{safety objectives}, which correspond to events $\G \neg \targetset$ for $\targetset$ a set of states.
  Note that by the duality of reachability and safety objectives, all
  results below also hold for safety objectives by inverting $\inf$ and $\sup$.
\end{rem}

For finite MDPs with rational transition probabilities, the computation of the above values for $\opt = \inf$
and $\opt = \sup$ is well-known (see
e.g.~\cite[Chap.~10]{BK08}). It reduces to solving a polynomial-size linear
program, resulting in a polynomial-time algorithm.
Moreover, the infimum and supremum values are
attained by pure and positional schedulers, as stated below.

\begin{lem} \label{lem:finiteMDP}
  Let $\calM$ be a finite MDP, $s_0$ be an initial state, and $\Goal$ be a
  target state. Then, for $\opt \in \{\inf,\sup\}$, there exists
  a pure and positional scheduler $\sigma^{\opt} \in \SchedPurePos(\calM)$ such that
  $\Prob^{\sigma^{\opt}}_{\calM,s_0}(\F \Goal) =
  \Prob^{\opt}_{\calM,s_0}(\F \Goal)$.
\end{lem}

As an alternative to solving a linear program, value-iteration
  techniques can also be used and often turn out to be more efficient
  in practice; see~\cite{HM18}.
They rely on a fixed-point
  characterisation (the \emph{Bellman equations}) of the values $\val^{\opt}_{\calM}(s)
  \eqdef \Prob^{\opt}_{\calM,s} (\F
  \Goal)$, where $\opt \in \{\inf,\sup\}$.
  This characterisation also holds for finitely action-branching countable MDPs~\cite{Put94}, and
  can even be extended to stochastic turn-based two-player games with
  reachability objectives~\cite{kucera11,BBKO11}.
Yet, the convergence of value iteration to the fixed point does not imply the existence of a \emph{stopping criterion} that can be used to identify when the computed value is sufficiently close to the actual value.

We recall existing results about the complexity of optimal schedulers for reachability objectives in MDPs, which we will use in later sections. The two items below are implied respectively by~\cite[Theorem~7.3.6]{Put94} and~\cite[Theorem~B]{Orn69}.
The latter was also discussed more recently in~\cite{KMSTW20}.

\begin{lem} \label{lem:schedulerComplexity}
  Let $\calM = (S, \Act, \transitions)$ be a countable MDP and $\Goal\in S$ be a target state.
  \begin{enumerate}
  \item Assume $\calM$ is finitely action-branching.
  There exists $\sigma \in \SchedPurePos(\calM)$ such that for all $s
    \in S$, $\Prob^{\sigma}_{\calM,s}(\F \Goal) =
    \Prob^{\inf}_{\calM,s}(\F \Goal)$.\label{item:infOptPurePos}%
  \item For all $\epsilon > 0$, there exists $\sigma \in \SchedPurePos(\calM)$ such that for all $s \in S$, $\Prob^{\sigma}_{\calM,s}(\F \Goal) \ge \Prob^{\sup}_{\calM,s}(\F \Goal) - \epsilon$.\label{item:supEpsOptPurePos}%
  \end{enumerate}%
\end{lem}%

Two remarks are useful:
\begin{itemize}
\item The finite action-branching assumption is needed for the first item. Optimal schedulers for infimum reachability probabilities may not exist for infinitely action-branching MDPs, and $\epsilon$-optimal schedulers may even require memory~\cite[Theorem~3]{KMSW17}.
\item For supremum reachability probabilities,
optimal schedulers may not exist, even in finitely branching MDPs; such an example is provided
in~\cite[Figure~1]{KMSTW20}. This is why we only consider $\epsilon$-optimal schedulers in the second item.
Interestingly, item~\ref{item:supEpsOptPurePos} fails to hold in MDPs with an uncountable state space~\cite[Theorem~A]{Orn69}.
As per the definition above, all MDPs in this paper are assumed countable.
\end{itemize}

\paragraph{Approximation schemes and algorithms.}
Even if characterisations of the values
exist in infinite MDPs~\cite{Put94}, no general
algorithm is known to compute $\Prob^{\inf}_{\calM,s_0}(\F{\Goal})$
and $\Prob^{\sup}_{\calM,s_0}(\F{\Goal})$, or to decide whether
these values exceed a threshold. Of course, such algorithms would
very much depend on the representation of infinite MDPs.

In this paper, we aim at providing generic \emph{approximation
schemes} for infimum and supremum reachability probabilities in
countable MDPs.

\begin{defi}
  An \emph{approximation algorithm} takes as input an effectively represented MDP~$\calM$, an
  initial state~$s_0$, a target state~$\Goal$, an optimisation
  criterion $\opt \in \{\inf,\sup\}$, and a rational precision $\epsilon>0$, and
  returns a rational value $v$ such that
  $|v- \Prob^{\opt}_{\calM,s_0}(\F \Goal)| \le \epsilon$.
\end{defi}

In this paper, we provide generic \emph{approximation schemes}, defined by two
sequences~$(r_n^-)_n$ and $(r_n^+)_n$, respectively non-decreasing and
non-increasing, such that for every $n\in\nats$,
\[r_n^- \le \Prob^{\opt}_{\calM,s_0}(\F \Goal) \le r_n^+\enspace.\]  An
approximation scheme is \emph{converging} on $\calM$ from $s_0$ (for fixed $\opt$ and $\Goal$) if for
every precision $\epsilon>0$, there exists $n \in \nats$ such that
$r_n^+ - r_n^- \le \epsilon$ (which means that any $v$ in the
interval $[r_n^-,r_n^+]$ is a solution to our problem). An
approximation scheme yields an \emph{approximation algorithm} if it is
converging and there is a uniform algorithm that can compute the values $r_n^-$ and $r_n^+$ for arbitrarily large $n$.

Converting a converging approximation scheme into an algorithm requires hypotheses on the MDPs considered (e.g., finitely representable, restrictions on branching). In Section~\ref{sec:constructions}, we focus on approximation schemes; in Section~\ref{sec:convergingSchemes}, we investigate when these schemes are converging; in Section~\ref{sec:applications}, we show on specific classes of models that these converging schemes can be made into algorithms.
All these results will be enabled by the notion of \emph{decisiveness} for MDPs, discussed in Section~\ref{sec:decisiveness}.

\section{Decisiveness for MDPs} \label{sec:decisiveness} In this
section, we define several flavours of \emph{decisiveness} for MDPs,
inspired by the notion of decisiveness defined for Markov
chains~\cite{ABM07}. We fix an MDP $\calM = (S,\Act,\proba)$ and an
absorbing target state $\Goal \in S$ for the rest of this section.

\subsection{Avoid sets}
For Markov chains, the first ingredient to define
decisiveness is the notion of \emph{avoid set}, which is the set of states
from which one can no longer reach~$\Goal$ (the avoid set was denoted
$\widetilde{\Goal}$ in~\cite{ABM07}). We extend this notion in several
directions.

If $\sigma \in \SchedPurePos(\calM)$, we define the \emph{avoid set of
$\calM$ w.r.t.\ $\sigma$} as:
\[
  \Avoid^\sigma_\calM(\Goal) = \big\{s \in S \mid
  \Prob^\sigma_{\calM,s}(\F \Goal) = 0\big\} \enspace.
\]
This is the avoid set of the Markov chain (as defined in~\cite{ABM07}) induced by the pure and positional scheduler $\sigma$ on $\calM$.

We also define two other notions of \emph{avoid set}, depending on whether one
considers the infimum or supremum value over schedulers.
For $\opt \in \{\inf,\sup\}$, we let:
\[
\Avoid^{\opt}_\calM(\Goal) = \big\{s \in S \mid \optOp_{\sigma \in
    \Sched(\calM)} \Prob^\sigma_{\calM,s}(\F \Goal) = 0\big\} \enspace.
\]

Note that
\begin{align*}
  \sup_{\sigma \in \Sched(\calM)} \Prob^\sigma_{\calM,s}(\F \Goal) =
  0 & \quad \Longleftrightarrow \quad & \forall \sigma \in \Sched(\calM),\
                              \Prob^\sigma_{\calM,s}(\F \Goal) = 0 \\
    & \quad \Longleftrightarrow \quad & \forall
                               \sigma \in \SchedPurePos(\calM),\
                               \Prob^\sigma_{\calM,s}(\F \Goal) = 0
                               \enspace,
\end{align*}
where the second equivalence can be shown using Lemma~\ref{lem:schedulerComplexity}, item~\ref{item:supEpsOptPurePos}. We deduce that
\[
  \Avoid_\calM^{\sup}(\Goal) = \bigcap_{\sigma \in
    \SchedPurePos(\calM)} \Avoid^\sigma_{\calM}(\Goal) \enspace.
\]
In contrast, it may happen that
$\inf_{\sigma \in \Sched(\calM)} \Prob^\sigma_{\calM,s}(\F \Goal) =
0$, yet there is no $\sigma \in \Sched(\calM)$ such that
$\Prob^\sigma_{\calM,s}(\F \Goal) = 0$. For instance, on the MDP $\calM^{\mathtt{L}}$ in
Figure~\ref{fig:distinguishingAvoidSets} (left), when choosing action $\alpha_i$
from $s_0$, the probability of $\F \Goal$ is $\frac 1 {2^i}$. Recall
that given Lemma~\ref{lem:schedulerComplexity} (item~\ref{item:infOptPurePos}), this behaviour
requires infinite action-branching.
When $\calM$ is finitely action-branching, we have that there exists a pure and positional scheduler $\sigma_{\inf}$ such that $\Avoid^{\sigma_{\inf}}_{\calM}(\Goal) = \Avoid^{\inf}_{\calM}(\Goal)$.

\begin{figure}[tb]
  \centering
  \begin{minipage}{.64\textwidth}
    \centering
  \begin{tikzpicture}
    \path (-0.7,0) node[j0] (goal) {$\Goal$};
    \path (4.4,-1.8) node[j0] (sink) {$\bad$};
    \path (-0.5,1.2) node[] (ML) {$\calM^{\mathtt{L}}$};

    \path (4.4,1.8) node[j0] (s0) {$s_0$};
    \path (1,0) node[j0] (s1) {$s_1$};
    \path (1.9,0) node (dots) {$\cdots$};
    \path (2.8,0) node[j0] (sil) {$s_{i - 1}$};
    \path (4.4,0) node[j0] (si) {$s_i$};
    \path (6.,0) node[j0] (sir) {$s_{i + 1}$};
    \path (7,0) node (dots) {$\cdots$};

    \path (s0) edge[-latex'] node[pos=.4,left] {$\alpha_i,1$} (si);

    \draw[-latex',rounded corners=2mm] (s0.-180) -| node[pos=0.2,above]
    {$\alpha_1,1$} (s1.90);
    \draw[-latex',rounded corners=2mm] (s0.0) -| node[pos=0.3,above]
    {$\alpha_{i + 1},1$} (sir.90);

     \path (sir) edge[-latex'] node[pos=.5,above] {$\beta,\frac 1 2$} (si);
     \path (si) edge[-latex'] node[pos=.5,above] {$\beta,\frac 1 2$} (sil);

    \path (si) edge[-latex'] node[pos=.5,left] {$\beta,\frac 1 2$} (sink);

    \draw[-latex',rounded corners=2mm] (sir.-90) |- node[pos=0.7,below]
    {$\beta,\frac 1 2$} (sink);

    \draw[-latex',rounded corners=2mm] (s1.-90) |- node[pos=0.7,below]
    {$\beta,\frac 1 2$} (sink);

    \path (s1) edge[-latex'] node[pos=.5,above] {$\beta,\frac 1 2$} (goal);
  \end{tikzpicture}
  \end{minipage}%
    \begin{minipage}{.36\textwidth}
      \centering
      \begin{tikzpicture}
        \path (0,0) node[j0] (s0) {$s_0$};
        \path (2,.5) node[j0] (goal) {$\Goal$};
        \path (2,-.5) node[j0] (sink) {$\bad$};
        \path (0.1,1.2) node[] (MR) {$\calM^{\mathtt{R}}$};
        \path (-0,1.5) node[] (empty1) {};
        \path (-0,-1.5) node[] (empty2) {};

        \path (s0) edge[-latex'] node[pos=.5,above,sloped] {$\beta,\frac 1 2$} (goal);
        \path (s0) edge[-latex'] node[pos=.5,below,sloped] {$\beta,\frac 1 2$} (sink);

        \draw [-latex'] (s0) .. controls +(-135:1.3cm) and +(135:1.3cm)
        .. (s0) node [midway,left] {$\alpha,1$};
      \end{tikzpicture}
    \end{minipage}
    \caption{Left: MDP $\calM^{\mathtt{L}}$ for which $\Prob^{\inf}_{\calM^{\mathtt{L}}, s_0}(\F \Goal) = 0$, yet for all schedulers $\sigma$, $\Prob^{\sigma}_{\calM^{\mathtt{L}}, s_0}(\F \Goal) > 0$.
    Right: MDP $\calM^{\mathtt{R}}$ such that $\Avoid^{\sup}_{\calM^{\mathtt{R}}}(\Goal) \neq \Avoid^\sigma_{\calM^{\mathtt{R}}}(\Goal)$ for some scheduler $\sigma$.}
  \label{fig:distinguishingAvoidSets}
\end{figure}

Following the definitions, for every $\sigma \in \SchedPurePos(\calM)$, we have
\[
  \Avoid^{\sup}_\calM(\Goal) \subseteq \Avoid^\sigma_{\calM}(\Goal) \subseteq
  \Avoid^{\inf}_\calM(\Goal) \enspace.
\]
We show two examples to illustrate all kinds of avoid sets and when they can differ.

\begin{exa} \label{ex:distinguishingAvoidSets}
    Consider the three-state MDP $\calM^{\mathtt{R}}$ on the right of Figure~\ref{fig:distinguishingAvoidSets}.
    We have that $\Avoid^{\sup}_{\calM^{\mathtt{R}}}(\Goal) \neq \Avoid^\sigma_{\calM^{\mathtt{R}}}(\Goal)$ for some scheduler $\sigma$: indeed, $\Avoid^{\sup}_{\calM^{\mathtt{R}}}(\Goal) = \{\bad\}$, but for the pure and positional scheduler $\sigma_\alpha$ that chooses $\alpha$ in $s_0$, we have $\Avoid^{\sigma_\alpha}_{\calM^{\mathtt{R}}}(\Goal) = \{\bad, s_0\}$.

    Consider again the infinitely branching MDP $\calM^{\mathtt{L}}$ given
    in Figure~\ref{fig:distinguishingAvoidSets} (left).  We have that
    $\Avoid^{\sigma}_{\calM^{\mathtt{L}}}(\Goal) \neq
    \Avoid^{\inf}_{\calM^{\mathtt{L}}}(\Goal)$ for all pure and
    positional schedulers~$\sigma$: indeed,
    $\Avoid^{\inf}_{\calM^{\mathtt{L}}}(\Goal) = \{s_0, \bad\}$, but
    $\Avoid^{\sigma}_{\calM^{\mathtt{L}}}(\Goal) = \{\bad\}$ for all
    such $\sigma$.
\end{exa}

\subsection{Decisiveness properties}
We now define several notions of decisiveness for MDPs, which are
natural extensions of the decisiveness for Markov
chains~\cite{ABM07}.

\begin{defi}[Decisiveness]
 Let $\calM = (S,\Act,\proba)$ be an MDP, $\Goal \in S$ be an absorbing target state, and $s \in S$ be a state.
 \begin{itemize}
 \item Let $\sigma \in \SchedPurePos(\calM)$. The MDP $\calM$ is
   \emph{$\sigma$-decisive w.r.t.\ $\Goal$ from $s$} whenever
   \[
     \Prob^\sigma_{\calM,s}\big(\F \Goal \vee \F
     \Avoid_\calM^{\sigma}(\Goal) \big) = 1 \enspace.
   \]
 \item The MDP $\calM$ is \emph{univ-decisive w.r.t.\ $\Goal$ from $s$}
   whenever for all $\sigma \in \SchedPurePos(\calM)$, $\calM$ is
   $\sigma$-decisive w.r.t.\ $\Goal$ from $s$; that is,
  \[
    \forall \sigma \in \SchedPurePos(\calM),\
    \Prob^\sigma_{\calM,s}\big(\F \Goal \vee \F
    \Avoid^\sigma_{\calM}(\Goal) \big) = 1 \enspace.
  \]
 \item Let $\opt \in \{\inf,\sup\}$. The MDP $\calM$ is \emph{$\opt$-decisive
     w.r.t.\ $\Goal$ from $s$} whenever
   \[
     \forall \sigma \in \SchedPurePos(\calM),\
     \Prob^\sigma_{\calM,s}\big(\F \Goal \vee \F
     \Avoid_\calM^{\opt}(\Goal) \big) = 1 \enspace.
   \]
 \end{itemize}
\end{defi}

  In the case of Markov chains, all these  notions are equivalent and coincide with the notion of decisiveness defined in~\cite{ABM07}. In the case of MDPs, these notions are different.
  Since
  $\Avoid^{\sup}_\calM(\Goal) \subseteq \Avoid^\sigma_{\calM}(\Goal) \subseteq \Avoid^{\inf}_\calM(\Goal)$ (for all pure and positional schedulers $\sigma$),
  $\sup$-decisiveness is a stronger condition than univ-decisiveness, which is itself stronger than
  $\inf$-decisiveness.
  We show examples distinguishing these notions.

\begin{exa} \label{ex:distinguishingDecisiveness}
  To distinguish $\sup$-decisiveness from univ-decisiveness, consider again the three-state MDP $\calM^{\mathtt{R}}$ from Example~\ref{ex:distinguishingAvoidSets} (Figure~\ref{fig:distinguishingAvoidSets}, right).
  Recall that $\Avoid^{\sup}_{\calM^{\mathtt{R}}}(\Goal) = \{\bad\}$. Hence, for the scheduler $\sigma_\alpha$ that chooses $\alpha$ in~$s_0$, we have $\Prob^{\sigma_\alpha}_{\calM^{\mathtt{R}},s_0}(\F \Goal \vee \F \Avoid^{\sup}_{\calM^{\mathtt{R}}}(\Goal)) = 0$.
  Hence, $\calM^{\mathtt{R}}$ is not $\sup$-decisive w.r.t.~$\Goal$ from~$s_0$.
  On the other hand, we can show it is univ-decisive by considering the only two pure and positional schedulers $\sigma_\alpha$ and $\sigma_\beta$.
  We have $\Avoid^{\sigma_\alpha}_{\calM^{\mathtt{R}}}(\Goal) = \{\bad, s_0\}$, so $\Prob^{\sigma_\alpha}_{\calM^{\mathtt{R}},s_0}(\F \Goal \vee \F \Avoid^{\sigma_\alpha}_{\calM^{\mathtt{R}}}(\Goal)) = 1$.
  We have $\Avoid^{\sigma_\beta}_{\calM^{\mathtt{R}}}(\Goal) = \{\bad\}$, so $\Prob^{\sigma_\beta}_{\calM^{\mathtt{R}},s_0}(\F \Goal \vee \F \Avoid^{\sigma_\beta}_{\calM^{\mathtt{R}}}(\Goal)) = 1$.
  Hence, $\calM^{\mathtt{R}}$ is univ-decisive w.r.t.~$\Goal$ from~$s_0$.

  To distinguish univ-decisiveness from $\inf$-decisiveness, consider the MDP $\calM$ depicted in Figure~\ref{fig:ex-mdp}, and assume that $p > \frac{1}{2}$ (i.e., the random walk when choosing $\alpha$ repeatedly is \emph{diverging}).
  Add to this MDP $\calM$ an initial state $s_0$ from which one can go to any state $s_i$ with an action $\alpha_i$ (this adds infinite action-branching). We have that $\Avoid^{\inf}_{\calM}(\Goal) = \{s_0, \bad\}$, since the probability to reach $\Goal$ can be made arbitrarily small by choosing $\alpha_i$ for a sufficiently large~$i$.
  Hence, $\Prob^{\sigma}_{\calM,s_0}(\F \Goal \vee \F \Avoid^{\inf}_{\calM}(\Goal)) = 1$ for all schedulers~$\sigma$,
  so $\calM$ is $\inf$-decisive w.r.t.~$\Goal$ from $s_0$.
  However, for a scheduler $\sigma$ that only plays $\alpha$ after the first action, we have $\Avoid^{\sigma}_{\calM}(\Goal) = \{\bad\}$, so $\Prob^{\sigma}_{\calM,s_0}(\F \Goal \vee \F \Avoid^{\sigma}_{\calM}(\Goal)) < 1$. So $\calM$ is not univ-decisive w.r.t.\ $\Goal$ from $s_0$.

  We finally show an example which is not $\inf$-decisive (and thus, not univ-decisive or $\sup$-decisive either).
  Consider again the MDP $\calM$ in Figure~\ref{fig:ex-mdp}, also with $p > \frac{1}{2}$ (and without the extra state $s_0$ added above).
  It is such that
  $\Avoid^{\inf}_{\calM}(\Goal) = \{\bad\}$ since there is a positive probability to visit $\Goal$ from every state (except~$\bad$), no matter the scheduler. The MDP~$\calM$ is not $\inf$-decisive from~$s_1$ w.r.t.~$\Goal$, since the scheduler
  which always selects~$\alpha$ avoids~$\Goal$ and~$\bad$ with
  positive probability.
\end{exa}

Observe that the definitions of decisiveness only quantify over \emph{pure and positional} schedulers.
We show that this is not a restriction for $\sup$- and $\inf$-decisiveness under finite action-branching, notably thanks to the scheduler complexity results from Lemma~\ref{lem:schedulerComplexity}.
This follows from Lemma~\ref{lem:sup-inf-decisive-PP} below, where one can take $T = \{\Goal\} \cup \Avoid^{\opt}_{\calM}(\Goal)$ for $\opt \in \{\inf,\sup\}$ to obtain this claim.

\begin{lem} \label{lem:sup-inf-decisive-PP}
  Let $\calM = (S,\Act,\proba)$ be a finitely action-branching MDP, $s\in S$, and $T\subseteq S$.
  Then, if for all pure and positional schedulers $\sigma \in \SchedPurePos(\calM)$, $\Prob^\sigma_{\calM,s}(\F T) = 1$, we have that for all schedulers $\sigma \in \Sched(\calM)$, $\Prob^\sigma_{\calM,s}(\F T) = 1$.
\end{lem}
\begin{proof}
We prove the contrapositive.
Assume that there exists a scheduler $\sigma \in \Sched(\calM)$ such that $\Prob^\sigma_{\calM,s}(\F T) < 1$.
Then, by Lemma~\ref{lem:schedulerComplexity}, item~\eqref{item:infOptPurePos}, using that $\calM$ is finitely action-branching, there exists a pure and positional scheduler $\sigma' \in \SchedPurePos(\calM)$ such that $\Prob^{\sigma'}_{\calM,s}(\F T) \le \Prob^\sigma_{\calM,s}(\F T) < 1$.
\end{proof}

Without finite action-branching, Lemma~\ref{lem:sup-inf-decisive-PP} does not hold: see, e.g.,~\cite[Figure~3a]{KMSW17} with $T = \{t\}$.

Unlike $\sup$- and $\inf$-decisiveness, the restriction to pure and positional schedulers affects univ-decisiveness, even under finite action-branching, as shown in the following remark.

\begin{rem} \label{rem:univ-decisive-PP}
  Quantifying over arbitrary schedulers could cause the non-univ-decisiveness to arise from the scheduler rather than the structure of the MDP.
  To see why, consider again the three-state MDP $\calM^{\mathtt{R}}$ in Figure~\ref{fig:distinguishingAvoidSets} (right).
  Consider the (infinite-memory) scheduler $\sigma$ that, as long as $s_0$ is not left, chooses~$\alpha$ with probability $1 - \frac{1}{2^{i+1}}$ and $\beta$ with probability~$\frac{1}{2^{i+1}}$ at step~$i$.
  This scheduler remains at $s_0$ forever with probability $\prod_i (1-\frac{1}{2^{i+1}}) > 0$.
  In particular, the scheduler avoids the target with positive probability.
  Yet, there is always a non-zero probability to reach $\Goal$.
  Fixing $\sigma$ induces an infinite Markov chain (whose state space is the product of the finite MDP and the memory of the scheduler) whose avoid set is $\{\bad\}$, but we do not have that $\{\Goal, \bad\}$ is reached with probability~$1$.
  If we were to consider such schedulers, the MDP $\calM^{\mathtt{R}}$ would not be univ-decisive.
\end{rem}

\subsection{Decisiveness criteria}
We show how to adapt two existing criteria for the decisiveness of Markov chains~\cite[Lemmas~3.4 \&~3.7]{ABM07} to MDPs.
In both cases, we generalise the definition of a property to MDPs and show that this property implies some form of decisiveness.

The first criterion relates to the existence of a \emph{finite attractor}.
It will be used in Section~\ref{subsec:nplcs} to show that a class of infinite MDPs (\emph{NPLCSs}) is $\inf$-decisive.

\begin{defi} \label{def:attractor}
  Let $\calM = (S,\Act,\proba)$ be an MDP.
  We say that $\calM$ has a \emph{finite attractor} if there exists a finite set $A \subseteq S$ such that from all states $s \in S$, for all schedulers $\sigma \in \SchedPurePos(\calM)$, $\Prob^\sigma_{\calM,s}(\F A) = 1$.
\end{defi}

\begin{rem}
  Quantifying only over \emph{pure and positional} schedulers in the definition of a finite attractor is not a restriction for finitely action-branching MDPs, which can be shown using Lemma~\ref{lem:sup-inf-decisive-PP} with $T = A$.
\end{rem}

\begin{prop} \label{prop:finiteAttractors}
  Let $\calM = (S,\Act,\proba)$ be an MDP and $\Goal\in S$ be an absorbing target state. If $\calM$ has a finite attractor, then $\calM$ is univ-decisive (hence also $\inf$-decisive) w.r.t.\ $\Goal$ from every state.
\end{prop}
\begin{proof}
  To show that $\calM$ is univ-decisive, we show that $\calM$ is $\sigma$-decisive w.r.t.\ $\Goal$ from every state and for every scheduler $\sigma\in\SchedPurePos(\calM)$.
  Let $\sigma\in\SchedPurePos(\calM)$.
  Consider the Markov chain $\calM^\sigma$ induced by $\sigma$ on $\calM$.
  In this Markov chain, the set $A$ is a finite attractor in the sense of~\cite{ABM07}, which is the same notion as Definition~\ref{def:attractor} but with no quantification over schedulers needed in Markov chains.
  Moreover, $\Avoid^\sigma_{\calM}(\Goal)$ is the avoid set of $\calM^\sigma$ in the sense of~\cite{ABM07}, which is simply the set of states from which $\Goal$ cannot be reached in $\calM^\sigma$ (once again, with no need for quantification over schedulers).
  Hence, by~\cite[Lemma~3.4]{ABM07}, for all $s\in S$, $\Prob^\sigma_{\calM,s}\big(\F \Goal \vee \F
  \Avoid^\sigma_{\calM}(\Goal) \big) = 1$.
\end{proof}

Observe that, in particular, all finite MDPs are univ-decisive and $\inf$-decisive (as for finite MDPs, the full state space $S$ is a finite attractor).
However, not all finite MDPs are
$\sup$-decisive; a counterexample was given in Example~\ref{ex:distinguishingDecisiveness}.

In \emph{finite} MDPs, we can relate the notion of $\sup$-decisiveness to the notion of \emph{end component}~\cite{deAlfaroThesis}.
An end component of an MDP $\calM = (S,\Act,\proba)$ is a pair $(R, A)$ where $R \subseteq S$ and $A\colon R \to 2^\Act$ such that $R\neq\emptyset$ and for all $s\in R$, $A(s) \neq \emptyset$, $A(s)\subseteq \En(s)$, and for all $a\in A(s)$, the support of $\proba(s, a, \cdot)$ is included in $R$, and the graph induced by $(R, A)$ is strongly connected.
As end components are commonly studied in MDPs, we formally state the relation here.
We generalise this characterisation to a subclass of finitely branching MDPs in Lemma~\ref{lem:finiteEndComponents}.

\begin{prop} \label{prop:endComponents}
  Let $\calM = (S,\Act,\proba)$ be a finite MDP and $\Goal\in S$ be an absorbing target state. We have that $\calM$ is $\sup$-decisive w.r.t.\ $\Goal$ from every state if and only if for all end components $(R, A)$ of $\calM$, either $R = \{\Goal\}$ or $R \subseteq\Avoid_\calM^{\sup}(\Goal)$.
\end{prop}
\begin{proof}
  Assume first that all end components $(R, A)$ of $\calM$ are either such that $R = \{\Goal\}$ or $R \subseteq \Avoid_\calM^{\sup}(\Goal)$.
  Standard results on end components imply that all schedulers eventually reach an end component almost surely~\cite[Section~3.3]{deAlfaroThesis}.
  Hence, for every scheduler $\sigma\in\SchedPurePos(\calM)$ and from every state $s\in S$, $\Prob^\sigma_{\calM,s}\big(\F \Goal \vee \F \Avoid_\calM^{\sup}(\Goal) \big) = 1$.
  So~$\calM$ is $\sup$-decisive.

  Assume now that there exists an end component $(R, A)$ such that we have neither $R = \{\Goal\}$ nor $R \subseteq \Avoid_\calM^{\sup}(\Goal)$.
  Note that as the states in $R$ are strongly connected, no state in $R$ is part of $\Avoid_\calM^{\sup}(\Goal)$.
  Let $s$ be a state in $R$.
  Standard results on end components imply that there exists a pure and positional scheduler $\sigma$ that, from $s$, only visits states in~$R$ forever.
  Hence, $\Prob^{\sigma}_{\calM,s}\big(\F \Goal \vee \F \Avoid_\calM^{\sup}(\Goal) \big) = 0$.
  So $\calM$ is not $\sup$-decisive.
\end{proof}

This result gives another reason why the three-state MDP $\calM^{\mathtt{R}}$ from Example~\ref{ex:distinguishingAvoidSets} is not $\sup$-decisive, as $(\{s_0\}, s_0 \mapsto \{\alpha\})$ is an end component which is neither~$\{\Goal\}$ nor contained in $\Avoid_\calM^{\sup}(\Goal)$.

For finite MDPs, Proposition~\ref{prop:endComponents} is equivalent to requiring that, after collapsing $\Goal$ and $\Avoid_\calM^{\sup}(\Goal)$ into two
absorbing states, the resulting MDP has no other end components. This is Assumption~1 in~\cite{BCCFKKMPU25} and corresponds to the structural property obtained by the \emph{max-reduction} of~\cite{HM18}; both works use it to establish convergence of their respective lower- and upper-bound iterations.

We now extend a second decisiveness criterion by generalising the concept of \emph{globally coarse} Markov chains~\cite[Lemma~3.7]{ABM07}.
Here, this extension yields a criterion for $\sup$-decisiveness in MDPs.

\begin{defi}
  Let $\calM = (S,\Act,\proba)$ be an MDP with an absorbing target state $\Goal$ and a distinct absorbing state $\bad$ (in particular, $\bad\in\Avoid_\calM^{\sup}(\Goal)$).
  The MDP~$\calM$ is \emph{semantically stopping w.r.t.\ $\Goal$ and $\bad$} if there exists $p > 0$ such that from every state~$s$, for all schedulers $\sigma\in\SchedPurePos(\calM)$, $\Prob^\sigma_{\calM,s}(\F \Goal \vee \F \bad) \ge p$.
\end{defi}

\begin{prop} \label{prop:coarse}
  Let $\calM = (S,\Act,\proba)$ be an MDP,~$\Goal$ be an absorbing target state, and~$\bad$ be an absorbing state.
  If $\calM$ is semantically stopping w.r.t.~$\Goal$ and~$\bad$, then $\calM$ is $\sup$-decisive w.r.t.\ $\Goal$ from every state.
\end{prop}
\begin{proof}
  Let $p > 0$ be such that from every state~$s$, for all schedulers $\sigma\in\SchedPurePos(\calM)$, $\Prob^\sigma_{\calM,s}(\F \Goal \vee \F \bad) \ge p$.
  By the continuity of probabilities, we deduce that for all states~$s$, for all schedulers $\sigma\in\SchedPurePos(\calM)$, there exists $n\in\nats$ such that $\Prob^\sigma_{\calM,s}(\F[\le n] \Goal \vee \F[\le n] \bad) \ge \frac{p}{2}$.
  This implies that for every pure and positional scheduler, the probability to avoid~$\Goal$ and~$\bad$ forever is upper bounded by $\lim_{k\to \infty} (1-\frac{p}{2})^k = 0$.
  Hence, $\Prob^\sigma_{\calM,s}(\F \Goal \vee \F \bad) = 1$.
  Since~$\bad\in\Avoid_\calM^{\sup}(\Goal)$, we have that $\calM$ is $\sup$-decisive w.r.t.\ $\Goal$ from every state.
\end{proof}

We can immediately deduce a natural syntactic class of MDPs that are $\sup$-decisive. We say that an MDP $\calM = (S,\Act,\proba)$ is \emph{stopping} if there exists $p > 0$ such that, from every state $s$, for every action $a\in\En(s)$, $\proba(s, a, \Goal) + \proba(s, a, \bad) \ge p$.
It means that there is a uniformly lower-bounded probability that a path ``terminates'' \emph{at every step}.
This is a natural adaptation to MDPs of the concept of \emph{stopping} introduced by Shapley for stochastic games in 1953~\cite{Sha53} and used, e.g., in~\cite{Con92}.

\section{Generic approximation schemes} \label{sec:constructions}

The objective of this section is to provide generic approximation
schemes for optimum reachability probabilities, and to understand under which conditions they converge.
To justify the correctness of these schemes, we first show technical properties of \emph{step-bounded events}, which are used in the proofs of the main results of this section.

\subsection{Step-bounded events} \label{sec:stepBounded}
We will often consider events (i.e., measurable sets of paths of an MDP) equipped with a finite \emph{step bound}, for instance $\F[\le n] \Goal$.
The reason is that the $\inf$/$\sup$ probability of such events can be computed using simple unfolding techniques.
We now discuss some technical properties of these events, which we call \emph{step-bounded}.
We say that an event $E$ is \emph{step-bounded} if there exists $n\in\nats$ such that for all histories~$h$ with~$|h| \ge n$, $\Cyl(h) \subseteq E$ or $\Cyl(h) \cap E = \emptyset$.
When this holds, we say that $E$ has \emph{step bound~$n$}.
It means that after finitely many steps, we can determine whether the event $E$ holds or not, no matter what happens afterwards.\footnote{\emph{Step-bounded} events correspond to the topological \emph{clopen} events for languages of infinite words over a finite alphabet. When the alphabet is infinite, step-bounded events form a stricter class than clopen events, since a clopen event need not admit a finite step bound.}
We say that a sequence $(E_n)_n$ of events is \emph{monotone} if it is non-increasing or non-decreasing for the inclusion relation.

This subsection is devoted to showing that in a finitely action-branching MDP, for a monotone sequence of step-bounded events with growing step bounds, the limit of the optimal probabilities is equal to the optimal probability of the limit event. Formally, if $(E_n)_n$ is such a sequence, then $\lim_{n\to\infty} \optOp_{\sigma} \Prob_{\mdp,s_0}^\sigma(E_n) = \optOp_{\sigma} \lim_{n\to\infty} \Prob_{\mdp,s_0}^\sigma(E_n)$.
We distinguish two cases.
First, when $\opt = \inf$ and the sequence $(E_n)_n$ is non-increasing (or, equivalently, when $\opt = \sup$ and the sequence is non-decreasing), the result is relatively straightforward and does not require the finite action-branching assumption.

\begin{lem} \label{lem:limoptEqualsoptlimEasy}
  Let $\mdp = (S, \Act, \proba)$ be an MDP, $s_0\in S$ be an initial state, and $\opt = \inf$ (resp.\ $\opt = \sup$).
  Let $(E_n)_{n\in\nats}$ be a non-increasing (resp.\ non-decreasing) sequence of step-bounded events, where $E_n$ has step bound~$n$.
  Then,
  \[
    \lim_{n\to\infty} \optOp_{\sigma} \Prob_{\mdp,s_0}^\sigma(E_n) = \optOp_{\sigma} \lim_{n\to\infty} \Prob_{\mdp,s_0}^\sigma(E_n) \enspace.
  \]
\end{lem}
\begin{proof}
  We focus on the case $\opt = \inf$ with $(E_n)_n$ non-increasing; the case $\opt = \sup$ with~$(E_n)_n$ non-decreasing follows by applying the same argument to the complements $(\neg E_n)_n$, using that $\sup_\sigma \Prob_{\mdp,s_0}^\sigma(E_n) = 1 - \inf_\sigma \Prob_{\mdp,s_0}^\sigma(\neg E_n)$.

  First, for all schedulers $\sigma$ and all $n$,
  \[
  \inf_{\tau}\Prob_{\mdp,s_0}^{\tau}(E_n)
  \le
  \Prob_{\mdp,s_0}^{\sigma}(E_n).
  \]
  Taking limits as $n\to\infty$ and then the infimum over $\sigma$ yields
  \[
  \lim_{n\to\infty}\inf_{\tau}\Prob_{\mdp,s_0}^{\tau}(E_n)
  \le
  \inf_{\sigma}\lim_{n\to\infty}\Prob_{\mdp,s_0}^{\sigma}(E_n).
  \]

  We prove the converse inequality.
  Let $\epsilon > 0$.
  For every $n\in\nats$, let $\sigma_n\in\Sched(\calM)$ be a scheduler such that $\Prob_{\mdp,s_0}^{\sigma_n}(E_n) \le \inf_{\sigma} \Prob_{\mdp,s_0}^\sigma(E_n) + \frac{\epsilon}{2}$.
  Since the sequence $(\inf_{\sigma} \Prob_{\mdp,s_0}^\sigma(E_n))_n$ is non-increasing and converges to its limit, there exists $n_0 \in\nats$ such that $\inf_{\sigma} \Prob_{\mdp,s_0}^\sigma(E_{n_0}) \le \lim_n \inf_{\sigma} \Prob_{\mdp,s_0}^\sigma(E_n) + \frac{\epsilon}{2}$.
  Hence, $\Prob_{\mdp,s_0}^{\sigma_{n_0}}(E_{n_0}) \le \lim_n \inf_{\sigma} \Prob_{\mdp,s_0}^\sigma(E_n) + \epsilon$.
  Hence,
  \begin{align*}
    \inf_{\sigma} \lim_{n\to\infty} \Prob_{\mdp,s_0}^\sigma(E_n)
    &\le \lim_{n\to\infty} \Prob_{\mdp,s_0}^{\sigma_{n_0}}(E_n) &\text{by definition of $\inf$}\\
    &\le \Prob_{\mdp,s_0}^{\sigma_{n_0}}(E_{n_0}) &\text{because $(E_n)_n$ is non-increasing}\\
    &\le \lim_n \inf_{\sigma} \Prob_{\mdp,s_0}^\sigma(E_n) + \epsilon\enspace.
  \end{align*}
  As this holds for every $\epsilon > 0$, we obtain the remaining inequality.
\end{proof}

When $\opt = \inf$ and the sequence $(E_n)_n$ is non-decreasing (or, equivalently, when $\opt = \sup$ and the sequence is non-increasing), the result is more involved and requires the finite action-branching assumption, as shown by the following example.

\begin{exa} \label{ex:schemeUseless}
  Consider the MDP $\calM$ depicted in Figure~\ref{fig:schemeUseless}.
  Note that $\Avoid^{\inf}_{\calM}(\Goal) = \emptyset$, so $\calM^{\inf}$ coincides with $\calM$ up to an unreachable absorbing state.
  From~$s_0$, there is a single decision to make: choose an action~$\alpha_i$ with $i \ge 1$, which determines that $\Goal$ will be reached in exactly~$i$~steps.
  This means that for all $n$, it is possible to avoid reaching $\Goal$ within $n$ steps (e.g., by choosing~$\alpha_{n+1}$).
  Hence, for all $n$, $\Prob^{\inf}_{\calM,s_0}(\F[\le n] \Goal) = 0$.

  However, $\Prob^{\inf}_{\calM,s_0}(\F \Goal) = 1$, as any action leads surely to $\Goal$.
  We conclude that
  \[
  0 = \lim_n \inf_\sigma \Prob^{\sigma}_{\calM,s_0}(\F[\le n] \Goal) < \inf_\sigma \lim_n \Prob^{\sigma}_{\calM,s_0}(\F[\le n] \Goal) = 1\enspace.
  \]
\end{exa}

\begin{figure}
  \centering
    \begin{tikzpicture}
      \path (0,0) node[j0] (s0) {$s_0$};

      \path (2,1.5) node[j0] (goal) {$\Goal$};
      \path (s0) edge[-latex',bend left] node[pos=.5,above right,yshift=6pt,xshift=2pt] {$\alpha_1$} (goal);

      \path (2,.75) node[j0] (s11) {$s_{21}$};
      \path ($(s11)+(2,0)$) node[j0] (s12) {$\Goal$};
      \path (s0) edge[-latex',bend left=12] node[pos=.5,above right,yshift=2pt] {$\alpha_2$} (s11);
      \path (s11) edge[-latex'] (s12);

      \path (2,0) node[j0] (s31) {$s_{31}$};
      \path ($(s31)+(2,0)$) node[j0] (s32) {$s_{32}$};
      \path ($(s32)+(2,0)$) node[j0] (s33) {$\Goal$};
      \path (s0) edge[-latex'] node[pos=.5,above right,xshift=-1pt] {$\alpha_3$} (s31);
      \path (s31) edge[-latex'] (s32);
      \path (s32) edge[-latex'] (s33);

      \path (2,-.75) node[j0] (s41) {$s_{41}$};
      \path ($(s41)+(2,0)$) node[j0] (s42) {$s_{42}$};
      \path ($(s42)+(2,0)$) node[j0] (s43) {$s_{43}$};
      \path ($(s43)+(2,0)$) node[j0] (s44) {$\Goal$};
      \path (s0) edge[-latex',bend right=10] node[pos=.5,above right,yshift=-3pt] {$\alpha_4$} (s41);
      \path (s41) edge[-latex'] (s42);
      \path (s42) edge[-latex'] (s43);
      \path (s43) edge[-latex'] (s44);

      \path (2,-1.4) node[] (dots) {$\vdots$};
    \end{tikzpicture}
  \caption{An infinitely action-branching MDP $\calM$ such that $0 = \lim_n \inf_\sigma \Prob^{\sigma}_{\calM,s_0}(\F[\le n] \Goal) < \inf_\sigma \lim_n \Prob^{\sigma}_{\calM,s_0}(\F[\le n] \Goal) = 1$.}
  \label{fig:schemeUseless}
\end{figure}

We will show that the finite action-branching assumption is sufficient to ensure that for any monotone sequence of step-bounded events, the limit of the optimal probabilities is equal to the optimal probability of the limit event.
To prove this, we need one technical lemma, the proof of which is deferred to Appendix~\ref{app:stepBounded}.

We show that for $\opt = \inf$ and a non-decreasing sequence of step-bounded events (or, equivalently, $\opt = \sup$ and a non-increasing sequence), for any $\epsilon > 0$, there is a single scheduler that is ``$\epsilon$-optimal'' at the limit.

\begin{restatable}{lem}{uniformizingClopen} \label{lem:uniformizingClopen}
  Let $\mdp = (S, \Act, \proba)$ be a finitely action-branching MDP, and $s_0\in S$ be an initial state. Let $(E_n)_{n\in\nats}$ be a sequence of non-decreasing (resp.\ non-increasing) step-bounded events such that for all $n$, $E_n$ has step bound $n$.
  For all $\epsilon \in(0, 1)$, there exists a scheduler $\sigma^\epsilon \in \SchedPureMemoryful(\calM)$ such that $\lim_n \Prob^{\sigma^\epsilon}_{\calM, s_0}(E_n) \le \lim_n \Prob^{\inf}_{\calM, s_0}(E_n) + \epsilon$ (resp.\ $\lim_n \Prob^{\sigma^\epsilon}_{\calM, s_0}(E_n) \ge \lim_n \Prob^{\sup}_{\calM, s_0}(E_n) - \epsilon$).
\end{restatable}

We can now conclude that for any monotone sequence of step-bounded events, the limit of the optimal probabilities is equal to the optimal probability of the limit event.

\begin{thm} \label{thm:limoptEqualsoptlim}
  Let $\mdp = (S, \Act, \proba)$ be an MDP, $s_0\in S$ be an initial state, and $\opt = \inf$ (resp.\ $\opt = \sup$).
  Let $(E_n)_{n\in\nats}$ be a sequence of step-bounded events, where $E_n$ has step bound~$n$.
  Then, if either
  \begin{enumerate}
    \item $(E_n)_n$ is non-increasing (resp.\ non-decreasing), or \label{hyp:1}
    \item $\calM$ is finitely action-branching and $(E_n)_n$ is non-decreasing (resp.\ non-increasing), \label{hyp:2}
  \end{enumerate}
  we have
  \[
    \lim_{n\to\infty} \optOp_{\sigma} \Prob_{\mdp,s_0}^\sigma(E_n) = \optOp_{\sigma} \lim_{n\to\infty} \Prob_{\mdp,s_0}^\sigma(E_n) \enspace.
  \]
\end{thm}
\begin{proof}
  The result under assumption~\eqref{hyp:1} is exactly Lemma~\ref{lem:limoptEqualsoptlimEasy}.

  Under assumption~\eqref{hyp:2}, we give the proof for $\opt = \inf$; the case $\opt = \sup$ is symmetric.
  Notice first that the left-hand side limit exists as the sequence $(\inf_{\sigma} \Prob_{\mdp,s_0}^\sigma(E_n))_{n\in\nats}$ is non-decreasing (due to~$(E_n)_n$ being non-decreasing) and bounded.

  We have in general $\lim_{n\to\infty} \inf_{\sigma} \Prob_{\mdp,s_0}^\sigma(E_n) \le \inf_{\sigma} \lim_{n\to\infty} \Prob_{\mdp,s_0}^\sigma(E_n)$, as any scheduler chosen independently of $n$ on the right-hand side can be chosen for all $n$ on the left-hand side.
  It remains to show the other inequality, which follows from Lemma~\ref{lem:uniformizingClopen} (using that $\calM$ is finitely action-branching): for all $\epsilon > 0$, there exists a scheduler $\sigma^\epsilon \in \SchedPureMemoryful(\calM)$ such that \[\lim_n \Prob^{\sigma^\epsilon}_{\calM, s_0}(E_n) \le \lim_n \Prob^{\inf}_{\calM, s_0}(E_n) + \epsilon\enspace.\]
  Hence, letting $\epsilon\to0$, we obtain $\inf_{\sigma} \lim_{n\to\infty} \Prob_{\mdp,s_0}^\sigma(E_n) \le \lim_n \inf_\sigma \Prob^{\sigma}_{\calM, s_0}(E_n)$.
\end{proof}

\subsection{Collapsing avoid sets and first approximation scheme} \label{subsec:collapse}

Let $\calM = (S,\Act,\proba)$ be an MDP, $s_0 \in S$ be an initial state, and $\Goal \in S$ be an absorbing target state.

For $\opt \in \{\inf,\sup\}$, we build a new MDP
$\calM^\opt = (S^\opt,\Act,\proba^\opt)$ by merging states in
$\Avoid^\opt_{\calM}(\Goal)$ into a fresh absorbing
state $\bad^\opt$.

Formally, $\calM^\opt = (S^\opt,\Act,\proba^\opt)$ with
\begin{itemize}
\item $S^\opt = \Big(S \setminus \Avoid^\opt_\calM(\Goal) \Big) \cup
  \{\bad^\opt\}$;
\item for every $s,s' \in S^\opt \setminus \{\bad^\opt\}$, for
  every $a \in \Act$, $\proba^\opt(s,a,s') = \proba(s,a,s')$;
\item for every $s \in S^\opt \setminus \{\bad^\opt\}$ and $a \in \Act$,
  $\proba^\opt(s,a,\bad^\opt) =
  \sum_{s' \in \Avoid^\opt_\calM(\Goal)} \proba(s,a,s')$;
\item we select an arbitrary action $a \in \Act$ and define $\En(\bad^\opt) = \{a\}$ and $\proba^\opt(\bad^\opt,a,\bad^\opt)=1$.
\end{itemize}

In both cases (when $\opt = \inf$ or when $\opt = \sup$), notice that
$\Goal \in S^\opt$.  Without loss of generality, we assume that the initial state
is preserved in the collapsed MDP (i.e., $s_0 \in S\cap S^\opt$);
otherwise, by definition of $\Avoid^\opt_\calM(\Goal)$,
$\optOp_\sigma \Prob^\sigma_{\calM,s_0}(\F \Goal) =0$ and the value to
be computed is trivially $0$.

Note also the following two properties:
\begin{itemize}
\item for every $s \in S^{\inf} \setminus \{\bad^{\inf}\}$, for every
  $\sigma \in \Sched(\calM^{\inf})$,
  $\Prob^\sigma_{\calM^{\inf},s}(\F \Goal)>0$;
\item for every $s \in S^{\sup} \setminus \{\bad^{\sup}\}$, there is
  $\sigma \in \Sched(\calM^{\sup})$ such that
  $\Prob^\sigma_{\calM^{\sup},s}(\F \Goal)>0$.
\end{itemize}
The above constructions collapsing avoid sets preserve optimum
probabilities (with no prior assumption on $\calM$; proof in Appendix~\ref{app:constructions}):
\begin{restatable}{lem}{Mopt} \label{lem:Mopt}
 Let $\calM = (S,\Act,\proba)$ be an MDP, $\Goal \in S$ be an absorbing target state, $\opt \in \{\inf,\sup\}$, and $s_0 \in S \cap S^\opt$ be an initial state.
 Then,
 $\Prob^{\opt}_{\calM,s_0}(\F \Goal) =
    \Prob^{\opt}_{\calM^\opt,s_0}(\F \Goal)$.
\end{restatable}

A scheduler in $\calM^{\sup}$ lifts to $\calM$ with the same reachability probability by continuing arbitrarily after entering $\Avoid^{\sup}_\calM(\Goal)$.
For $\inf$, the same holds when $\calM$ is finitely action-branching, using a pure positional scheduler that avoids $\Goal$ almost surely on $\Avoid^{\inf}_\calM(\Goal)$ (Lemma~\ref{lem:schedulerComplexity}, item~\ref{item:infOptPurePos}); effective lifting requires such a scheduler to be available.

According to Lemma~\ref{lem:Mopt}, computing the supremum probability
(resp.\ infimum probability) in $\calM$ can equivalently be done in
$\calM^{\sup}$ (resp.\ $\calM^{\inf}$).

To do so, for every integer $n\in\nats$, we define the following events in
$\calM^{\opt}$:
\[
\left\{\begin{array}{l}
    R_n = \F[\le n] \Goal \\
    H_n^{\opt} = \G[\le n] (\neg \Goal \wedge \neg \bad^{\opt})
    \enspace.
  \end{array}\right.
\]
In words, $R_n$ expresses that the target is \emph{reached} within $n$
steps, and $H^{\opt}_n$ denotes that the target has not been reached
within $n$ steps, but that the current state has a positive infimum probability of reaching $\Goal$ (in the case $\opt = \inf$) or can reach $\Goal$ with positive probability (in the case $\opt = \sup$).
Note that
$R_n \vee H_n^{\opt} = \F[\le n] \Goal \vee \G[\le n] \neg
\bad^{\opt}$.

We use these events to find lower and upper bounds on the desired probability $p = \Prob^{\opt}_{\calM^{\opt},s_0}(\F \Goal)$.
The aim is that, thanks to the step bound $n$, these bounds are easier to compute than $p$ in many classes of MDPs.
A lower bound for~$p$ is trivially given by $\Prob^{\opt}_{\calM^{\opt},s_0}(R_n)$: reaching $\Goal$ \emph{within $n$ steps} naturally implies reaching $\Goal$.
An upper bound is given by $\Prob^{\opt}_{\calM^{\opt},s_0}(R_n \vee H^{\opt}_n)$: to reach $\Goal$, it is necessary to either reach $\Goal$ within $n$ steps or to be in a state from which reaching $\Goal$ is still possible after $n$ steps.
We state these observations formally.
\begin{restatable}{lem}{HnRn}
  \label{lem:HnRn}
  For every initial state $s_0 \in S^\opt$ and every $n \in \nats$,
  \begin{align*}
    \Prob^{\opt}_{\calM^{\opt},s_0}(R_n) \le
    \Prob^{\opt}_{\calM^{\opt},s_0}(\F \Goal) &\le
    \Prob^{\opt}_{\calM^{\opt},s_0}(R_n \vee H^{\opt}_n)\\
    &\le
    \Prob^{\opt}_{\calM^{\opt},s_0}(R_n) +
    \Prob^{\sup}_{\calM^{\opt},s_0}(H^{\opt}_n) \enspace.
  \end{align*}
\end{restatable}

Thanks to Lemma~\ref{lem:HnRn}, it is natural to define an
approximation scheme with $\Prob^{\opt}_{\calM^{\opt},s_0}(R_n)$ as a
lower bound, and
$\Prob^{\opt}_{\calM^{\opt},s_0}(R_n \vee H^{\opt}_n)$ as an upper
bound, as formalised in Scheme~\ref{approx_scheme1}.  If the input is a
Markov chain, this corresponds exactly to the path enumeration
algorithm from~\cite[Algorithm~1]{ABM07}.
This is not yet an algorithm, since we make no effectiveness assumptions here; such assumptions are established for concrete classes of MDPs in Section~\ref{sec:applications}.

\begin{algorithm}[htb]
  \SetKwInOut{Input}{Input}\SetKwInOut{Output}{Output} \Input{An MDP
    $\calM$, $s_0 \in S$, $\Goal \in S$, $\epsilon \in (0,1)\cap\bbQ$, and $\opt\in\{\inf,\sup\}$.}
    \Output{Upon termination, a value $v\in[0,1]$ such that
  $v \le \Prob^{\opt}_{\calM,s_0}(\F\Goal) \le v+\epsilon$.} \BlankLine
  $n\gets 0$\;
  \Repeat{$p_n^{\opt,+} - p_n^{\opt,-} \le \epsilon$}{$n \gets n + 1$\;
    $p_n^{\opt,-} \gets \Prob^{\opt}_{\calM^{\opt},s_0}(\F[\le n] \Goal)$\;
    $p_n^{\opt,+} \gets \Prob^{\opt}_{\calM^{\opt},s_0}(\F[\le n] \Goal \vee
    \G[\le n] (\neg \Goal \wedge \neg \bad^{\opt}))$\;}
  \KwRet{$p_n^{\opt,-}$}\;
\caption{$\ApproxScheme_\unfold^{\opt}$}\label{approx_scheme1}
\end{algorithm}

Upon termination at iteration~$n$, the scheme also yields an $\epsilon$-optimal scheduler, provided $\calM$ is finitely action-branching.
Choose a pure scheduler attaining $p_n^{\sup,-}$ for $\sup$, or $p_n^{\inf,+}$ for $\inf$ (Lemma~\ref{lem:stepBounded}) (and extend it with arbitrary pure choices beyond depth~$n$ in $\calM^\opt$ so that it is well-defined).
Its reachability probability is respectively at least $p_n^{\sup,-}$ or at most $p_n^{\inf,+}$, so it is $\epsilon$-optimal from~$s_0$ and lifts to $\calM$ as above.
On effectively given finite unfoldings with rational probabilities, these choices can be computed by backward induction.
The resulting scheduler is not memoryless in $\calM$ in general, since it is built on the unfolding of $\calM^\opt$ up to depth~$n$ (Lemma~\ref{lem:stepBounded}).

We now turn to convergence: by the last inequality of Lemma~\ref{lem:HnRn}, it suffices that the probability of $H_n^\opt$ tends to~$0$ uniformly over all schedulers.

\begin{thm} \label{thm:genericCorrectnessTermination}
  Let $\calM = (S,\Act,\proba)$ be an MDP, $s_0 \in S$ be an initial
  state, $\Goal \in S$ be a target state, and $\opt\in\{\inf,\sup\}$. Assume that
  $\lim_{n \to \infty} \Prob^{\sup}_{\calM^{\opt},s_0}(H^{\opt}_n) =
  0$. Then $\ApproxScheme_\unfold^{\opt}$ provides a converging approximation
  scheme for $\Prob^{\opt}_{\calM,s_0}(\F \Goal)$.
\end{thm}
\begin{proof}
  The sequence $(p_n^{\opt,-})_n$ is non-decreasing and the sequence
  $(p_n^{\opt,+})_n$ is non-increasing. Assuming they converge to the
  same value (which is the case when
  $\lim_{n \to \infty} \Prob^{\sup}_{\calM^{\opt},s_0}(H^{\opt}_n) =
  0$ thanks to Lemma~\ref{lem:HnRn}), we have that $\ApproxScheme_\unfold^{\opt}$ converges, which means that it returns an $\epsilon$-approximation of $\Prob^{\opt}_{\calM^{\opt},s_0}(\F \Goal)$.
  By Lemma~\ref{lem:Mopt}, this corresponds to an
  $\epsilon$-approximation of $\Prob^{\opt}_{\calM,s_0}(\F \Goal)$.
\end{proof}

Scheme $\ApproxScheme_\unfold^{\opt}$ is based on unfoldings of the MDP
to deeper and deeper depths. Precisely, the lower bound
$p_n^{\opt,-}$ is the probability in the unfolding up to depth $n$ of
histories that reach $\Goal$; $p_n^{\opt,+}$ is the probability in
the same unfolding of histories that either reach $\Goal$ or end in a
state from which there is a path to $\Goal$ in $\calM^{\opt}$.

For completeness, we further clarify the relevance of the sequences $(p_n^{\opt,-})_n$ and $(p_n^{\opt,+})_n$ with respect to our aim.
Focusing on $(p_n^{\opt,-})_n$, observe that what the scheme computes is an approximation of the limit of this sequence, i.e., $\lim_n \optOp_\sigma \Prob^{\sigma}_{\calM^{\opt},s_0}(\F[\le n] \Goal)$.
Yet, the actual value we want to approximate is $\Prob^{\opt}_{\calM^{\opt},s_0}(\F \Goal)$, which is equal to $\optOp_\sigma \lim_n \Prob^{\sigma}_{\calM^{\opt},s_0}(\F[\le n] \Goal)$.
The convergence of the scheme is a sufficient condition for $\lim_{n\to\infty} p_n^{\opt, -} = \Prob^{\opt}_{\calM^{\opt},s_0}(\F \Goal)$: indeed, given Lemma~\ref{lem:HnRn}, this is the only possible limit value.
Independently of the convergence of the scheme, these two values also always coincide in finitely action-branching MDPs.

\begin{lem}
  Let $\calM = (S,\Act,\proba)$ be a finitely action-branching MDP, $\Goal \in S$ be an absorbing target state, $\opt \in \{\inf,\sup\}$, and $s_0 \in S \cap S^\opt$ be an initial state.
  Then,
  \[
    \lim_{n\to\infty} p_n^{\opt, -} = \Prob^{\opt}_{\calM^{\opt},s_0}(\F \Goal) \ \text{and} \ \lim_{n\to\infty} p_n^{\opt, +} = \Prob^{\opt}_{\calM^{\opt},s_0}(\F \Goal \vee \G (\neg \Goal \wedge \neg \bad^{\opt}))\enspace.
  \]
\end{lem}
\begin{proof}
  It directly follows from Theorem~\ref{thm:limoptEqualsoptlim} and the continuity of probability measures by taking $E_n = \F[\le n] \Goal$ for the first equality, and~$E_n = \F[\le n] \Goal \vee \G[\le n] (\neg \Goal \wedge \neg \bad^\opt)$ for the second equality.
\end{proof}

On finite MDPs, with $\opt = \sup$, $\ApproxScheme_\unfold^{\opt}$ may not even converge.
Consider the three-state MDP $\calM^{\mathtt{R}}$ from Example~\ref{ex:distinguishingAvoidSets}: we have that for every $n\ge 1$, $p_n^{\sup,-} = \frac 1 2$ and $p_n^{\sup,+} = 1$.
Hence, the scheme does not terminate when $\epsilon < \frac{1}{2}$.

\subsection{Sliced MDP and second approximation scheme}
\label{subsec:sliced}
To overcome the aforementioned shortcoming of
$\ApproxScheme_\unfold^{\sup}$, we propose a refined
approximation scheme.
Intuitively, instead of unfolding the MDP up to a fixed depth, as
implicitly done in $\ApproxScheme_\unfold^{\opt}$, we consider slices of
the MDP consisting of the restrictions to all states that are
reachable from $s_0$ within a fixed number of steps.
This ensures convergence on finite MDPs.

Let $\calM = (S,\Act,\proba)$ be an MDP, $s_0\in S$ be an initial state, $\opt \in \{\inf,\sup\}$, and $\calM^{\opt}$ be as defined in Section~\ref{subsec:collapse}.
For every $n \in \nats$, we define the \emph{sliced MDP} $\calM^{\opt}_n$ as the restriction of $\calM^{\opt}$ to states that can
be reached within~$n$~steps from~$s_0$.
This construction is illustrated in Figure~\ref{fig:sliced-mdp}.

For $n \ge 0$, let $\Reachable{n}$ be the set of states reachable from $s_0$ with a positive probability in at most $n$ steps. Formally, $\Reachable{0} = \{s_0\}$ and for $n \ge 0$, \[\Reachable{n+1} = \Reachable{n} \cup \{s' \in S^{\opt} \mid \exists s \in \Reachable{n}, \exists a \in \Act, \proba^{\opt}(s,a,s')>0\}\enspace.\]

For $n \ge 0$, the sliced MDP
$\calM^{\opt}_n = (S_n^{\opt},\Act_n^{\opt},\proba^{\opt}_n)$ is defined
as follows:
\begin{itemize}
\item $S^{\opt}_n = \Reachable{n} \cup \{\Goal, s_\bot^n\}$, where $s_\bot^n$ is a fresh absorbing state;
\item $\Act^{\opt}_n = \{a_\bot\} \cup \bigcup_{s\in\Reachable{n}\cup\{\Goal\}}\En(s)$, where $a_\bot$ is a fresh action only enabled in $s_\bot^n$;
\item for all $s,s' \in \Reachable{n}$, for all $a
  \in \En(s)$, $\proba^{\opt}_n(s,a,s') =
  \proba^{\opt}(s,a,s')$;
\item for all $s \in \Reachable{n}$, and $a
\in \En(s)$, $\proba^{\opt}_n(s,a,s^n_\bot) =
  \sum_{s' \notin \Reachable{n}} \proba^\opt(s,a,s')$;
\item we define $\En(s^n_\bot) = \{a_\bot\}$ and $\proba^{\opt}_n(s^n_\bot,a_\bot,s^n_\bot) = 1$, and $\proba^{\opt}_n(\Goal,a,\Goal) = 1$ for all $a \in \En(\Goal)$.
\end{itemize}

The state spaces of $\calM^\opt$ and $\calM^\opt_{n}$ coincide on
$\Reachable{n}$, and all transitions going out of
$\Reachable{n}$ in $\calM^\opt$ are directed to $s_\bot^n$ in
$\calM^\opt_{n}$.
Any path in $\calM^{\opt}$ induces a unique path in
$\calM^\opt_n$ which either stays in the common state space
$\Reachable{n}$ or reaches $s_\bot^n$. Moreover,
any path in $\calM^\opt_n$ that reaches $\Goal$ corresponds to a path
in $\calM^\opt$ that also reaches $\Goal$. In the sequel, we seamlessly rely on the correspondence between paths in $\calM^{\opt}$ that only
visit states reachable within~$n$~steps, and paths in $\calM^\opt_n$
that avoid $s_\bot^n$.
Observe that the sliced MDPs of finitely branching MDPs are all finite.

\begin{figure}[tbp]
  \centering
  \begin{tikzpicture}
    \path (0,0) node[j0] (s0) {$s_0$};
    \path (1.5,0) node (dots) {$\cdots$};
    \path (3,0) node[j0] (s) {$s$};

    \path (1,0) node (s0r) {};
    \path (2,0) node (sl) {};

    \path (4,1) node (sa1) {};
    \path (4,.5) node (sa2) {};

    \path (4,-1) node (sb1) {};
    \path (2.2,-.9) node (sb2) {};

    \path (4.9,1.1) node (tag) {$\calM^{\opt}$};

    \path (s0) edge[-latex'] (s0r);
    \path (sl) edge[-latex'] (s);

     \path (s) edge[-latex'] node[pos=.75,left=1pt]
     {$\alpha,1 - p$} (sa1);
     \path (s) edge[-latex'] node[pos=.75,below right,xshift=-3pt]
     {$\alpha,p$} (sa2);

     \path (s) edge[-latex'] node[pos=.5,right=2pt,yshift=1pt]
     {$\beta,1 - q$} (sb1);
     \path (s) edge[-latex'] node[pos=.5,left=2pt]
     {$\beta,q$} (sb2);

    \draw [densely dashed](-.5,1.5) -- (5.5,1.5) -- (5.5,-1.5) -- (-.5,-1.5) --
    cycle;
    \draw[dotted] (.5,1.5) -- (.5,-1.5);

    \draw[dotted] (2.6,1.5) -- (2.6,-1.5);
    \draw[dotted] (3.4,1.5) -- (3.4,-1.5);

    \draw[decorate,decoration={brace,mirror,amplitude=4pt}]
      (-.5,-1.65) -- (3.4,-1.65)
      node[midway,below=7pt] {$\Reachable{n}$};

    \path (0,1.75) node {$0$};
    \path (3,1.75) node {$n$};
    \path (1.5,1.75) node {$\cdots$};
  \end{tikzpicture}
  \quad\quad\quad
  \begin{tikzpicture}
    \path (0,0) node[j0] (s0) {$s_0$};
    \path (1.5,0) node (dots) {$\cdots$};
    \path (3,0) node[j0] (s) {$s$};
    \path (6,0) node[j0] (sinkn) {$s^n_\bot$};

    \path (1,0) node (s0r) {};
    \path (2,0) node (sl) {};

    \path (4,1) node (sa1) {};
    \path (4,.5) node (sa2) {};

    \path (4,-1) node (sb1) {};
    \path (2.2,-.9) node (sb2) {};
    \path (6,1.4) node (tag) {$\calM^{\opt}_n$};

    \path (s0) edge[-latex'] (s0r);
    \path (sl) edge[-latex'] (s);

     \path (s) edge[-latex',bend right] node[pos=.58,below]
     {$\beta,1-q$} (sinkn);
     \path (s) edge[-latex',bend left] node[pos=.58,above]
     {$\alpha,1$} (sinkn);
     \path (s) edge[-latex'] node[pos=.5,left=2pt]
     {$\beta,q$} (sb2);

    \draw [densely dashed] (-.5,1.5) -- (3.4,1.5) -- (3.4,-1.5) -- (-.5,-1.5) --
    cycle;
    \draw[dotted] (.5,1.5) -- (.5,-1.5);

    \draw[dotted] (2.6,1.5) -- (2.6,-1.5);
    \draw[dotted] (3.4,1.5) -- (3.4,-1.5);

    \draw[decorate,decoration={brace,mirror,amplitude=4pt}]
      (-.5,-1.65) -- (3.4,-1.65)
      node[midway,below=7pt] {$\Reachable{n}$};
  \end{tikzpicture}
\caption{Construction of the sliced MDP
  $\calM^{\opt}_n$ (right) from $\calM^{\opt}$ (left).}
\label{fig:sliced-mdp}
\end{figure}

We use events on the sliced MDP to find lower and upper bounds on the desired probability $p = \Prob^{\opt}_{\calM^{\opt},s_0}(\F \Goal)$.
A lower bound can be obtained through $\Prob^{\opt}_{\calM^{\opt}_n,s_0}(\F \Goal)$: reaching $\Goal$ in $\calM^{\opt}_n$ (only through states in $\Reachable{n}$) implies reaching $\Goal$ in $\calM^{\opt}$.
An upper bound is given by $\Prob^{\opt}_{\calM^{\opt}_n,s_0}(\F (\Goal \vee s_\bot^n))$: a path that reaches~$\Goal$ in $\calM^{\opt}$ would either reach $\Goal$ or $s_\bot^n$ in $\calM^{\opt}_n$.
We state these bounds formally, along with relations with the sequences~$(p_n^{\opt, -})_n$ and $(p_n^{\opt, +})_n$ from $\ApproxScheme_\unfold^{\opt}$, in the following lemma (proved in Appendix~\ref{app:constructions}).

\begin{restatable}{lem}{slices} \label{lem:slices}
  The sliced MDP $\calM^\opt_n$ satisfies the following inequalities for every $n \in \nats$:
  \begin{enumerate}
  \item $\Prob^{\opt}_{\calM_n^{\opt},s_0}(\F \Goal) \le
    \Prob^{\opt}_{\calM^\opt,s_0}(\F \Goal) \le
    \Prob^{\opt}_{\calM^{\opt}_n,s_0}(\F (\Goal \vee s_\bot^n))$, \label{item:slicesBounds}
  \item $\Prob^{\opt}_{\calM^{\opt}_n,s_0}(\F (\Goal \vee s_\bot^n)) \le
    \Prob^{\opt}_{\calM^{\opt}_n,s_0}(\F \Goal) +
    \Prob^{\sup}_{\calM^{\opt}_n,s_0} (\F s_\bot^n)$, \label{item:optsup}
  \item $p_n^{\opt, -} = \Prob^{\opt}_{\calM^\opt,s_0}(\F[\le n] \Goal)  \le
    \Prob^{\opt}_{\calM_n^{\opt},s_0}(\F \Goal)$, \label{item:pq-}
  \item $\Prob^{\opt}_{\calM^{\opt}_n,s_0}(\F (\Goal \vee
  s_\bot^n)) \le \Prob^{\opt}_{\calM^\opt,s_0} (\F[\le n] \Goal \vee
  \G[\le n] (\neg \Goal \wedge \neg \bad^\opt)) = p_n^{\opt, +}$. \label{item:pq+}
  \end{enumerate}
\end{restatable}

Thanks to Lemma~\ref{lem:slices} (item~\ref{item:slicesBounds}), it is natural to define an
approximation scheme with $\Prob^{\opt}_{\calM_n^{\opt},s_0}(\F \Goal)$ as
a lower bound, and
$\Prob^{\opt}_{\calM^{\opt}_n,s_0}(\F (\Goal \vee s_\bot^n))$ as an upper
bound. It is formalised in Scheme~\ref{approx_scheme2}.

\begin{algorithm}[htb]
  \SetKwInOut{Input}{Input}\SetKwInOut{Output}{Output} \Input{An MDP
  $\calM$, $s_0 \in S$, $\Goal \in S$, $\epsilon \in (0,1)\cap\bbQ$, and $\opt\in\{\inf,\sup\}$.} \Output{Upon termination, a value $v\in[0,1]$ such that
  $v \le \Prob^{\opt}_{\calM,s_0}(\F\Goal) \le v+\epsilon$.} \BlankLine
  $n\gets 0$\;
  \Repeat{$q_n^{\opt,+} - q_n^{\opt,-} \le \epsilon$}{$n\gets n+1$\;
    $q_n^{\opt,-} \gets \Prob^{\opt}_{\calM^{\opt}_n,s_0}(\F \Goal)$\;
    $q_n^{\opt,+} \gets \Prob^{\opt}_{\calM^{\opt}_n,s_0}(\F (\Goal \vee s_\bot^n))$\;}
  \KwRet{$q_n^{\opt,-}$}\;
  \caption{$\ApproxScheme_\sliced^{\opt}$}\label{approx_scheme2}
\end{algorithm}

For finitely branching $\calM$, the finite slice likewise yields an $\epsilon$-optimal scheduler upon termination, which can moreover be chosen pure and positional.
At termination, retain a pure positional scheduler attaining $q_n^{\sup,-}$ for $\sup$, or $q_n^{\inf,+}$ for $\inf$ (Lemma~\ref{lem:finiteMDP}), and extend it by arbitrary positional choices outside $\Reachable n$ in $\calM^\opt$.
Its reachability probability is respectively at least $q_n^{\sup,-}$ or at most $q_n^{\inf,+}$, since reaching $\Goal$ requires reaching it within the slice or first leaving the slice.
The stopping criterion therefore yields a pure positional $\epsilon$-optimal scheduler from~$s_0$, which lifts to $\calM$ with these properties.

Through Lemma~\ref{lem:slices} (items~\ref{item:pq-} and~\ref{item:pq+}), we learn that for
every $n$, $p_n^{\opt,-} \le q_n^{\opt,-}$ and
$q_n^{\opt,+} \le p_n^{\opt,+}$. Thus, $\ApproxScheme_\sliced^{\opt}$ is a
refinement of $\ApproxScheme_\unfold^{\opt}$, stated as follows.

\begin{thm} \label{thm:refinement}
  Let $\calM = (S,\Act,\proba)$ be an MDP, $s_0 \in S$ be an initial
  state, $\Goal \in S$ be a target state, and $\opt\in\{\inf,\sup\}$. Assume that
  $\ApproxScheme_\unfold^{\opt}$ provides a converging approximation scheme
  for $\Prob^{\opt}_{\calM,s_0}(\F \Goal)$. Then so does
  $\ApproxScheme_\sliced^{\opt}$.
\end{thm}
\begin{proof}
  By Lemma~\ref{lem:slices}, for every $n$, $p_n^{\opt,-}\le q_n^{\opt,-}\le q_n^{\opt,+}\le p_n^{\opt,+}$, so convergence of $\ApproxScheme_\unfold^{\opt}$
  immediately implies convergence of $\ApproxScheme_\sliced^{\opt}$.
\end{proof}

We give below a criterion for ensuring
that $\ApproxScheme_\sliced^{\opt}$ is a converging approximation scheme.
It refines Theorem~\ref{thm:genericCorrectnessTermination}, as we have
$\Prob^{\sup}_{\calM^{\opt}_n,s_0}(\F s_\bot^n) \le
\Prob^{\sup}_{\calM^{\opt},s_0}(H_n^\opt)$: indeed, any path in $\calM_n^{\opt}$ that reaches $s_\bot^n$ (which takes at least $n + 1$ steps) corresponds to a set of paths that reach neither $\Goal$ nor $\bad^\opt$ within $n$ steps in $\calM^{\opt}$.

\begin{thm} \label{thm:generic2CorrectnessTermination}
  Let $\calM = (S,\Act,\proba)$ be an MDP, $s_0 \in S$ be an initial
  state, $\Goal \in S$ be a target state, and $\opt\in\{\inf,\sup\}$. Assume that
  $\lim_{n \to \infty} \Prob^{\sup}_{\calM_n^{\opt},s_0}(\F s_\bot^n)
  = 0$. Then, $\ApproxScheme_\sliced^{\opt}$ provides a converging
  approximation scheme for $\Prob^{\opt}_{\calM,s_0}(\F \Goal)$.
 \end{thm}
 \begin{proof}
  The sequences $(q_n^{\opt,-})_n$ and $(q_n^{\opt,+})_n$ are respectively
  non-decreasing and non-increasing.
  Indeed, restricting or extending schedulers between consecutive slices, reaching~$\Goal$ before leaving $\Reachable{n}$ implies reaching it before leaving $\Reachable{n+1}$, whereas reaching~$\Goal$ or leaving $\Reachable{n+1}$ implies reaching~$\Goal$ or leaving $\Reachable{n}$.
  Taking the corresponding infima or suprema yields both monotonicities.

  When $(q_n^{\opt,-})_n$ and $(q_n^{\opt,+})_n$ converge to the
  same limit, $\ApproxScheme_\sliced^{\opt}$ terminates for all $\epsilon > 0$. Moreover, thanks to
  Lemma~\ref{lem:slices} (item~\ref{item:slicesBounds}), for every $n \in \nats$,
  $q_n^{\opt,-} \le \Prob^{\opt}_{\calM^\opt,s_0}(\F \Goal) \le
  q_n^{\opt,+} $ so that upon termination, $\ApproxScheme_\sliced^{\opt}$
  returns an $\epsilon$-approximation of
  $\Prob^{\opt}_{\calM^{\opt},s_0}(\F \Goal)$.
  This is an $\epsilon$-approximation of $\Prob^{\opt}_{\calM,s_0}(\F \Goal)$ by Lemma~\ref{lem:Mopt}.

  Under the assumption that
  $\lim_{n \to \infty}\Prob^{\sup}_{\calM^{\opt}_n,s_0} (\F
  s_\bot^n) = 0$, item~\ref{item:optsup} of Lemma~\ref{lem:slices} implies
  that
  $\lim_{n \to \infty} \Prob^{\opt}_{\calM_n^{\opt},s_0}(\F (\Goal
  \vee s_\bot^n)) = \lim_{n \to \infty}
  \Prob^{\opt}_{\calM_n^{\opt},s_0}(\F \Goal)$. We obtain that the two
  sequences $(q_n^{\opt,-} )_n$ and $(q_n^{\opt,+} )_n$ converge
  towards $\Prob^{\opt}_{\calM^\opt,s_0}(\F \Goal)$, and
  $\ApproxScheme_\sliced^{\opt}$ converges.
 \end{proof}

\begin{rem}[The case of finite MDPs]
   $\ApproxScheme_\sliced^{\opt}$ converges for
  finite MDPs (and stops at the latest after a number of iterations
  equal to the number of reachable states). In contrast, recall that
  approximating the supremum probability of a reachability objective with
  $\ApproxScheme_\unfold^{\sup}$ may not converge on some finite MDPs (see
  Example~\ref{ex:distinguishingDecisiveness} and Proposition~\ref{prop:endComponents}).
\end{rem}

\section{When do these schemes converge?} \label{sec:convergingSchemes}


In this section, we give criteria related to decisiveness
that ensure the convergence of the approximation schemes. We
start with criteria ensuring the convergence of
$\ApproxScheme_\unfold^{\opt}$.
We then show that, for finitely branching MDPs, convergence of $\ApproxScheme_\sliced^{\inf}$
implies convergence of $\ApproxScheme_\unfold^{\inf}$. Finally, we give conditions on the MDPs that ensure
the convergence of $\ApproxScheme_\sliced^{\sup}$ (but not necessarily $\ApproxScheme_\unfold^{\sup}$).
Omitted proofs for this section are deferred to
Appendix~\ref{app:convergence}.

\subsection{Convergence of \texorpdfstring{$\ApproxScheme_\unfold^{\opt}$}{ApproxSchemeUnfold}} \label{subsec:scheme1}

We give conditions related to decisiveness which ensure the
convergence of the approximation schemes (recall that the convergence
of $\ApproxScheme_\unfold^{\opt}$ implies the convergence of
$\ApproxScheme_\sliced^{\opt}$ by Theorem~\ref{thm:refinement}).

\begin{thm} \label{thm:approx-opt}
  Let $\calM = (S,\Act,\proba)$ be an MDP, $s_0 \in S$ be an initial state,
  and~$\Goal$ be a target state. Let $\opt \in \{\inf,\sup\}$. Assume
  that $\calM$ is finitely action-branching and $\opt$-decisive
  w.r.t.\ $\Goal$ from $s_0$. Then $\ApproxScheme_\unfold^{\opt}$ converges on $\calM$ from~$s_0$.
\end{thm}

To highlight the role of the decisiveness hypotheses in Theorem~\ref{thm:approx-opt}, we show on some examples that without decisiveness, the approximation schemes may not converge.
We first show that for some non-$\inf$-decisive MDPs, the approximation schemes do not converge. Indeed, consider
the MDP~$\calM$ from Figure~\ref{fig:ex-mdp}, which has
$\Avoid^{\inf}_{\calM}(\Goal) = \{\bad\}$, so that
$\calM^{\inf}$ is exactly the same as $\calM$ (up to renaming the $\bad$ state). If $p > \frac 1 2$, $\calM$ is not
$\inf$-decisive from $s_1$ w.r.t.\ $\Goal$ since the pure and positional
scheduler $\sigma$ that always picks action $\alpha$ has a positive
probability, say $\lambda_p$, to reach neither~$\Goal$ nor~$\bad$. For every $n \ge 1$,
\[p_n^{\inf, +} = \Prob^{\inf}_{\calM, s_1}(\F[\le n] \Goal \vee \G[\le n] (\neg \Goal \wedge \neg \bad)) = 1 - \Prob^{\sup}_{\calM, s_1}(\F[\le n]\bad) = q\enspace.\]
This is achieved by choosing $\beta$ straight away; any other scheduler runs the risk of reaching $\Goal$.
On the other hand, $p_n^{\inf, -} \le \Prob^{\inf}_{\calM, s_1}(\F \Goal) \le 1 - \lambda_p$ (which is the value obtained by the scheduler always choosing $\alpha$).
Hence, by picking~$q$ and~$p$ such that $q > 1 - \lambda_p$, $\ApproxScheme_\unfold^{\inf}$ does not converge on $\calM$ from $s_1$.

Similar arguments show that $\ApproxScheme_\sliced^{\inf}$ does not converge on $\calM$ from~$s_1$.
First, $q_n^{\inf, +} = \Prob^{\inf}_{\calM^{\inf}_n, s_1}(\F (\Goal \vee s_\bot^n)) = q$ --- this is achieved by choosing~$\beta$ straight away; any other scheduler runs the risk of reaching $\Goal$ or $s_\bot^n$.
Second, $q_n^{\inf, -} \le \Prob^{\inf}_{\calM, s_1}(\F \Goal) \le 1 - \lambda_p$.  Thus, $\ApproxScheme_\sliced^{\inf}$ does not converge either if $q > 1 - \lambda_p$.

Observe also that $\ApproxScheme_\sliced^{\sup}$ (and thus $\ApproxScheme_\unfold^{\sup}$) do not converge from $s_1$ on the MDP $\calM$ in Figure~\ref{fig:ex-mdp}.
This MDP is not $\sup$-decisive (as it is not $\inf$-decisive) w.r.t.\ $\Goal$ from $s_1$.
We have that for every $n \in \nats$, $q_n^{\sup,+} = 1$ (achieved by only choosing $\alpha$), and yet $\Prob^{\sup}_{\calM,s_1}(\F \Goal) < 1$.

Finally, the finitely action-branching hypothesis is also critical.
Recall that for the infinitely action-branching MDP in Example~\ref{ex:schemeUseless}, $\ApproxScheme_\unfold^{\inf}$ does not converge from $s_0$.
Yet, this MDP is $\inf$-decisive w.r.t.\ $\Goal$ from $s_0$.

Despite the differences between the two approximation schemes, we have that for finitely branching MDPs, the convergence of $\ApproxScheme_\sliced^{\inf}$ implies the convergence of $\ApproxScheme_\unfold^{\inf}$.

\begin{restatable}{thm}{approxSchemeTwoImpliesOne} \label{thm:approxSchemeTwoImpliesOne}
  Let $\calM = (S,\Act,\proba)$ be a finitely branching MDP, $s_0 \in S$ be an initial state,
  and~$\Goal$ be a target state. If $\ApproxScheme_\sliced^{\inf}$ converges on $\calM$ from~$s_0$, then $\ApproxScheme_\unfold^{\inf}$ converges on $\calM$ from
  $s_0$.
\end{restatable}

\subsection{Convergence of \texorpdfstring{$\ApproxScheme_\sliced^{\opt}$}{ApproxSchemeSliced}} \label{subsec:nonFleeing}

By applying the result and the discussion of
Section~\ref{subsec:scheme1}, we already know that
$\ApproxScheme_\sliced^{\opt}$ converges under the conditions of
Theorem~\ref{thm:approx-opt}, that is, when the MDP is finitely
action-branching and $\opt$-decisive. The
$\sup$-decisiveness property is rather restrictive and is not
satisfied by all finite MDPs, while $\ApproxScheme_\sliced^{\sup}$ obviously
converges on finite MDPs. We therefore propose alternative
conditions that ensure the convergence of
$\ApproxScheme_\sliced^{\sup}$.

\begin{defi} \label{def:nonFleeing}
  Let $\calM = (S,\Act,\proba)$ be an MDP, $\Goal \in S$ be a target
  state, and $s_0 \in S \cap S^{\sup}$ be an (uncollapsed) initial state. The MDP $\calM$ is
  \emph{non-fleeing w.r.t.\ $\Goal$ from $s_0$} whenever for every
  $\sigma \in \SchedPurePos(\calM^{\sup})$,
  \[
    \Prob^\sigma_{\calM^{\sup},s_0} \left(\Div\right) =0
  \]
  where $\Div$ is the event
  $\bigcap_{n \in \nats} \left(\F \left(\Reachable{n+1} \setminus
      \Reachable{n}\right)\right)$.
\end{defi}

We explain the intuition behind this notion through its negation:
being \emph{fleeing} corresponds to the existence of a scheduler under which, with positive probability, arbitrarily deep reachable layers are visited (in particular, the target is never reached) in such a way that, at any point during the execution, there exists a
deviating scheduler that could reach the target (otherwise, it would have
hit the avoid set summarised as $\bad^{\sup}$ in $\calM^{\sup}$).

\begin{thm} \label{thm:nonFleeing}
  Let $\calM = (S,\Act,\proba)$ be a finitely branching MDP, $s_0 \in S\cap S^{\sup}$ be an initial state,
  and~$\Goal$ be a target state. Assume that $\calM$ is non-fleeing w.r.t.\ $\Goal$ from $s_0$.
  Then, $\ApproxScheme_\sliced^{\sup}$ converges on $\calM$ from $s_0$.
\end{thm}

The applicability of the convergence condition in this theorem is incomparable with the one in Theorem~\ref{thm:approx-opt}: on the one hand, $\sup$-decisiveness implies non-fleeingness, so non-fleeingness is less restrictive; on the other hand, this theorem deals with finitely branching MDPs, as opposed to the more general finitely action-branching MDPs of Theorem~\ref{thm:approx-opt}.
By Theorem~\ref{thm:generic2CorrectnessTermination}, it suffices to establish the following lemma (proof in Appendix~\ref{app:convergence}).

\begin{restatable}{lem}{lemmaNonFleeing} \label{lem:lemmaNonFleeing}
  Let $\calM = (S,\Act,\proba)$ be a finitely branching MDP, $s_0 \in S\cap S^{\sup}$ be an initial state,
  and~$\Goal$ be a target state.
  If $\calM$ is non-fleeing w.r.t.\ $\Goal$ from $s_0$, then
  \[
    \lim_{n \to \infty} \Prob^{\sup}_{\calM_n^{\sup},s_0} \left(\F
      s_\bot^n\right) = 0 \enspace.
  \]
\end{restatable}

To emphasise that non-fleeingness is weaker than $\sup$-decisiveness, we characterise what is still missing for $\sup$-decisiveness in the case of finitely branching MDPs: all reachable finite end components must be ``trivial'', that is, either they consist of the target state only, or they are included in its avoid set.
This generalises Proposition~\ref{prop:endComponents}, which dealt with finite MDPs, for a subclass of finitely branching MDPs.

\begin{lem}\label{lem:finiteEndComponents}
  Let $\calM = (S, \Act, \proba)$ be a finitely branching MDP, $s_0 \in S$ be an initial state, and $\Goal \in S$ be a target state.
  Assume $\calM$ is non-fleeing w.r.t.\ $\Goal$ from $s_0$.
  Then, $\calM$ is $\sup$-decisive w.r.t.\ $\Goal$ from $s_0$ if and only if for all finite end components $(R, A)$ of $\calM$ reachable from $s_0$, either $R = \{\Goal\}$ or $R \subseteq \Avoid_\calM^{\sup}(\Goal)$.
\end{lem}
\begin{proof}
  Assume first that $\calM$ is $\sup$-decisive w.r.t.\ $\Goal$.
  Let $(R, A)$ be a finite end component of~$\calM$ reachable from $s_0$ such that $R \neq \{\Goal\}$.
  If $R \nsubseteq \Avoid_\calM^{\sup}(\Goal)$, then no state of $R$ can belong to $\Avoid_\calM^{\sup}(\Goal)$: if there were a scheduler that could reach $\Goal$ from a state $r \in R$, then we could reach $\Goal$ with positive probability from any state of $R$, by first reaching $r$ from within the end component using only actions in $A$.
  Standard results on end components imply that there exists a pure and positional scheduler $\sigma$ and a state $s \in R$ such that, from~$s$, the scheduler~$\sigma$ only visits states in~$R$ forever~\cite[Section~3.3]{deAlfaroThesis}.
  Hence,
  \[
    \Prob^\sigma_{\calM,s}\bigl(\F \Goal \vee \F \Avoid_\calM^{\sup}(\Goal)\bigr) = 0 \enspace.
  \]
  This contradicts the $\sup$-decisiveness of $\calM$ w.r.t.\ $\Goal$ from $s_0$: there is a pure and positional scheduler with a positive probability to reach $s$ from $s_0$, and that then mimics $\sigma$ from $s$, which ensures that neither $\Goal$ nor $\Avoid_\calM^{\sup}(\Goal)$ is reached with positive probability.
  We conclude that every finite end component $(R, A)$ reachable from $s_0$ is such that either $R = \{\Goal\}$ or $R \subseteq \Avoid_\calM^{\sup}(\Goal)$.

  Assume now that every finite end component $(R, A)$ of $\calM$ reachable from $s_0$ is such that either $R = \{\Goal\}$ or $R \subseteq \Avoid_\calM^{\sup}(\Goal)$.
  We prove that $\calM$ is $\sup$-decisive from $s_0$.
  We assume $s_0 \in S^{\sup}$, otherwise the result is trivial.
  Fix a scheduler $\sigma \in \SchedPurePos(\calM^{\sup})$.

  Assume towards a contradiction that
  \[
    \Prob^\sigma_{\calM^{\sup},s_0}\bigl(\G (\neg \Goal \wedge \neg \bad^{\sup})\bigr) > 0 \enspace.
  \]
  Since $\Prob^\sigma_{\calM^{\sup},s_0}(\Div) = 0$ by the definition of non-fleeingness (Definition~\ref{def:nonFleeing}), we obtain
  \[
    \Prob^\sigma_{\calM^{\sup},s_0}\bigl(\G (\neg \Goal \wedge \neg \bad^{\sup}) \wedge \neg \Div\bigr) > 0 \enspace.
  \]
  Now, $\neg \Div = \bigcup_{n \in \nats} \G \Reachable{n}$, so there exists $n \in \nats$ such that
  \[
    \Prob^\sigma_{\calM^{\sup},s_0}\bigl(\G (\neg \Goal \wedge \neg \bad^{\sup}) \wedge \G \Reachable{n}\bigr) > 0 \enspace.
  \]
  As $\calM$ is finitely branching, the set $\Reachable{n}$ is finite.
  Hence, under the Markov chain induced by $\sigma$ on~$\Reachable{n}$ (redirecting transitions escaping $\Reachable{n}$ to an absorbing state), standard results on finite Markov chains imply that with positive probability, a path satisfying the event above eventually remains in a reachable bottom strongly connected component~$C$ contained in $\Reachable{n} \setminus \{\Goal,\bad^{\sup}\}$.

  Let $A_C(t) = \{\sigma(t)\}$ for every $t \in C$. Then $(C, A_C)$ is a reachable finite end component of~$\calM^{\sup}$.
  Since $C$ does not contain $\bad^{\sup}$, and the transitions between states outside $\Avoid_\calM^{\sup}(\Goal)$ are unchanged by the collapse into $\calM^{\sup}$, the same pair $(C, A_C)$ is also a finite end component of~$\calM$.
  Moreover, $C \neq \{\Goal\}$ and $C \cap \Avoid_\calM^{\sup}(\Goal)=\emptyset$, contradicting the assumption on finite end components.

  Therefore, for every scheduler $\sigma \in \SchedPurePos(\calM^{\sup})$,
  \[
    \Prob^\sigma_{\calM^{\sup},s_0}\bigl(\F \Goal \vee \F \bad^{\sup}\bigr) = 1 \enspace.
  \]
  Since $\bad^{\sup}$ represents the states of $\Avoid_\calM^{\sup}(\Goal)$, this exactly means that $\calM$ is $\sup$-decisive w.r.t.\ $\Goal$.
\end{proof}

\section{Applications} \label{sec:applications}


We discuss the instantiation of the above approximation schemes into
approximation \emph{algorithms} for two concrete classes of systems:
\emph{nondeterministic probabilistic lossy channel systems}
(NPLCSs) and \emph{partially observable MDPs} (POMDPs).
Although these models have distinct sources of randomness
and infiniteness, they both
induce countable MDPs, where states are ``configurations'' of
the system (control states and channel contents for NPLCSs, rational
beliefs for POMDPs).
In each case, we show that the induced MDPs (or
small modifications thereof) satisfy a kind of decisiveness, which
allows us to use approximation schemes. We then show how to effectively
compute approximations of the infimum reachability
probabilities for both classes, and supremum reachability probabilities for NPLCSs under a decidable sufficient condition.

Known results about inf/sup reachability problems in these models are summarised in Table~\ref{tab:results}, with the new results highlighted in bold.

\begin{table}
  \centering
  \renewcommand{\arraystretch}{1.25}
  \begin{tabular}{@{}ccc@{}}
    \toprule
    & NPLCSs & POMDPs\\
    \midrule
    $\Prob^{\inf}(\F \Goal) = 0$ & Decidable (Theorem~\ref{thm:decNPLCS}) & Decidable~\cite[Corollary~1]{RevelationsArXiv}\\
    $\Prob^{\sup}(\F \Goal) = 1$ & \textbf{Undecidable} (Theorem~\ref{thm:undec}) & Undecidable~\cite{GO10}\\
    \addlinespace
    $\Prob^{\inf}(\F \Goal)$ & \textbf{Approximable} (Theorem~\ref{thm:nplcs-inf-decid}) & \textbf{Approximable} (Theorem~\ref{thm:pomdpSafetyDecidable})\\
    $\Prob^{\sup}(\F \Goal)$ & {\renewcommand{\arraystretch}{0.88}\begin{tabular}{@{}c@{}}Open in general\\\textbf{Approximable if non-fleeing}\\(Proposition~\ref{prop:nplcs-sup-nonfleeing})\end{tabular}} & {\renewcommand{\arraystretch}{0.88}\begin{tabular}{@{}c@{}}Not approximable,\\
    even for fixed $\epsilon=\frac{1}{3}$~\cite{MHC03}\end{tabular}}\\
    \bottomrule
  \end{tabular}
  \caption{Known results about the decidability of qualitative and quantitative reachability problems in NPLCSs and POMDPs. Approximability refers to \emph{additive} approximations. This table recalls and extends Table~\ref{tab:results-intro} from the introduction. New results are highlighted in bold. To our knowledge, the result for the top-left cell was not stated explicitly in the literature, but follows easily from existing results~\cite{BBS-acmtocl07}, which is why it is not bolded. To the best of our knowledge, whether $\Prob^{\sup}(\F \Goal)$ can be approximated in unrestricted NPLCSs is open; Proposition~\ref{prop:nplcs-sup-nonfleeing} gives a decidable sufficient condition under which it is approximable.}
  \label{tab:results}
\end{table}
The dependencies behind the three positive approximation results are summarised in
Figures~\ref{fig:nplcs-inf-dependencies},~\ref{fig:nplcs-sup-dependencies},
and~\ref{fig:pomdp-inf-dependencies}.

\subsection{Lossy channel systems}
\label{subsec:nplcs}
In our first application, we consider the case where
MDPs are induced by a probabilistic variant of \emph{lossy channel
systems}. \emph{Nondeterministic and probabilistic lossy channel
systems} build on channel systems, incorporating probabilistic
message losses and allowing nondeterministic choices between possible
read/write actions~\cite{BBS-acmtocl07}.

\paragraph{Lossy channel systems and induced MDP semantics.}
A \emph{channel system} is a tuple $\calS= (Q,\calC,\Mess,\ActLCS,\Delta)$
consisting of a finite set $Q$ of \emph{control states}, a finite set
$\calC$ of \emph{channels}, a finite \emph{message alphabet} $\Mess$,
a finite set $\ActLCS$ of \emph{silent action labels}, and a finite
set $\Delta$ of \emph{transition rules}. Each transition rule has the
form $q \step{\op} q'$ where $\op$ is an \emph{operation} of one of the following three kinds:
\begin{itemize}
\item \emph{send operation}: $c!m$ for $c \in \calC$ and $m \in \Mess$, representing the
  sending of message $m$ along channel~$c$,
\item \emph{receive operation}: $c?m$ for $c \in \calC$ and $m \in \Mess$, representing the
  reception of message $m$ from channel~$c$, or
\item $\ell \in \ActLCS$, representing an internal action labeled
  with $\ell$ with no corresponding send/receive operation.
\end{itemize}
Messages are stored in FIFO queues, and the contents of the queues are
naturally represented by finite words over $\Mess$. A
\emph{configuration} of a channel system $\calS$ is a pair
$(q,\ccont) \in Q \times (\Mess^*)^\calC$ consisting of a control state
and a family of words describing each channel's contents. A transition rule
$\delta = (q,\op,q')$ is \emph{enabled} in a configuration $(p,\ccont)$
if $p=q$ and one of the following conditions applies:
$\op = c?m$ and $\ccont(c)=m v$ with $v\in \Mess^*$, or $\op = c!m$, or
$\op \in \ActLCS$. If so, firing $\delta$ from $(q,\ccont)$ yields the
configuration $\delta(q,\ccont) = (q',\ccont')$ where,
\begin{itemize}
  \item if $\op = c?m$, then $\ccont'(c) = v$ and for every $c' \neq c$, $\ccont'(c') = \ccont(c')$ (message $m$ is read from channel $c$),
  \item if $\op = c!m$ then $\ccont'(c) = \ccont(c) m$ and for every $c'\neq c$, $\ccont'(c') = \ccont(c')$ (message $m$ is written to channel~$c$), and
  \item if $\op=\ell \in \ActLCS$ then for every $c \in \calC$, $\ccont'(c) = \ccont(c)$ (no operation is performed on the channel contents).
\end{itemize}
We write $|\ccont|$ for the total number of messages in the channels, i.e., $|\ccont| = \sum_{c \in \calC} |\ccont(c)|$.

A \emph{nondeterministic probabilistic lossy channel system}
(NPLCS) is a pair $\calN = (\calS,\lambda)$ consisting of a channel
system $\calS$ and a loss rate $\lambda \in (0,1)\cap\bbQ$. Its semantics is
the MDP~$\calM[\calN] = (S,\Act,\proba)$ where
\begin{itemize}
\item $S= Q \times (\Mess^*)^\calC$: states are configurations of
  $\calS$;
\item $\Act = \Delta \cup \{\bot\}$: actions are transition rules of $\calS$ and a dummy action $\bot$ that is enabled in deadlocked configurations;
\item the probabilistic transition function $\proba$ is defined as follows
  \[
    \proba((q,\ccont),\delta,(q',\ccont')) =
    \begin{cases}
      \lambda^{|\ccv| - |\ccont'|} (1-\lambda)^{|\ccont'|} \binom{\ccv}{\ccont'} & \text{ if } \delta(q,\ccont) = (q',\ccv)\\
0  & \text{ in all other cases. }
\end{cases}
    \]
    where the combinatorial coefficient $\binom{\ccv}{\ccont'} = \prod_{c\in\calC}\binom{\ccv(c)}{\ccont'(c)}$ is the
    number of different embeddings of $\ccont'$ in $\ccv$. When writing
    $\delta(q,\ccont) = (q',\ccv)$, we implicitly assume that $\delta$ is
    enabled in~$(q,\ccont)$.
\end{itemize}
Note that there may be ``deadlocked'' configurations in $\calM[\calN]$, i.e., configurations $(q,\ccont)$
where no transition rule is enabled. In this case, we use the dummy enabled action $\bot$ which leads with probability $1$ to an absorbing control state where only the dummy action $\bot$ is enabled.
We assume the target state $\Goal$ to be a control state of $\calS$, and we assume that all configurations $(\Goal,\ccont)$ for some channel contents $\ccont$ are collapsed into a state $\Goal$ of $\calM[\calN]$.
Observe that the induced MDP $\calM[\calN]$ is countable and finitely branching.

So defined, actions available from a configuration of an NPLCS
correspond to transition rules that are enabled in the underlying
channel system, and the successor configuration is obtained in two
steps: first the rule is applied (possibly modifying the channel
contents from $\ccont$ to $\ccv$), and second each message is lost
independently with probability $\lambda$ (and kept with probability
$(1-\lambda)$) so that the resulting channel-content vector is $\ccont'$.

In the sequel, when $\calN = (\calS, \lambda)$ is clear from the context,
we may simply write $\calM$ for the MDP $\calM[\calN]$.
We also frequently write ``initial control state $q_0$ of $\calM$'' to refer to the initial configuration $(q_0,\emptyConfig)$.

Figure~\ref{fig:ex-cs} represents a simple example of a channel system with a single channel (whose name is therefore omitted from action labels).
Figure~\ref{fig:nplcs-mdp} shows an excerpt of the MDP induced by this channel system with a loss rate $\lambda = 0.2$.
Because of
the FIFO policy, the control state~$\Goal$ can only be reached from
the initial configuration $(q,\emptyConfig)$ if messages are lost, for
instance along this execution where a message is lost in the first
step:
\[ (q,\emptyConfig) \xrightarrow{!b} (q',\emptyConfig) \xrightarrow{!a}
  (q',a) \xrightarrow{!a} (q',aa) \xrightarrow{?a} (q,a) \xrightarrow{!b} (q',ab) \xrightarrow{?a} (q,b) \xrightarrow{?b}
  (\Goal,\emptyConfig)\enspace. \]

\begin{figure}[htbp]
  \centering
  \begin{tikzpicture}
    \path (0,0) node[j0] (goal) {$\Goal$};
    \path (3,0) node[j0] (p) {$q$};
    \path (6,0) node[j0] (q) {$q'$};

    \path (p) edge[-latex'] node[pos=.5,above]
    {$?b$} (goal);

    \path (p) edge[-latex',bend left] node[pos=.5,above]
    {$!b$} (q);

    \path (q) edge[-latex',bend left] node[pos=.5,below]
    {$?a$} (p);

    \path (q) edge[-latex',loop right,looseness=15] node[pos=.5,right]
    {$!a$} (q);
  \end{tikzpicture}
  \caption{A simple example of a channel system (with a single FIFO
    channel).}
  \label{fig:ex-cs}
\end{figure}

\begin{figure}
    \centering
    \begin{tikzpicture}
      \draw ($(1,0)$) node[j0] (q0) {$q,\emptyConfig$};

      \draw ($(q0)+(3,0)$) node[j0,align=left] (q1) {$q',\emptyConfig$};
      \draw ($(q1)+(3.5,1.5)$) node[j0,align=left] (q2) {$q',b$};
      \draw ($(q1)+(4.5,0)$) node[j0,fill=gray!30] (q3) {$q',ba$};
      \draw ($(q3)+(0,-2)$) node[j0,fill=gray!30] (q4) {$q',aa$};
      \draw ($(q0)+(3,-2)$) node[j0] (q5) {$q',a$};

      \draw (q0) edge[-latex'] node[above] {$!b, 0.2$} (q1);
      \draw (q0) edge[-latex',bend left] node[above] {$!b, 0.8$} (q2);

      \draw[-latex'] (q1)  .. controls +(-30:1cm) and
      +(30:1cm) .. (q1) node[midway, right] {$!a, 0.2$};
      \draw (q1) edge[-latex',bend right] node[midway,left] {$!a, 0.8$} (q5);

      \draw (q5) edge[-latex',bend right] node[midway,right] {$!a, 0.04$} (q1);
      \draw (q5) edge[-latex'] node[below] {$!a, 0.64$} (q4);
      \draw[-latex'] (q5)  .. controls +(-60:1cm) and
      +(-120:1cm) .. (q5) node[midway, below] {$!a, 0.32$};

      \draw (q2) edge[-latex',bend right] node[midway,left=3pt,xshift=-1pt] {$!a, 0.04$} (q1);
      \draw (q2) edge[-latex',bend left] node[midway,right] {$!a, 0.64$} (q3);
      \draw[-latex'] (q2)  .. controls +(60:1cm) and
      +(120:1cm) .. (q2) node[midway, above] {$!a, 0.16$};
      \draw (q2) edge[-latex',bend left] node[midway,right,yshift=-2pt] {$!a, 0.16$} (q5);
    \end{tikzpicture}
    \caption{An excerpt of the MDP $\calM[\calN]$ induced by the NPLCS
      $\calN$ from Figure~\ref{fig:ex-cs} with $\lambda = 0.2$:
      actions and states beyond the gray configurations are omitted.}
    \label{fig:nplcs-mdp}
\end{figure}

We first state decidability and undecidability of comparing optimal reachability
probabilities with the qualitative thresholds $0$ and $1$.
Then, we use the inf-decisiveness property to show that the infimum reachability probability can be approximated in NPLCSs.

\paragraph{Qualitative problems.} We start with qualitative
reachability in NPLCSs.
\begin{thm} \label{thm:decNPLCS}
  Let $\calN$ be an NPLCS with initial control state~$\initState$ and a target control state~$\Goal \in Q$, and let~$\calM$ be the MDP induced by~$\calN$.
  The following problems are decidable:
  \begin{enumerate}
  \item $\Prob^{\inf}_{\calM, \initState}(\F \Goal) = 1$, \label{item:inf1}
  \item $\Prob^{\sup}_{\calM, \initState}(\F \Goal) = 0$, \label{item:sup0}
  \item $\Prob^{\inf}_{\calM, \initState}(\F \Goal) = 0$. \label{item:inf0}
  \end{enumerate}
\end{thm}
\begin{proof}
Let us prove each case separately.

For item~\eqref{item:inf1}, we have
\begin{align*}
  \Prob^{\inf}_{\calM, \initState}(\F \Goal) = 1 \quad
        \iff \quad \neg \bigl(\exists
        \sigma:\ \Prob^\sigma_{\calM, \initState}(\F \Goal) <1 \bigr)\enspace.
\end{align*}

For item~\eqref{item:sup0}, we have
\begin{align*}
  \Prob^{\sup}_{\calM, \initState}(\F \Goal) = 0 \quad  & \iff \quad \neg \bigl(\exists \sigma:\ \Prob^\sigma_{\calM, \initState}(\F \Goal) >0 \bigr)\enspace.
\end{align*}

For item~\eqref{item:inf0}, we have
\begin{align*}
  \Prob^{\inf}_{\calM, \initState}(\F \Goal) = 0 \quad
        \iff \quad \exists \sigma:\ \Prob^\sigma_{\calM, \initState}(\F \Goal) = 0\enspace.
\end{align*}
Only this last equivalence is not immediate.
The right-to-left implication is clear; conversely, the induced MDP is finitely action-branching, so Lemma~\ref{lem:schedulerComplexity}, item~\ref{item:infOptPurePos}, yields a pure positional scheduler attaining the infimum.

In each case, the last statement is decidable (even when $\Goal$ is a
\emph{regular set of configurations}~\cite{Bertrand-PhD06}, and in
particular if $\Goal \subseteq Q$ is a set of control states with
no constraints on channel contents~\cite{BBS-acmtocl07}).
\end{proof}

Yet, we establish the undecidability of the \emph{value-$1$} problem
for NPLCSs, which also contrasts with the fact that the existence of a
scheduler ensuring a reachability objective almost surely is decidable
for NPLCSs~\cite{BBS-acmtocl07}.
Since this result is not directly related to decisiveness, we defer its proof to Appendix~\ref{app:nplcs}.
\begin{restatable}{thm}{undec} \label{thm:undec}
  Given an NPLCS $\calN$, deciding whether $\Prob^{\sup}_{\calM[\calN], \initState}(\F \Goal) = 1$ for a control state~$\Goal$ and initial control state $\initState$ is undecidable.
\end{restatable}

\paragraph{Approximation of the infimum reachability probability.}
Iyer and Narasimha provided an approximation scheme for reachability
probabilities in probabilistic channel systems, whose semantics is
given by a countable Markov chain~\cite{IN97}. This result was then
generalised to all decisive Markov chains by Abdulla, Ben Henda, and
Mayr~\cite{ABM07}. Here we show that, as far as infimum reachability
probabilities are concerned, our approximation schemes can be used for
MDPs induced by NPLCSs, thus lifting the earlier result of~\cite{IN97}
from Markov chains to MDPs.
See Figure~\ref{fig:nplcs-inf-dependencies} for the chain of arguments supporting this result.

The key to proving the feasibility of approximating infimum reachability
probabilities for NPLCSs is their finite attractor property:
\begin{lem}[\cite{BBS-acmtocl07}, Proposition~4.2] \label{lem:finite-attractor}
  Let $\calN$ be an NPLCS. The set of configurations with empty channels is a finite attractor for $\calM[\calN]$.
\end{lem}
We deduce from Proposition~\ref{prop:finiteAttractors} that $\calM[\calN]$ is
univ-decisive, thus $\inf$-decisive, from any configuration w.r.t.\
any control state. Since $\calM[\calN]$ is finitely action-branching, the approximation schemes $\ApproxScheme_\unfold^{\inf}$
and $\ApproxScheme_\sliced^{\inf}$ converge and are correct by
Theorems~\ref{thm:refinement} and~\ref{thm:approx-opt}.

\begin{figure}[tb]
  \centering
  \begin{tikzpicture}[
    dependency/.style={draw, rounded corners=2pt, align=center, inner sep=3pt,
      minimum height=1.05cm},
    outcome/.style={dependency, fill=gray!15},
    implication/.style={-latex', semithick},
    every node/.style={font=\footnotesize}
  ]
    \node[dependency, text width=3.2cm] (general) at (0,0.8)
      {Finite attractor $\Longrightarrow$ univ-decisive $\Longrightarrow$
       $\inf$-decisive\\(Proposition~\ref{prop:finiteAttractors})};
    \node[dependency, text width=3.2cm] (attractor) at (0,-0.8)
      {NPLCSs have a finite attractor\\(Lemma~\ref{lem:finite-attractor})};
    \node[dependency, text width=2.5cm] (decisive) at (3.7,0)
      {The induced MDP is $\inf$-decisive and finitely
       action-branching};
    \node[dependency, text width=2.8cm] (scheme) at (7.2,0)
      {$\ApproxScheme_\unfold^{\inf}$ converges\\(Theorem~\ref{thm:approx-opt})};
    \node[outcome, text width=3.1cm] (algorithm) at (10.9,0)
      {Effective\\approximation of $\Prob^{\inf}(\F\Goal)$\\
       (Theorem~\ref{thm:nplcs-inf-decid})};

    \draw[implication] (general.east) -- (decisive);
    \draw[implication] (attractor.east) -- (decisive);
    \draw[implication] (decisive.east) -- (scheme.west);
    \draw[implication] (scheme.east) -- (algorithm.west);
  \end{tikzpicture}
  \caption{High-level dependencies for approximating infimum reachability
    probabilities in NPLCSs.}
  \label{fig:nplcs-inf-dependencies}
\end{figure}

\begin{thm}\label{thm:nplcs-inf-decid}
  There exists an algorithm that, given an NPLCS $\calN$, an initial
  state $\initState$, a target control state $\Goal \in Q$, and a rational
  number $\epsilon > 0$, returns a value $\epsilon$-close to
  $\Prob_{\calM[\calN], \initState}^{\inf}(\F \Goal)$.
\end{thm}
\begin{proof}
  We have already established that the two schemes converge.
  It remains to discuss the effectivity of the schemes. Since the induced MDP is
  finitely action-branching,
  Lemma~\ref{lem:schedulerComplexity} (item~\ref{item:infOptPurePos}) implies that
  computing $\Avoid^{\inf}_\calM(\Goal)$ amounts to computing states
  from which one can almost-surely avoid~$\Goal$ under a pure and
  positional scheduler, which amounts to computing states from which one
  can (surely) avoid $\Goal$. The latter set can be effectively computed
  as a fixed point~\cite{BBS-acmtocl07}.
  Also, an unfolding of the induced MDP $\calM[\calN]$ can be effectively constructed and solved.
\end{proof}

\paragraph{Approximation of the supremum reachability probability.}
It is unknown whether the supremum reachability probability can be approximated in general for NPLCSs.
A direct application of Theorem~\ref{thm:approx-opt} is not possible because NPLCSs need not be $\sup$-decisive: for instance, one can turn $\calM^\texttt{R}$ from Figure~\ref{fig:distinguishingAvoidSets} into an NPLCS, and the approximation scheme $\ApproxScheme_\unfold^{\sup}$ does not converge.

When the induced MDP is finitely branching and non-fleeing, we know that the approximation scheme $\ApproxScheme_\sliced^{\sup}$ converges (Theorem~\ref{thm:nonFleeing}).
Our main observation in this section is that non-fleeingness is decidable for NPLCSs thanks to a simple structural characterisation.

\newcommand{\graph}{G^{\not{\, ?},+}_{\initState,\Goal}}

Given an NPLCS $\calN$ inducing an MDP $\calM$, an initial control state $\initState$, and an absorbing target state $\Goal \in Q$, let
\[
  Q^+_{\Goal} = \{q \in Q \setminus \{\Goal\} \mid
  (q,\emptyConfig) \notin \Avoid^{\sup}_{\calM}(\Goal)\}\enspace.
\]
We define the \emph{non-receive control graph} $\graph = (V, E)$ where
\begin{itemize}
  \item $V \subseteq Q^+_{\Goal}$ is the set of control states $q \in Q^+_{\Goal}$ such that there exists a configuration $(q,\ccont)$ reachable from $(\initState,\emptyConfig)$ with positive probability under some scheduler;
  \item $E \subseteq V \times V$ is the set of edges $(q, q')$ such that there exists a transition rule $\delta = (q, \op, q')$ in $\calN$ with $\op = \ell \in \ActLCS$ or $\op = c!m$ for some $c \in \calC$ and $m \in \Mess$ (that is, $\delta$ is not a receive operation).
\end{itemize}
In this notation, $\not ?$ indicates that receive edges are excluded, while $+$ indicates that the vertices can reach $\Goal$ with positive probability from empty channels.
We call an edge $(q, q')$ in $E$ a \emph{send edge} if there is a corresponding send operation $\delta = (q, c!m, q')$ for some $c \in \calC$ and $m \in \Mess$.

We first observe that, for every $q \in Q^+_{\Goal}$ and all channel contents $\ccont$, $(q,\ccont) \notin \Avoid^{\sup}_{\calM}(\Goal)$.
Indeed, consider a finite path with positive probability from $(q,\emptyConfig)$ to $\Goal$.
As the channels are initially empty, the first rule of this path is not a receive operation and is therefore also enabled in $(q,\ccont)$.
With positive probability, all messages initially present in $\ccont$ are lost during this transition, while the newly sent message, if any, has the same loss outcome as along the chosen path.
The transition then yields the same successor configuration, and the remainder of the path can be followed with positive probability.

\begin{lem} \label{lem:nplcs-nonfleeing}
  Let $\calN$ be an NPLCS with initial control state $\initState$ and absorbing target state $\Goal \in Q$, and let $\calM$ be the MDP induced by~$\calN$.
  Assume $(\initState,\emptyConfig) \notin \Avoid^{\sup}_{\calM}(\Goal)$.
  Then, $\calM$ is non-fleeing w.r.t.\ $\Goal$ from $(\initState,\emptyConfig)$ if and only if the non-receive control graph $\graph$ has no cycle containing a send edge.
\end{lem}
\begin{proof}
  We first prove the ``only if'' direction.
  Assume that $\graph$ has a cycle containing a send edge.
  By definition of $V$, there is a scheduler (which can be assumed to be pure and positional) that reaches a configuration whose control state belongs to this cycle with positive probability, and then loops around the cycle forever.
  It is always possible to loop around the cycle, independently of the channel contents, because the cycle does not contain a receive operation.
  Moreover, all configurations encountered while following the cycle are outside $\Avoid^{\sup}_{\calM}(\Goal)$, by the observation above, so the corresponding paths in $\calM^{\sup}$ never reach $\bad^{\sup}$.

  The cycle contains a send edge, hence at least one message is sent during each traversal.
  Fix $k \ge 1$ and divide the successive traversals into disjoint blocks of $k$ traversals.
  In each block, there is a fixed positive probability that one newly sent message from every traversal survives until the end of the block, and thus that the total channel length becomes at least~$k$; these events rely on disjoint loss trials in distinct blocks.
  It follows that, conditional on reaching the cycle, arbitrarily large channel contents are reached almost surely.
  Since arbitrarily large channel contents cannot all belong to any fixed finite set $\Reachable{n}$, the resulting paths leave every $\Reachable{n}$ and therefore satisfy $\Div$.
  Thus, $\calM$ is fleeing.

  We now prove the ``if'' direction.
  Assume that $\graph$ has no cycle containing a send edge.
  Let $B$ be the number of send edges in $\graph$.
  Every finite path in $\graph$ contains therefore at most $B$ send edges --- if a send edge occurred twice along a path, it would imply the existence of a cycle containing a send edge.
  Fix $\sigma \in \SchedPurePos(\calM^{\sup})$ and extend it arbitrarily to a pure and positional scheduler of $\calM$.
  By repeatedly applying Lemma~\ref{lem:finite-attractor}, configurations with empty channels are visited infinitely often with probability~$1$ under this scheduler.
  Paths in $\calM$ and $\calM^{\sup}$ coincide until $\Avoid^{\sup}_{\calM}(\Goal)$ is reached; hence, almost surely, every path in~$\calM^{\sup}$ that never reaches $\bad^{\sup}$ visits empty-channel configurations infinitely often.

  Consider a visit to an empty-channel configuration $(q,\emptyConfig)$ along a path of $\calM^{\sup}$ that has reached neither $\Goal$ nor $\bad^{\sup}$.
  Necessarily, $q \in V$.
  Consider the event that every subsequently sent message is immediately lost as long as the control state remains in~$V$.
  As long as this event occurs, the channels remain empty, receive operations are disabled, and the sequence of control states follows a path in~$\graph$.
  Hence, including the exiting transition, at most $B+1$ messages are sent, and this event has probability at least $\lambda^{B+1} > 0$.

  Under this event, either the control state leaves $V$, in which case the empty-channel successor has control state $\Goal$ or belongs to $\Avoid^{\sup}_{\calM}(\Goal)$, or the control state remains in~$V$ and eventually only internal operations are taken, so that the channels remain empty forever.
  In both cases, $\Div$ does not occur.
  Every empty-channel visit therefore provides a conditional probability at least $\lambda^{B+1}$ of preventing $\Div$.
  Since, almost surely, such visits occur infinitely often along paths that have not reached $\Goal$ or $\bad^{\sup}$, the probability of failing at every visit can be shown to be at most $\lim_{k \to \infty}(1-\lambda^{B+1})^k=0$.
  We conclude that $\Prob^\sigma_{\calM^{\sup},(\initState,\emptyConfig)}(\Div)=0$.
  Thus, $\calM$ is non-fleeing w.r.t.\ $\Goal$ from $(\initState,\emptyConfig)$.
\end{proof}

\begin{figure}[htbp]
  \centering
  \begin{tikzpicture}
    \node[j0] (q0) at (0,0) {$q_0$};
    \node[j0] (q1) at (-1.8,0.9) {$q_1$};
    \node[j0] (q2) at (-1.8,-0.9) {$q_2$};
    \node[j0] (try) at (2,0) {$q_{\mathsf{try}}$};
    \node[j0] (goal) at (4,0) {$\Goal$};
    \draw[-latex'] (0,0.8) -- (q0);
    \path (q0) edge[-latex'] node[above,sloped] {$!a$} (q1);
    \path (q1) edge[-latex'] node[left] {$!a$} (q2);
    \path (q2) edge[-latex'] node[below,sloped] {$?a$} (q0);
    \path (q0) edge[-latex'] node[above] {$\ell_{\mathsf{try}}$} (try);
    \path (try) edge[-latex'] node[above] {$?a$} (goal);
  \end{tikzpicture}
  \caption{A non-fleeing NPLCS with one channel and message alphabet $\{a\}$.}
  \label{fig:nplcs-nonfleeing-example}
\end{figure}

\begin{exa} \label{ex:nplcs-nonfleeing}
  Consider the NPLCS in Figure~\ref{fig:nplcs-nonfleeing-example}, with any loss rate $\lambda\in(0,1)\cap\bbQ$, initial configuration $(q_0,\emptyConfig)$, and absorbing target~$\Goal$.
  Each traversal of the cycle without losses adds one~$a$ to the channel, so arbitrarily long channel contents are reachable.
  Taking $\ell_{\mathsf{try}}$ with $k$ messages succeeds with probability $1-\lambda^k$: at least one message must survive this step to enable the final receive.
  Thus, accumulating messages improves the chance of success when trying, but there is a trade-off: each additional cycle risks deadlock at $(q_2,\emptyConfig)$.
  Nevertheless, $Q^+_{\Goal}=V=\{q_0,q_1\}$, since $q_2$ and $q_{\mathsf{try}}$ are deadlocked with empty channels.
  The graph $\graph$ therefore consists only of the edge $q_0\step{!a}q_1$, so Lemma~\ref{lem:nplcs-nonfleeing} implies non-fleeingness.
  The supremum reachability probability is neither $0$ (because there is a possible path to the target) nor $1$ (because going directly to $q_{\mathsf{try}}$ loses, and there is probability at least $\lambda^2$ of deadlock at the first visit to $q_2$), so it is a non-trivial value in $(0,1)$.
  Moreover, it depends on the loss rate $\lambda$, which affects the trade-off discussed above.
  Using the upcoming Proposition~\ref{prop:nplcs-sup-nonfleeing}, we can compute an $\epsilon$-approximation of the supremum reachability probability for any $\epsilon>0$ using the approximation scheme $\ApproxScheme_\sliced^{\sup}$.
\end{exa}

We conclude this section with the following result, which combines the above observations and simple effectiveness considerations.
See Figure~\ref{fig:nplcs-sup-dependencies} for the chain of arguments supporting this result.
\begin{prop} \label{prop:nplcs-sup-nonfleeing}
  Let $\calN$ be an NPLCS with initial control state $\initState$ and absorbing target control state $\Goal$, and let $\calM$ be the MDP induced by~$\calN$.
  There is an algorithm that
  \begin{itemize}
    \item returns~$0$ if $(\initState,\emptyConfig) \in \Avoid^{\sup}_{\calM}(\Goal)$, and
    \item otherwise decides whether $\calM$ is non-fleeing w.r.t.\ $\Goal$ from $(\initState,\emptyConfig)$ and, if so, for a given rational number $\epsilon > 0$, returns a value $\epsilon$-close to
    $\Prob_{\calM[\calN], \initState}^{\sup}(\F \Goal)$.
  \end{itemize}
\end{prop}
\begin{figure}[tb]
  \centering
  \begin{tikzpicture}[
    dependency/.style={draw, rounded corners=2pt, align=center, inner sep=3pt,
      minimum height=1.05cm},
    outcome/.style={dependency, fill=gray!15},
    implication/.style={-latex', semithick},
    every node/.style={font=\footnotesize}
  ]
    \node[dependency, text width=4.5cm] (test) at (0,-1.)
      {Non-fleeingness is decidable in NPLCSs\\
       (Lemma~\ref{lem:nplcs-nonfleeing})};
    \node[dependency, text width=4.5cm] (scheme) at (0,.85)
      {For finitely branching MDPs, non-fleeing
       $\Longrightarrow \ApproxScheme_\sliced^{\sup}$ converges\\
       (Theorem~\ref{thm:nonFleeing})};
    \node[outcome, text width=4.5cm] (algorithm) at (6.5,0)
      {Non-fleeingness is decidable;\\if it holds, an approximation
       of $\Prob^{\sup}(\F\Goal)$\\can be computed\\
       (Proposition~\ref{prop:nplcs-sup-nonfleeing})};

    \draw[implication] (test.east) -- (algorithm);
    \draw[implication] (scheme.east) -- (algorithm);
  \end{tikzpicture}
  \caption{High-level dependencies for approximating the supremum reachability
    probabilities in non-fleeing NPLCSs.}
  \label{fig:nplcs-sup-dependencies}
\end{figure}
\begin{proof}
  Let $\calN = ((Q,\calC,\Mess,\ActLCS,\Delta), \lambda)$.
  By item~\ref{item:sup0} of Theorem~\ref{thm:decNPLCS}, we first decide whether $(\initState,\emptyConfig) \in \Avoid^{\sup}_{\calM}(\Goal)$.
  If so, the supremum reachability probability is~$0$.
  Otherwise, the fact that non-fleeingness is decidable is a corollary of Lemma~\ref{lem:nplcs-nonfleeing}:
  \begin{itemize}
    \item since $Q$ is finite, the set $Q^+_{\Goal}$ can be computed by applying item~\ref{item:sup0} of Theorem~\ref{thm:decNPLCS} to every~$(q,\emptyConfig)$;
    \item membership in $V$ is decidable by a control-state reachability test~\cite{BBS-acmtocl07}, so the finite graph $\graph$ can be effectively constructed;
    \item then, one can decide whether a finite directed graph has a cycle containing a send edge.
  \end{itemize}

  If $\calM$ is non-fleeing, Theorem~\ref{thm:nonFleeing} and finite branching ensure that $\ApproxScheme_\sliced^{\sup}$ converges.
  As for the infimum case, membership in $\Avoid^{\sup}_\calM(\Goal)$ is decidable, so the finite slices can be effectively built and solved.
  Hence, $\ApproxScheme_\sliced^{\sup}$ can be used to
  compute the required $\epsilon$-approximation of $\Prob_{\calM[\calN], \initState}^{\sup}(\F \Goal)$.
\end{proof}

\subsection{Partially observable MDPs}


We focus in this section on \emph{partially observable Markov decision processes}, abbreviated \emph{POMDPs}~\cite{KLC98,CCT16}.
Like MDPs, they exhibit both nondeterminism and probabilistic transitions; they are more general in that the scheduler making decisions does not, in general, observe the current state, but receives a \emph{signal} at each step that gives \emph{partial information} about the current state.
All decisions must be based on the sequence of actions selected and signals received (and not the states visited) up to some point.
Given such a sequence, a common way to represent the most accurate information about our current knowledge of the state of the system is through the probability distribution on the possible states, called a \emph{belief}.
Even though we consider POMDPs with finitely many states, actions, and signals, POMDPs are relevant in our framework as, from a fixed initial belief, they induce naturally a countable MDP on the state space of beliefs.

Many basic quantitative problems in POMDPs are undecidable, already for simple reachability and safety objectives (equivalently, $\sup$-reachability and $\inf$-reachability).
This undecidability stems from results on the less general model of \emph{probabilistic automata}~\cite{Paz71,MHC03,GO10,Fij17}.
Here are some examples of undecidable problems for probabilistic automata (and thus for POMDPs):
\begin{itemize}
    \item Given a probabilistic automaton and a threshold $\lambda \in (0, 1)$, decide whether there is a finite input word that ensures that a goal state is reached with probability at least $\lambda$~\cite{Paz71}.
    The same holds when replacing ``reached'' with ``avoided''.
    \item Given a probabilistic automaton, decide whether the supremum probability of reaching a goal state is $1$~\cite{GO10}.
    \item Given $0 < \epsilon < \frac12$ and a probabilistic automaton such that either $(i)$ there is an input word accepted with probability at least $1 - \epsilon$ or $(ii)$ all input words are accepted with probability at most $\epsilon$, decide which case holds~\cite{MHC03}.
\end{itemize}
The latter problem is especially relevant to our setting, as it implies that there is no approximation algorithm for the \emph{supremum} probability of reachability in POMDPs.
Hence, we will not be able to make use of $\ApproxScheme_\unfold^{\sup}$ or $\ApproxScheme_\sliced^{\sup}$ on general POMDPs.

Yet, none of these results imply that the \emph{infimum} reachability value (or equivalently, via complementation, the supremum safety value) cannot be approximated in POMDPs.
Using the $\inf$-decisiveness property and $\ApproxScheme_\unfold^{\inf}$, we show that there exists such an algorithm.
See Figure~\ref{fig:pomdp-inf-dependencies} for the chain of arguments supporting this result.

\begin{figure}[tb]
  \centering
  \begin{tikzpicture}[
    dependency/.style={draw, rounded corners=2pt, align=center, inner sep=3pt,
      minimum height=1.05cm},
    outcome/.style={dependency, fill=gray!15},
    implication/.style={-latex', semithick},
    every node/.style={font=\footnotesize}
  ]
    \node[dependency, text width=2.8cm] (belief) at (0,0)
      {The belief MDP need not be $\inf$-decisive\\
       (Example~\ref{ex:pomdpNotInf})};
    \node[dependency, text width=2.5cm] (modified) at (3.3,0)
      {Modified belief MDP $\mdp[\pomdp]^{\tfrac\epsilon2}$};
    \node[dependency, text width=3.8cm] (decisive) at (7.2,1.1)
      {$\mdp[\pomdp]^{\tfrac\epsilon2}$ is $\inf$-decisive
       (Lemma~\ref{lem:pomdpInfDecisive}); as it is finitely action-branching,
       $\ApproxScheme_\unfold^{\inf}$ converges
       (Theorem~\ref{thm:approx-opt})};
    \node[dependency, text width=3.8cm] (close) at (7.2,-1.1)
      {Its infimum reachability value is within $\frac\epsilon2$ of the value in
       $\mdp[\pomdp]$\\(Lemma~\ref{lem:pomdpEpsilonClose})};
    \node[outcome, text width=3.2cm] (algorithm) at (11.6,0)
      {Effective {$\epsilon$-approximation} for the original POMDP\\
       (Theorem~\ref{thm:pomdpSafetyDecidable})};

    \draw[implication] (belief.east) -- (modified.west);
    \draw[implication] (modified.east) -- (decisive.west);
    \draw[implication] (modified.east) -- (close.west);
    \draw[implication] (decisive.east) -- (algorithm);
    \draw[implication] (close.east) -- (algorithm);
  \end{tikzpicture}
  \caption{High-level dependencies for approximating infimum reachability
    probabilities in POMDPs.}
  \label{fig:pomdp-inf-dependencies}
\end{figure}

\paragraph{POMDPs and induced MDP semantics.}
We first recall basic notions on partially observable MDPs (POMDPs).
\begin{defi}
    A \emph{partially observable MDP} is a tuple $\pomdp = \pomdpFull$ where $\states$ is a finite set of states, $\actions$ is a finite set of actions, $\signals$ is a finite set of signals, and $\transitions\colon \states \times \actions \times \signals \times \states \to [0, 1]\cap\bbQ$ is a transition function such that for all $\state \in \states$ and $\act \in \actions$, $\sum_{\sig \in \signals} \sum_{\state' \in \states} \transitions(\state, \act, \sig, \state') = 1$.
\end{defi}

The main difference from the semantics of an MDP is that, in the case of POMDPs, the information of the current state is not known by schedulers in general; schedulers must base their decisions on the signals they receive (as well as the actions they have already selected).
To keep this section concise, we will only express the semantics of POMDPs through the equivalent formulation of \emph{belief MDPs}~\cite{KLC98} below.

Let $\pomdp = \pomdpFull$ be a POMDP.
We assume without loss of generality that there is a distinguished absorbing state $\goalPomdp\in\states$ and a distinguished signal $\goalSig\in\signals$ that occurs if and only if the successor state is $\goalPomdp$ (this reduction is folklore---see~\cite[Lemma~2.1]{FGKPV26} for details).
In practice,
\begin{itemize}
    \item for all $\state\in\states$, $\act\in\actions$, and $\sig\in\signals$, $\transitions(\state, \act, \sig, \goalPomdp) > 0$ implies $\sig = \goalSig$,
    \item for all $\state,\state'\in\states$ and $\act\in\actions$, $\transitions(\state, \act, \goalSig, \state') > 0$ implies $\state' = \goalPomdp$, and
    \item for all $\act\in\actions$, $\transitions(\goalPomdp, \act, \goalSig, \goalPomdp) = 1$.
\end{itemize}
In other words, when the state $\goalPomdp$ is reached, the scheduler is aware of it (through the observation of signal $\goalSig$) and cannot escape it.
We recall that, in this section, we focus on the \emph{infimum} probability of reachability, which means that we try as much as possible \emph{not to} reach state~$\goalPomdp$.

We write $\beliefs = \{\belief\colon \states \to [0, 1] \cap \bbQ \mid \sum_{\state\in\states} \belief(\state) = 1\}$ for the set of distributions over $\states$ with rational values.
A \emph{belief (of $\pomdp$)} is a probability distribution $\belief\in\beliefs$.
The \emph{belief-update function} is the partial function $\beliefUpd\colon \beliefs \times \actions \times \signals \rightharpoonup \beliefs$ such that for all $(\belief, \act, \sig) \in \beliefs \times \actions \times \signals$,
\[
\beliefUpd(\belief, \act, \sig)(\state') = \frac{\sum_{\state \in \states} \belief(\state) \cdot \transitions(\state, \act, \sig, \state')}{\sum_{\state\in\states}(\belief(\state)\cdot \sum_{\state''\in\states} \transitions(\state, \act, \sig, \state''))}\enspace,
\]
which is not defined if the denominator is $0$ (i.e., if the signal $\sig$ cannot be observed from belief $\belief$ after selecting action $\act$).
The belief $\beliefUpd(\belief, \act, \sig)$ corresponds to the new belief that the scheduler has after selecting action $\act$ and observing signal $\sig$ from belief $\belief$.
The \emph{support} of $\belief$, denoted $\supp(\belief)$, is the set $\{\state\in\states\mid\belief(\state) > 0\}$.
The set of all belief supports then corresponds to the set $\powerSetNonEmpty$.
We use fraktur letters (i.e., $\belief$) for beliefs and italic letters (i.e., $\beliefSupp$) for supports.

The \emph{belief MDP of $\pomdp$} is the countable MDP $\mdp[\pomdp] = (\beliefs, \actions, \transitions_\pomdp)$ where $\beliefs$ is the set of beliefs, $\actions$ is the set of actions, and $\transitions_\pomdp\colon \beliefs \times \actions \times \beliefs \to [0, 1] \cap \bbQ$ is
\[
    \transitions_\pomdp(\belief, \act, \belief') = \sum_{\substack{\sig \in \signals\, \text{s.t.}\\ \beliefUpd(\belief, \act, \sig) = \belief'}} \sum_{\state, \state'\in\states} \belief(\state) \cdot \transitions(\state, \act, \sig, \state')\enspace.
\]
Given our assumptions about the state $\goalPomdp$ of $\pomdp$, we have that as soon as $\goalPomdp$ is reached in the POMDP, we reach the corresponding belief $\goalPomdp \mapsto 1$ in the belief~MDP.
We denote this belief by $\Goal$.

If $\initState\in\states$, we slightly abuse notation and write $\Prob_{\mdp[\pomdp], \initState}^{\inf}(\F \Goal)$ for the infimum probability of reaching $\Goal$ in $\mdp[\pomdp]$ starting from the belief $\initState \mapsto 1$ (i.e., we identify the notation $\initState$ with the Dirac distribution $\initState \mapsto 1$).

\paragraph{Approximation of the infimum reachability probability.}
Our plan is to apply Theorem~\ref{thm:approx-opt} (and thus $\ApproxScheme_\unfold^{\inf}$) to the countable MDP $\mdp[\pomdp]$ (or a modified version of it) to approximate the infimum probability of reaching the goal state in a POMDP $\pomdp$.
Observe first that $\mdp[\pomdp]$ is finitely action-branching as there are only finitely many actions in $\actions$ (and actually, even finitely branching as there are only finitely many signals in $\signals$, but this is not necessary to use Theorem~\ref{thm:approx-opt}).
We would therefore need an $\inf$-decisiveness condition for $\mdp[\pomdp]$.
However, in general, $\mdp[\pomdp]$ is not $\inf$-decisive.

\begin{exa} \label{ex:pomdpNotInf}
    Consider the POMDP $\pomdp$ in Figure~\ref{fig:pomdpNotInf}.
    It has a single action $\alpha$.
    Starting from~$\state_0$, if $\state_1$ is reached, then only signal $\sig$ will ever be seen.
    Yet, through successive observations of $\sig$, the scheduler can never be sure to be in $\state_1$; there is a decreasing but positive probability that the current state is~$\state_2$.
    Formally, one can show that the belief~$\belief_n$ after seeing the sequence of signals $\sig^n$ (with $n \ge 1$) is defined by $\belief_n(\state_1) = \frac{3^n-1}{3^n+1}$ and $\belief_n(\state_2) = \frac{2}{3^n+1}$.
    All these beliefs still have a positive probability to reach $\Goal$ (from $\state_2$), so none of them are in $\Avoid^{\inf}_{\mdp[\pomdp]}(\Goal)$.
    There is therefore a positive probability to stay in a region of $\mdp[\pomdp]$ that is neither $\Goal$ nor in $\Avoid^{\inf}_{\mdp[\pomdp]}(\Goal)$, which shows that $\mdp[\pomdp]$ is not $\inf$-decisive w.r.t.~$\Goal$ from~$\state_0$.
\end{exa}

\begin{figure}
    \centering
    \begin{tikzpicture}
        \tikzstyle{dot}=[draw,circle,fill=black]
		\draw ($(1,0)$) node[j0] (q0) {$\state_0$};

        \draw ($(q0)+(1.5,0)$) node[dot] (q0dot) {};
		\draw ($(q0)+(2,1.2)$) node[j0] (q1) {$\state_1$};
		\draw ($(q0)+(2,-1.2)$) node[j0] (q2) {$\state_2$};
		\draw ($(q0)+(5,0)$) node[j0] (goal) {$\goalPomdp$};

		\draw (q0) edge[-latex'] node[above] {$\alpha$} (q0dot);
		\draw (q0dot) edge[-latex',bend right] node[right] {$\sig, \frac{1}{3}$} (q1);
		\draw (q0dot) edge[-latex',bend left] node[right] {$\sig, \frac{1}{3}$} (q2);
		\draw (q0dot) edge[-latex'] node[above right] {$\mathsf{done}, \frac{1}{3}$} (goal);
		\draw (q2) edge[-latex',bend left] node[left] {$\alpha$} (q0dot);
		\draw (q1) edge[-latex',loop left] node[left] {$\alpha, \sig, 1$} (q1);
	\end{tikzpicture}
    \caption{A POMDP $\pomdp$ such that $\mdp[\pomdp]$ is not $\inf$-decisive w.r.t.~$\Goal$ from~$\state_0$. The black dot is used for probabilistic branching after playing action $\alpha$.}
    \label{fig:pomdpNotInf}
\end{figure}

Our proof scheme is as follows: even though $\mdp[\pomdp]$ is not $\inf$-decisive in general, we show that for every POMDP $\pomdp$ and every $\epsilon\in (0, 1)$, we can modify $\mdp[\pomdp]$ slightly to obtain an infinite MDP $\mdp[\pomdp]^\epsilon$ such that:
\begin{itemize}
    \item $\mdp[\pomdp]^\epsilon$ is $\inf$-decisive (Lemma~\ref{lem:pomdpInfDecisive}),
    \item the infimum probability of reaching the goal state in $\mdp[\pomdp]^\epsilon$ is within $\epsilon$ of the infimum probability of reaching the goal state in $\mdp[\pomdp]$ (Lemma~\ref{lem:pomdpEpsilonClose}).
\end{itemize}
To obtain an approximation \emph{algorithm}, we then discuss how to effectively compute the sequences $(p_n^{\inf, -})_n$ and $(p_n^{\inf, +})_n$ in the infinite $\mdp[\pomdp]^\epsilon$ given the finite representation of $\pomdp$ (Theorem~\ref{thm:pomdpSafetyDecidable}).

Let $\pomdp = \pomdpFull$ be a POMDP, $\initState\in\states$ be an initial state, and $\epsilon\in (0, 1)$.
We construct $\mdp[\pomdp]^\epsilon$ from $\mdp[\pomdp]$.

Observe that if a belief $\belief$ is such that there is a scheduler $\sigma$ such that $\Prob_{\mdp[\pomdp], \belief}^{\sigma}(\F \Goal) = 0$, then for all beliefs~$\belief'$ with $\supp(\belief') = \supp(\belief)$, we also have $\Prob_{\mdp[\pomdp], \belief'}^{\sigma}(\F \Goal) = 0$.
We define \[\ASbeliefSupps = \{\beliefSupp \in \powerSetNonEmpty \mid \exists \sigma, \forall \belief, (\supp(\belief) = \beliefSupp \Longrightarrow \Prob_{\mdp[\pomdp], \belief}^{\sigma}(\F \Goal) = 0)\}.\]

To build $\mdp[\pomdp]^\epsilon$, we merge some specific beliefs of $\mdp[\pomdp]$ into a single, fresh absorbing state~$\bad^\epsilon$.
The beliefs that are merged are the beliefs $\belief$ such that
\[
    \exists \beliefSupp' \subseteq \supp(\belief)\ \text{such that}\ \beliefSupp' \in \ASbeliefSupps\ \text{and}\ \sum_{\state\in\beliefSupp'} \belief(\state) \ge 1 - \epsilon\enspace.
\]
We call such a belief a \emph{$(1 - \epsilon)$-avoiding} belief.
Intuitively, this is a belief where, if the current state is in a specific subset $\beliefSupp'$ of the support (which occurs with probability $\ge 1 - \epsilon$), a scheduler can ensure that the goal state is never reached.
An advantage of this definition (compared to the actual infimum reachability probability, which we are trying to compute) is that it is decidable: we can precompute the finite set $\ASbeliefSupps$ (of cardinality at most $2^{|\states|} - 1$) and then easily check whether a belief is $(1 - \epsilon)$-avoiding.

Formally, we define the state space of $\mdp[\pomdp]^\epsilon$ as
\[
    S^\epsilon = \{\bad^\epsilon\} \cup \{\belief\in\beliefs \mid
    \text{$\belief$ is not $(1-\epsilon)$-avoiding}\}\enspace.
\]
The transitions are then kept the same as in $\mdp[\pomdp]$, except that the transitions to a $(1-\epsilon)$-avoiding belief are redirected to the absorbing state $\bad^\epsilon$.
We assume that the initial belief $\initState \mapsto 1$ is not $(1-\epsilon)$-avoiding, otherwise we can immediately return $0$ as an $\epsilon$-approximation of the infimum reachability probability.

\begin{exa}
    We build the MDP $\mdp[\pomdp]^\epsilon$ obtained from the POMDP $\pomdp$ in Example~\ref{ex:pomdpNotInf}, with $\epsilon = \frac{1}{8}$.
    It is shown in Figure~\ref{fig:pomdpNotInfEps}; we recall that state $\Goal$ in $\mdp[\pomdp]$ corresponds to the belief $\goalPomdp \mapsto 1$.

    Observe that the only belief support in $\ASbeliefSupps$ is $\{\state_1\}$.
    Hence, the beliefs that are $(1 - \epsilon)$-avoiding are the beliefs $\belief$ such that $\belief(\state_1) \ge \frac{7}{8}$.
    In this case, the set of states reachable from~$\state_0$ is finite, and so $\mdp[\pomdp]^\epsilon$ is clearly $\inf$-decisive from $\state_0$.
\end{exa}

\begin{figure}
    \centering
    \begin{tikzpicture}
		\draw ($(1,0)$) node[j0] (q0) {$\state_0 \mapsto 1$};

        \draw ($(q0)+(3,0)$) node[j0,align=left] (q1) {$\state_1 \mapsto \frac{1}{2}$\\$\state_2 \mapsto \frac{1}{2}$};
		\draw ($(q1)+(3,0)$) node[j0,align=left] (q2) {$\state_1 \mapsto \frac{4}{5}$\\$\state_2 \mapsto \frac{1}{5}$};
		\draw ($(q2)+(3,0)$) node[j0] (bad) {$\bad^\epsilon$};
		\draw ($(q1)+(0,-2)$) node[j0] (goal) {$\Goal$};

		\draw (q0) edge[-latex'] node[above] {$\alpha, \frac{2}{3}$} (q1);
		\draw (q1) edge[-latex'] node[above] {$\alpha, \frac{5}{6}$} (q2);
		\draw (q2) edge[-latex'] node[above] {$\alpha, \frac{14}{15}$} (bad);
		\draw (q0) edge[-latex',bend right] node[below left] {$\alpha, \frac{1}{3}$} (goal);
		\draw (q1) edge[-latex'] node[left] {$\alpha, \frac{1}{6}$} (goal);
		\draw (q2) edge[-latex',bend left] node[below right] {$\alpha, \frac{1}{15}$} (goal);
	\end{tikzpicture}
    \caption{The MDP $\mdp[\pomdp]^\epsilon$ obtained from $\pomdp$ in Example~\ref{ex:pomdpNotInf} with $\epsilon = \frac{1}{8}$.}
    \label{fig:pomdpNotInfEps}
\end{figure}

We first recall a lemma about reachability properties in POMDPs: if a state is reached with positive probability under every scheduler, then the probability to reach it is uniformly bounded away from $0$ for all schedulers.
For a POMDP $\pomdp$, we define $\leastProb_\pomdp\in(0, 1]\cap\bbQ$ as the minimal positive probability appearing in the syntactic description of~$\pomdp$ (i.e., $\leastProb_\pomdp = \min\{\transitions(\state, \act, \sig, \state') \mid \transitions(\state, \act, \sig, \state') > 0\}$).

\begin{lem}[Corollary of~{\cite[Lemma~4]{BFGHPV25}}] \label{lem:lowerBoundPOMDP}
    Let $\pomdp = \pomdpFull$ be a POMDP, with $\initState\in\states$ an initial state and $\goalPomdp\in\states$ a target state.
    If for all schedulers $\sigma$, $\Prob_{\mdp[\pomdp], \initState}^{\sigma}(\F \Goal) > 0$, then for all schedulers $\sigma$, $\Prob_{\mdp[\pomdp], \initState}^{\sigma}(\F[\le 2^{|\states|} - 1] \Goal) \ge \leastProb_\pomdp^{2^{|\states|} - 1}$.
\end{lem}

We can now state and prove the aforementioned results leading to an approximation algorithm for the infimum probability of reachability in POMDPs.
For the next two lemmas, let $\pomdp = \pomdpFull$ be a POMDP, $\initState\in\states$ be an initial state, $\goalPomdp\in\states$ be a target state, and $\epsilon\in (0, 1)$.

\begin{lem} \label{lem:pomdpInfDecisive}
    Let $\sigma\in\SchedPurePos(\mdp[\pomdp]^\epsilon)$ be a pure and positional scheduler.
    We have
    \[\Prob^\sigma_{\mdp[\pomdp]^\epsilon,\initState}\big(\F \Goal \vee \F \bad^\epsilon \big) = 1\enspace.\]
    In particular, the MDP $\mdp[\pomdp]^\epsilon$ is $\inf$-decisive w.r.t.~$\Goal$ from~$\initState$.
\end{lem}
\begin{proof}
  We show that
  \[\Prob^\sigma_{\mdp[\pomdp]^\epsilon,\initState}\big(\F \Goal \vee \F
  \bad^\epsilon \big) = 1\enspace.\]
  This implies that $\mdp[\pomdp]^\epsilon$ is $\inf$-decisive w.r.t.~$\Goal$ from~$\initState$, since $\bad^\epsilon\in\Avoid_{\mdp[\pomdp]^\epsilon}^{\inf}(\Goal)$ by construction.

  It suffices to show that there is a positive probability $\beta > 0$ such that any belief $\belief$ that is not $(1-\epsilon)$-avoiding has probability $\ge \beta$ to reach~$\Goal$ in $\mdp[\pomdp]$, for all schedulers.
  Indeed, this implies that, in $\mdp[\pomdp]^\epsilon$, a path either ends up in~$\bad^\epsilon\in\Avoid_{\mdp[\pomdp]^\epsilon}^{\inf}(\Goal)$, or only visits beliefs that are not $(1-\epsilon)$-avoiding, and thus reaches~$\Goal$ with probability~$1$.

  We show the existence of such a $\beta$.
  Let $\belief\in\beliefs$ be a reachable belief that is not $(1-\epsilon)$-avoiding.
  Consider $\beliefSupp' = \{\state\in\supp(\belief) \mid \belief(\state) \ge \frac{\epsilon}{|\states|}\}$.
  We have $\sum_{\state\in\beliefSupp'} \belief(\state) \ge 1 - \epsilon$.
  Since $\belief$ is not $(1-\epsilon)$-avoiding, we have $\beliefSupp' \notin \ASbeliefSupps$.
  Consider the normalisation~$\belief'$ of $\belief$ to the elements in~$\beliefSupp'$.
  Consider a modified POMDP $\pomdp^{\belief'}$ with an extra initial state $\initState^{\belief'}$ such that after one step, no matter the action, the belief is~$\belief'$.
  There is a natural correspondence between schedulers in $\mdp[\pomdp^{\belief'}]$ and schedulers in $\mdp[\pomdp]$ that play the same actions after the first step.
  Since $\beliefSupp' \notin \ASbeliefSupps$, for all schedulers $\sigma$, $\Prob_{\mdp[\pomdp^{\belief'}], \initState^{\belief'}}^{\sigma}(\F \Goal) = \Prob_{\mdp[\pomdp], \belief'}^{\sigma}(\F \Goal) > 0$.
  By Lemma~\ref{lem:lowerBoundPOMDP}, for all schedulers $\sigma$, $\Prob_{\mdp[\pomdp^{\belief'}], \initState^{\belief'}}^{\sigma}(\F[\le 2^{|\states| + 1}] \Goal) \ge \leastProb_{\pomdp^{\belief'}}^{2^{|\states| + 1}}$ (recall that $\pomdp^{\belief'}$ has $|\states| + 1$ states).
  Observe that $\leastProb_{\pomdp^{\belief'}} \ge \min\{\leastProb_\pomdp, \frac{\epsilon}{|\states|}\}$ since $\pomdp^{\belief'}$ is similar to~$\pomdp$ with an extra state and an extra transition based on the probabilities of~$\belief'$, which are all greater than or equal to $\frac{\epsilon}{|\states|}$.

  We conclude that we have the desired property by choosing
  \[
      \beta = (1 - \epsilon) \cdot \min\left\{\leastProb_{\pomdp}, \frac{\epsilon}{|\states|}\right\}^{2^{|\states| + 1}} > 0\enspace,
  \]
  where the $(1 - \epsilon)$ factor accounts for the normalisation of $\belief$ to $\belief'$.
\end{proof}

\begin{lem} \label{lem:pomdpEpsilonClose}
    We have that
\[
    \Prob_{\mdp[\pomdp]^\epsilon, \initState}^{\inf}(\F \Goal) \le \Prob_{\mdp[\pomdp], \initState}^{\inf}(\F \Goal) \le \Prob_{\mdp[\pomdp]^\epsilon, \initState}^{\inf}(\F \Goal) + \epsilon\enspace.
\]
\end{lem}
\begin{proof}
  We first show that $\Prob_{\mdp[\pomdp]^\epsilon, \initState}^{\inf}(\F \Goal) \le \Prob_{\mdp[\pomdp], \initState}^{\inf}(\F \Goal)$.
  For every scheduler $\sigma\in\Sched(\mdp[\pomdp])$, we can define a scheduler $\sigma^\epsilon\in\Sched(\mdp[\pomdp]^\epsilon)$ that chooses the same actions as $\sigma$ as long as $\bad^\epsilon$ is not reached and, if $\bad^\epsilon$ is reached, chooses any action.
  We have by construction that $\Prob_{\mdp[\pomdp]^\epsilon, \initState}^{\sigma^\epsilon}(\F \Goal) \le \Prob_{\mdp[\pomdp], \initState}^{\sigma}(\F \Goal)$.

  We now show that \[\Prob_{\mdp[\pomdp], \initState}^{\inf}(\F \Goal) \le \Prob_{\mdp[\pomdp]^\epsilon, \initState}^{\inf}(\F \Goal) + \epsilon\enspace.\]
  Let $\sigma^\epsilon\in\SchedPurePos(\mdp[\pomdp]^\epsilon)$ be such that $\Prob_{\mdp[\pomdp]^\epsilon, \initState}^{\sigma^\epsilon}(\F \Goal) = \Prob_{\mdp[\pomdp]^\epsilon, \initState}^{\inf}(\F \Goal)$ (such a scheduler exists by Lemma~\ref{lem:schedulerComplexity}, item~\ref{item:infOptPurePos} as $\mdp[\pomdp]^\epsilon$ is finitely action-branching).
  Let $\sigma\in\Sched(\mdp[\pomdp])$ be a scheduler that plays the same actions as $\sigma^\epsilon$ as long as only beliefs that are not $(1-\epsilon)$-avoiding are visited, and follows a scheduler that avoids $\Goal$ with probability $\ge (1 - \epsilon)$ as soon as a $(1-\epsilon)$-avoiding belief is visited for the first time.
  Let $B = \{\belief\in\beliefs \mid \belief\ \text{is $(1-\epsilon)$-avoiding}\}$.

  Using that $\Goal$ and $\bad^\epsilon$ are absorbing and disjoint, we have
  \begin{align}
      \Prob_{\mdp[\pomdp]^\epsilon, \initState}^{\sigma^\epsilon}(\F \Goal)
      &= 1 - \Prob_{\mdp[\pomdp]^\epsilon, \initState}^{\sigma^\epsilon}(\G (\neg \Goal \wedge \neg \bad^\epsilon)) - \Prob_{\mdp[\pomdp]^\epsilon, \initState}^{\sigma^\epsilon}(\F \bad^\epsilon)\notag\\
      &= 1 - \Prob_{\mdp[\pomdp]^\epsilon, \initState}^{\sigma^\epsilon}(\F \bad^\epsilon)\notag\\
      &= 1 - \Prob_{\mdp[\pomdp], \initState}^{\sigma}(\F B)\enspace, \label{eq:POMDPepsilon}
  \end{align}
  where the second line holds by Lemma~\ref{lem:pomdpInfDecisive}, and the third line holds by construction of $\mdp[\pomdp]^\epsilon$ and $\sigma$.

  Observe that when $B$ is first reached by $\sigma$, the probability to avoid $\Goal$ is at least $1 - \epsilon$.
  Hence,
  \begin{align*}
  \Prob_{\mdp[\pomdp], \initState}^{\sigma}(\G \neg \Goal)
  \ge \Prob_{\mdp[\pomdp], \initState}^{\sigma}(\F B)\cdot (1 - \epsilon)
  \ge \Prob_{\mdp[\pomdp], \initState}^{\sigma}(\F B) - \epsilon
  \enspace.
  \end{align*}
  We deduce that $\Prob_{\mdp[\pomdp], \initState}^{\sigma}(\F B) \le 1 - \Prob_{\mdp[\pomdp], \initState}^{\sigma}(\F \Goal) + \epsilon$.
  Plugging it back in the above equality~\eqref{eq:POMDPepsilon}, we find
  \[
      \Prob_{\mdp[\pomdp]^\epsilon, \initState}^{\sigma^\epsilon}(\F \Goal) \ge \Prob_{\mdp[\pomdp], \initState}^{\sigma}(\F \Goal) - \epsilon\enspace.
  \]
  Hence, $\Prob_{\mdp[\pomdp], \initState}^{\inf}(\F \Goal) \le \Prob_{\mdp[\pomdp]^\epsilon, \initState}^{\inf}(\F \Goal) + \epsilon$.
\end{proof}

\begin{thm} \label{thm:pomdpSafetyDecidable}
    There exists an algorithm that, given a POMDP $\pomdp$, an initial state $\initState$, a target state $\goalPomdp$, and a rational number~$\epsilon\in(0, 1)$, returns a value $\epsilon$-close to $\Prob_{\mdp[\pomdp], \initState}^{\inf}(\F \Goal)$.
\end{thm}
\begin{proof}
    We describe the algorithm.
    Let $\pomdp$ be a POMDP and $\epsilon \in (0, 1)$ be rational.
    We consider the MDP $\mdp[\pomdp]^{\epsilon/2}$, for which we know that
    \[
        \forall\sigma\in\SchedPurePos(\mdp[\pomdp]^{\epsilon/2}),
        \Prob^\sigma_{\mdp[\pomdp]^{\epsilon/2},\initState}\big(\F \Goal \vee \F \bad^{\epsilon/2} \big) = 1
    \]
    by Lemma~\ref{lem:pomdpInfDecisive}.
    We run a modified version of the approximation scheme $\ApproxScheme_\unfold^{\inf}$ on $\mdp[\pomdp]^{\epsilon/2}$, where we replace the avoid set $\Avoid_{\mdp[\pomdp]^{\epsilon/2}}^{\inf}(\Goal)$ by its under-approximation~$\bad^{\epsilon/2}$.
    Indeed, we have shown no way to compute $\Avoid_{\mdp[\pomdp]^{\epsilon/2}}^{\inf}(\Goal)$ effectively, but we can check beliefs for collapse into $\bad^{\epsilon/2}\in \Avoid_{\mdp[\pomdp]^{\epsilon/2}}^{\inf}(\Goal)$ on the fly, and we have the sufficient convergence guarantees from Lemma~\ref{lem:pomdpInfDecisive}.
    As $\mdp[\pomdp]^{\epsilon/2}$ is finitely action-branching, the proof of Theorem~\ref{thm:approx-opt} shows that the modified approximation scheme $\ApproxScheme_\unfold^{\inf}$ converges.

    To conclude, it remains to argue that the sequences $(p_n^{\inf, -})_n$, $(p_n^{\inf, +})_n$ appearing in $\ApproxScheme_\unfold^{\inf}$ can be computed effectively.
    We first compute the set $\ASbeliefSupps$; this corresponds to multiple almost-sure safety problems on $\pomdp$, which are decidable~\cite[Theorem~7]{CCT16}.
    All beliefs up to a fixed depth can be computed exactly (they are all arrays of rational numbers), and since $\ASbeliefSupps$ was precomputed, we can decide whether a belief is $(1 - \frac{\epsilon}{2})$-avoiding (and thus whether we have reached $\bad^{\epsilon/2}$ in~$\mdp[\pomdp]^{\epsilon/2}$).

    Hence, we have an effective algorithm that returns a value $v$ such that
    \[
    \Prob^{\inf}_{\mdp[\pomdp]^{\epsilon/2}, \initState}(\F \Goal) - \frac{\epsilon}{2} \le v \le \Prob^{\inf}_{\mdp[\pomdp]^{\epsilon/2}, \initState}(\F \Goal) + \frac{\epsilon}{2}\enspace.
    \]
    By Lemma~\ref{lem:pomdpEpsilonClose}, we have
    \[
    \Prob_{\mdp[\pomdp]^{\epsilon/2}, \initState}^{\inf}(\F \Goal) \le \Prob_{\mdp[\pomdp], \initState}^{\inf}(\F \Goal) \le \Prob_{\mdp[\pomdp]^{\epsilon/2}, \initState}^{\inf}(\F \Goal) + \frac{\epsilon}{2}\enspace.
    \]
    Hence,
    \[
        \Prob_{\mdp[\pomdp], \initState}^{\inf}(\F \Goal) - \epsilon \le v \le \Prob_{\mdp[\pomdp], \initState}^{\inf}(\F \Goal) + \frac{\epsilon}{2}\enspace,
    \]
    which suffices for an approximation algorithm with precision $\epsilon$.
\end{proof}

\section{Conclusion}


We extended decisiveness from countable Markov chains to three notions for countable MDPs: $\inf$-, univ-, and $\sup$-decisiveness.
We demonstrated how to use these notions to obtain converging approximation schemes for the infimum and supremum reachability probabilities (equivalently, the supremum and infimum safety probabilities, respectively).
The notion of $\inf$-decisiveness proved algorithmically useful, as we showed that it enables the approximation of the infimum reachability probability in two classes of infinite-state models: nondeterministic probabilistic lossy channel systems (NPLCSs) and partially observable MDPs (after a suitable transformation of the belief MDP).
For NPLCSs, we also proved that non-fleeingness is decidable and obtained an approximation algorithm for the supremum reachability probability whenever it holds.
We connected the stronger notion of $\sup$-decisiveness to the \emph{stopping} notion and, for finite MDPs, to an end-component condition used to ensure convergence of lower- and upper-bounded value iteration.

One concrete problem remains open for NPLCSs: whether the supremum reachability probability can be approximated without the non-fleeingness assumption.
More generally, natural directions for future research include extending our framework to richer objectives, such as Büchi objectives (repeated reachability), and to broader classes of models, such as \emph{probabilistic vector addition systems with states}.

\section*{Acknowledgements}
Thomas Brihaye is partly supported by the Fonds de la Recherche Scientifique -- FNRS under grant no.~T.0027.21 and by the Belgian National Lottery.
Pierre~Vandenhove is partly supported by the Fonds de la Recherche Scientifique -- FNRS through the CDR project ``INSTAP'' (grant no.~J.0228.26).
Nathalie Bertrand and Patricia Bouyer are partly supported by the ANR project BisoUS (ANR-22-CE48-0012).

\bibliographystyle{alphaurl}
\bibliography{refs}

\newpage
\appendix


\section{Technical properties of step-bounded events}
\label{app:stepBounded}

We prove Lemma~\ref{lem:uniformizingClopen} about step-bounded events discussed in Section~\ref{sec:stepBounded}.
To do so, we use the fact that pure schedulers suffice for step-bounded events, which is given by the following lemma.
\begin{lem} \label{lem:stepBounded}
  Let $\calM$ be a finitely action-branching MDP and let $s_0$ be an initial state.
  Let~$E$ be a step-bounded event.
  For all $\opt \in \{\inf,\sup\}$, there
  exists $\sigma \in \SchedPureMemoryful(\calM)$ such that
  $\Prob^{\sigma}_{\calM,s_0}(E) =
  \Prob^{\opt}_{\calM,s_0}(E)$.
\end{lem}
\begin{proof}
  Let $n\in\nats$ be a step bound of event $E$.
  Consider the unfolding $\calM_{n, s_0}$ of $\calM$ up to depth~$n$ from $s_0$.
  Add two absorbing states $s_E$ and $s_{\neg E}$; direct all histories~$h$ of length~$n$ to~$s_E$ if $\Cyl(h) \subseteq E$ and to $s_{\neg E}$ if $\Cyl(h) \cap E = \emptyset$.

  Notice that $\Prob^{\inf}_{\calM,s_0}(E) = \Prob^{\inf}_{\calM_{n, s_0},s_0}(\F s_E)$ and $\Prob^{\sup}_{\calM,s_0}(E) = 1 - \Prob^{\inf}_{\calM_{n, s_0},s_0}(\F s_{\neg E})$.
  Hence, for all $\opt \in \{\inf,\sup\}$, we reduce to the problem of infimum reachability probability on the finitely action-branching MDP $\calM_{n, s_0}$.
  Using Lemma~\ref{lem:schedulerComplexity} (item~\ref{item:infOptPurePos}), we can find a pure and positional scheduler $\sigma \in \SchedPurePos(\calM_{n, s_0})$ that achieves the infimum reachability probability.
  This scheduler corresponds to a pure history-dependent scheduler in $\calM$.
\end{proof}

The finite action-branching assumption is
needed for this statement to hold. Indeed, consider the infinitely
action-branching MDP
$\calM = (\{s_0,\Goal,\bad\},\{\alpha_i \mid i \ge 1\},\proba)$
with~$\Goal$ and~$\bad$ absorbing, $\proba(s_0,\alpha_i,\Goal) = 1-2^{-i}$ and
$\proba(s_0,\alpha_i,\bad) = 2^{-i}$.
We have $\Prob_{\calM,s_0}^{\sup}(\F[\le 2] \Goal) = 1$, but there is no
scheduler that achieves the supremum reachability probability within two steps.
Also, even in the finitely action-branching case, an optimal scheduler for step-bounded reachability may
require memory~\cite[Figure~6(a)]{BGMR23}, hence the use of $\SchedPureMemoryful$ (and not $\SchedPurePos$) in the statement.

We now restate and prove Lemma~\ref{lem:uniformizingClopen}.
\uniformizingClopen*
\begin{proof}
For convenience, we extend the cylinder notation to \emph{sets} of histories: for a set $K$ of histories, we let $\Cyl(K) = \bigcup_{h \in K} \Cyl(h)$.

For all $n\in\nats$, using that $E_n$ is step-bounded and Lemma~\ref{lem:stepBounded}, there exists $\sigma_n\in\SchedPureMemoryful(\calM)$ such that $\Prob^{\sigma_n}_{\calM,s_0}(E_n) = \Prob^{\inf}_{\calM,s_0}(E_n)$ (resp.\ $\Prob^{\sigma_n}_{\calM,s_0}(E_n) = \Prob^{\sup}_{\calM,s_0}(E_n)$).
Let $\epsilon > 0$.
The construction of $\sigma^\epsilon$ is the same for both cases; we distinguish later whether we consider~$E_n$ as non-decreasing or non-increasing.

Let $(\epsilon_j)_{j\ge 1}$ be a sequence of numbers in $(0,1)$ such that
$\prod_{j=1}^{\infty}(1-\epsilon_j) \ge 1-\epsilon$.
We define a new
scheduler $\sigma^{\epsilon} \in \SchedPureMemoryful(\calM)$
using the above family $(\sigma_n)_{n \in \nats}$, inductively on the
length of histories.  We explicitly state the induction hypothesis: for $k \ge 1$, there is an infinite subset $N_k \subseteq \nats$ and a
finite set $K_k$ of histories of length~$k$ starting from $s_0$ such that:
\begin{itemize}
\item for every $h$ that is a strict prefix of some $h' \in K_k$,
  $\sigma^{\epsilon}(h) = \sigma_n(h)$ for all $n \in N_k$;
\item
  $\Prob^{\sigma^{\epsilon}}_{\calM,s_0}(\Cyl(K_k)) \ge
  \prod_{j=1}^{k}(1-\epsilon_j)$.
\end{itemize}

We first consider the case $k=1$. There is a unique length-$0$
history $h = s_0$. Since~$\calM$ is finitely
action-branching, there is an infinite subset $N_1 \subseteq \nats$ such
that $\sigma_n(s_0)$ is the same action for every $n \in N_1$; we define
$\sigma^{\epsilon}(s_0)$ as this action. Unless $\calM$
is finitely \probBranching, there might be infinitely many outcomes
of $\sigma^{\epsilon}$ of length~$1$. We select a finite number of
such outcomes $K_1$ which has probability at least $1-\epsilon_1$:
hence,
$\Prob^{\sigma^{\epsilon}}_{\calM,s_0}(\Cyl(K_1)) \ge 1-\epsilon_1$.
This shows that the induction hypothesis initially holds.

We now assume that the induction hypothesis holds for some $k \ge 1$.  We
let $N_{k+1} \subseteq N_k$ be an infinite set such that for all
$h \in K_k$ (which have length~$k$), all $\sigma_n(h)$ coincide for
any $n \in N_{k+1}$; we set $\sigma^{\epsilon}(h)$ as this action.
As in the base case, there may be infinitely many outcomes of
$\sigma^{\epsilon}$ of length $k+1$ whose prefixes are in $K_k$. We
select a finite set of such outcomes, with a large relative
probability; formally, let $K_{k+1}$ be a finite set of histories of length~$k + 1$ such that $\Cyl(K_{k+1}) \subseteq \Cyl(K_k)$ and
$\Prob^{\sigma^{\epsilon}}_{\calM,s_0}(\Cyl(K_{k+1}) \mid \Cyl(K_k)) \ge
1-\epsilon_{k+1}$.  Using the induction hypothesis and the definition of conditional probability, we get
$\Prob^{\sigma^{\epsilon}}_{\calM,s_0}(\Cyl(K_{k+1})) \ge
\prod_{j=1}^{k+1}(1-\epsilon_j)$. This ends the induction step.
Note that we have only defined $\sigma^{\epsilon}$ on the histories in sets $K_k$ for some $k$; we extend it arbitrarily to all other histories.

It remains to show that the scheduler $\sigma^{\epsilon}$ we built satisfies the properties of the statement.

We first focus on the case of a non-decreasing sequence of events $(E_n)_n$.
Let $k\ge 1$. By construction, there is an infinite set $N_k$ such that for all $n\in N_k$, $\sigma^{\epsilon}(h) = \sigma_n(h)$ for all strict prefixes of histories in~$K_k$. Take $n_k \in N_k$ with $n_k \ge k$.
We have
\begin{align*}
  \Prob^{\sigma^{\epsilon}}_{\calM,s_0}(E_k)
  &= \Prob^{\sigma^{\epsilon}}_{\calM,s_0}(E_k \cap \Cyl(K_k)) + \Prob^{\sigma^{\epsilon}}_{\calM,s_0}(E_k \cap \neg \Cyl(K_k)) \\
  &\le \Prob^{\sigma^{\epsilon}}_{\calM,s_0}(E_k \cap \Cyl(K_k)) + (1 - \prod_{j=1}^{k}(1-\epsilon_j)) \\
  &= \Prob^{\sigma_{n_k}}_{\calM,s_0}(E_k\cap \Cyl(K_k)) + (1 - \prod_{j=1}^{k}(1-\epsilon_j)) \\
  &\le \Prob^{\sigma_{n_k}}_{\calM,s_0}(E_k) + (1 - \prod_{j=1}^{k}(1-\epsilon_j)) \\
  &\le \Prob^{\sigma_{n_k}}_{\calM,s_0}(E_{n_k}) + (1 - \prod_{j=1}^{k}(1-\epsilon_j))\enspace.
\end{align*}
The first line follows from finite additivity; the second line follows from the construction of $\sigma^{\epsilon}$; the third line follows from the equality of $\sigma^{\epsilon}$ and~$\sigma_{n_k}$ on the prefixes of histories in $K_k$, and the fact that $E_k$ has step bound~$k$; the fourth line follows from the inclusion $E_k \cap \Cyl(K_k) \subseteq E_k$; the last line follows from the non-decreasingness of $(E_n)_n$ and the fact that $n_k \ge k$.
Taking the limit as $k$ grows to $\infty$, by choosing an increasing subsequence~$(n_k)_k$, we get
\begin{align*}
  \lim_{k\to\infty} \Prob^{\sigma^{\epsilon}}_{\calM,s_0}(E_k)
  &\le \lim_{k\to\infty} \Big(\Prob^{\sigma_{n_k}}_{\calM,s_0}(E_{n_k}) + (1 - \prod_{j=1}^{k}(1-\epsilon_j))\Big) \\
  &\le \lim_{n\to\infty} \Prob^{\sigma_{n}}_{\calM,s_0}(E_{n}) + \epsilon \\
  &= \lim_{n\to\infty} \Prob^{\inf}_{\calM,s_0}(E_{n}) + \epsilon\enspace.
\end{align*}
This settles the non-decreasing case.

The claim for a non-increasing sequence $(E_n)_n$ (and $\sup$) follows from the non-decreasing claim (with $\inf$) proved above by taking the complements of events $E_n$.
\end{proof}

\section{Omitted proofs of Section~\ref{sec:constructions}}
\label{app:constructions}

This section is devoted to the deferred proofs of Lemmas~\ref{lem:Mopt},~\ref{lem:HnRn}, and~\ref{lem:slices} about the correctness of our approximation schemes (Section~\ref{sec:constructions}).
We will make use of the following inequality.
\begin{rem}
  \label{fact:union-ub-sup}
  Let $\calM$ be an MDP, $s$ be a state of $\calM$, $\opt \in
  \{\inf,\sup\}$, and $\phi_1$ and~$\phi_2$ be two events. For every scheduler $\sigma$,
  $\Prob_{\calM,s}^\sigma(\phi_1 \vee \phi_2) \le
  \Prob_{\calM,s}^\sigma(\phi_1) +
  \Prob_{\calM,s}^{\sup}(\phi_2)$.
  Thus:
  \[
    \Prob_{\calM,s}^\opt(\phi_1 \vee \phi_2) \le \Prob_{\calM,s}^\opt(\phi_1) +
    \Prob_{\calM,s}^{\sup}(\phi_2) \enspace.
  \]
\end{rem}

The first two deferred proofs deal with events that are used in the first approximation scheme (Section~\ref{subsec:collapse}).

\Mopt*
\begin{proof}
  We first consider the case $\opt = \inf$.
  Any scheduler $\sigma$ in $\calM$ straightforwardly yields a scheduler $\sigma'$
  in $\calM^{\inf}$ by mimicking $\sigma$ until $\bad^{\inf}$ is reached
  (if this occurs) and then behaving arbitrarily.
  We have that
  \begin{align*}
    \Prob^{\sigma'}_{\calM^{\inf},s_0}(\F \Goal) & = \Prob^{\sigma'}_{\calM^{\inf},s_0}(\F \Goal \wedge \G \neg \bad^{\inf})\\
    & = \Prob^\sigma_{\calM,s_0}(\F \Goal \wedge \G \neg \Avoid^{\inf}_{\calM}(\Goal))\\
    & \le \Prob^\sigma_{\calM,s_0}(\F \Goal) \enspace.
  \end{align*}
  Thus,
  $\Prob^{\inf}_{\calM^{\inf},s_0}(\F \Goal) \le
  \Prob^{\inf}_{\calM,s_0}(\F \Goal)$.

  To prove the other inequality, for every scheduler $\sigma$ in
  $\calM^{\inf}$, and every $\epsilon\in(0, 1)$, we show that there exists
  a scheduler $\sigma'$ in $\calM$ with
  $\Prob^{\sigma'}_{\calM,s_0}(\F \Goal) \le
  \Prob^\sigma_{\calM^{\inf},s_0}(\F \Goal) + \epsilon$.  To do so,
  for every state
  $s \in \Avoid^{\inf}_\calM(\Goal)$, let $\sigma^s_\epsilon$ be a
  scheduler achieving
  $\Prob^{\sigma^s_\epsilon}_{\calM,s} (\F \Goal) \le \epsilon$. This is possible
  by definition of $\Avoid^{\inf}_\calM(\Goal)$.  We build a
  scheduler~$\sigma'$ which follows $\sigma$ until the
  first visit to $\Avoid^{\inf}_\calM(\Goal)$, and then switches to
  $\sigma^s_\epsilon$ (where $s$ is the first state visited in
  $\Avoid^{\inf}_\calM(\Goal)$). Then,
  $\Prob^{\sigma'}_{\calM,s_0}(\F \Goal) \le
  \Prob^\sigma_{\calM^{\inf},s_0}(\F \Goal) + \epsilon$ and thus
  $\Prob^{\inf}_{\calM,s_0}(\F \Goal) \le
  \Prob^{\inf}_{\calM^{\inf},s_0}(\F \Goal) + \epsilon$. Since this
  holds for every $\epsilon > 0$, we obtain that
  $\Prob^{\inf}_{\calM,s_0}(\F \Goal) \le
  \Prob^{\inf}_{\calM^{\inf},s_0}(\F \Goal)$.

  Let us now consider the case $\opt=\sup$. Observe that
  any path in $\calM$ that reaches~$\Goal$ satisfies
  $\G \neg \Avoid^{\sup}_{\calM}(\Goal)$ (recall that $\Goal$ is absorbing). Indeed, by definition of
  $\Avoid^{\sup}_{\calM}(\Goal)$, $\F \Avoid^{\sup}_{\calM}(\Goal)$
  implies $\G \neg \Goal$. Therefore, for any scheduler $\sigma$ in
  $\calM$,
  \[\Prob^\sigma_{\calM, s_0}(\F \Goal) = \Prob^\sigma_{\calM, s_0}(\F \Goal
  \wedge \G \neg \Avoid^{\sup}_{\calM}(\Goal))\enspace.\] We use this fact to
  prove both inequalities on supremum probability values.

  Pick a scheduler $\sigma$ in $\calM^{\sup}$, and lift it to
  a scheduler $\sigma'$ in $\calM$ that mimics $\sigma$, unless
  $\Avoid^{\sup}_\calM(\Goal)$ is reached---in this case, it behaves
  arbitrarily. We have:
  \begin{align*}
    \Prob^{\sigma'}_{\calM,s_0}(\F \Goal) & =     \Prob^{\sigma'}_{\calM,s_0}(\F \Goal \wedge \G \neg \Avoid^{\sup}_\calM(\Goal))\\
                                          & =  \Prob^\sigma_{\calM^{\sup},s_0}(\F \Goal \wedge \G \neg \bad^{\sup})\\
    & = \Prob^\sigma_{\calM^{\sup},s_0}(\F \Goal) \enspace.
  \end{align*}
  Therefore $\Prob^{\sup}_{\calM^{\sup},s_0}(\F \Goal) \le \Prob^{\sup}_{\calM,s_0}(\F \Goal)$.

  Conversely,  with the same
  construction as above, for any scheduler $\sigma$ in $\calM$, we
  build a scheduler $\sigma'$ in $\calM^{\sup}$ that mimics $\sigma$
  until $\bad^{\sup}$ is reached. We observe that
  \begin{align*}
    \Prob^\sigma_{\calM, s_0}(\F \Goal) &= \Prob^\sigma_{\calM, s_0}(\F \Goal
                                     \wedge \G \neg \Avoid^{\sup}_{\calM}(\Goal))\\
                                   & = \Prob^{\sigma'}_{\calM^{\sup}, s_0}(\F \Goal \wedge \G \neg \bad^{\sup})\\
    & =  \Prob^{\sigma'}_{\calM^{\sup}, s_0}(\F \Goal) \enspace.
  \end{align*}
  Thus
  $\Prob^{\sup}_{\calM^{\sup},s_0}(\F \Goal) \ge
  \Prob^{\sup}_{\calM,s_0}(\F \Goal)$, which allows us to conclude.
\end{proof}

\HnRn*
\begin{proof}
  Let us show that for every $n$,
  \[
    \sem{R_n}{\calM^{\opt},s_0} \subseteq \sem{\F
    \Goal}{\calM^{\opt},s_0} \subseteq \sem{R_n \vee
    H_n^{\opt}}{\calM^{\opt},s_0}\enspace.
  \]
  The first inclusion follows from the definition of $R_n$. To
  prove the second,
  observe that the complement of
  $\sem{R_n \vee H_n^{\opt}}{\calM^{\opt},s_0}$ is
  $\sem{(\G[\le n] \neg \Goal) \wedge (\F[\le n]
    \bad^{\opt})}{\calM^{\opt},s_0} = \sem{\F[\le n]
    \bad^{\opt}}{\calM^{\opt},s_0}$.
    Any path in that set thus reaches
  the losing absorbing state $\bad^{\opt}$ in the first $n$ steps, and
  therefore belongs to $\sem{\neg \F \Goal}{\calM^{\opt},s_0}$.
  Finally, the last inequality is a direct application of
  Remark~\ref{fact:union-ub-sup}.
\end{proof}

We now focus on the deferred proof for the second approximation scheme (Lemma~\ref{lem:slices}, Section~\ref{subsec:sliced}).

\slices*
\begin{proof}\hfill
  \begin{enumerate}
  \item For the first item, notice that every path which reaches $\Goal$ in
  $\calM^{\opt}_n$ also reaches~$\Goal$ in $\calM^\opt$. Also,
  every path which reaches $\Goal$ in $\calM^\opt$ can be projected onto $\calM^{\opt}_n$: it either reaches $\Goal$ or ends up in $s_\bot^n$.
  \item The second item is a direct application of
  Remark~\ref{fact:union-ub-sup}.
  \item For the third item, notice that every path in $\calM^\opt$ which reaches $\Goal$ within~$n$ steps also reaches $\Goal$ in $\calM^{\opt}_n$ (avoiding
    $s_\bot^n$).
  \item For the fourth item, notice that every path which reaches $\Goal$ or~$s_\bot^n$
    in $\calM^{\opt}_n$ either reaches~$\Goal$ within $n$ steps or
    visits neither $\Goal$ nor $\bad^\opt$ in $\calM^\opt$ during
    the $n$ first steps.
    \qedhere
  \end{enumerate}
\end{proof}

\section{Omitted proofs from Section~\ref{sec:convergingSchemes}}
\label{app:convergence}

This section is devoted to the missing proofs of Theorem~\ref{thm:approx-opt}, about the convergence of our approximation schemes, of Theorem~\ref{thm:approxSchemeTwoImpliesOne}, about a strong link between the convergence of the two approximation schemes (Section~\ref{subsec:scheme1}), and of Lemma~\ref{lem:lemmaNonFleeing}, about a relaxed condition for the convergence of $\ApproxScheme_\sliced^{\sup}$ (Section~\ref{subsec:nonFleeing}).

We prove Theorem~\ref{thm:approx-opt} using the next two lemmas.
We fix $\opt \in \{\inf,\sup\}$. We let~$\calM = (S,\Act,\proba)$ be an~MDP,
$s_0 \in S$ be an initial state, and $\Goal$ be a target state. We first
prove that the decisiveness of $\calM$ transfers to
$\calM^{\opt}$ (under the finitely action-branching assumption when $\opt = \inf$).

\begin{lem} \label{lem:optdecisiveMopt}
  If $\opt = \inf$, assume that $\calM$ is
  finitely action-branching. If~$\calM$ is $\opt$-decisive
  w.r.t.~$\Goal$ from $s_0$, then so is $\calM^{\opt}$.
\end{lem}
\begin{proof}
  Assume that $\calM$ is $\opt$-decisive
  w.r.t.~$\Goal$ from $s_0$.
  We show that $\calM^{\opt}$ is $\opt$-decisive w.r.t.~$\Goal$ from $s_0$.

  We fix a pure and positional
  scheduler $\sigma_0 \in \SchedPurePos(\calM)$ on $\Avoid^{\opt}_\calM(\Goal)$ such that for
  every $s \in \Avoid^{\opt}_\calM(\Goal)$,
  $\Prob^{\sigma_0}_{\calM,s}(\F \Goal) = 0$; any scheduler will do
  the job if $\opt = \sup$, while one will be given by item~\ref{item:infOptPurePos}
  of Lemma~\ref{lem:schedulerComplexity} when $\opt = \inf$ (thanks to the
  finitely action-branching assumption). In particular, starting from
  $\Avoid^{\opt}_\calM(\Goal)$, paths under scheduler $\sigma_0$ only
  visit states in $\Avoid^{\opt}_\calM(\Goal)$.

  Let $\sigma \in \SchedPurePos(\calM^{\opt})$ be a
  pure and positional scheduler in $\calM^{\opt}$; we show that
  $\Prob^\sigma_{\calM^{\opt},s_0}(\F \Goal \vee \F
  \Avoid^{\opt}_{\calM^{\opt}}(\Goal)) = 1$.  Observe that by
  construction of $\calM^{\opt}$,
  $\Avoid^{\opt}_{\calM^{\opt}}(\Goal) = \{\bad^{\opt}\}$.

  We define a scheduler $\sigma'$ in $\calM$ that mimics $\sigma$ until $\Avoid^{\opt}_{\calM}(\Goal)$ is reached, and then behaves as $\sigma_0$. Gluing these two pure and positional
  partial schedulers, $\sigma'$ is a pure and positional scheduler
  that satisfies
  $\Prob^{\sigma'}_{\calM,s_0} (\F \Goal \vee \F
  \Avoid^{\opt}_\calM(\Goal)) = \Prob^\sigma_{\calM^{\opt},s_0}(\F
  \Goal \vee \F \bad^{\opt})$.
  Since $\calM$ is $\opt$-decisive w.r.t.\ $\Goal$ from~$s_0$, we have $\Prob^{\sigma'}_{\calM,s_0} (\F \Goal \vee \F \Avoid^{\opt}_\calM(\Goal)) = 1$.

  Combining all of the above, we have
  \begin{align*}
    \Prob^\sigma_{\calM^{\opt},s_0}(\F \Goal \vee \F \Avoid^{\opt}_{\calM^{\opt}}(\Goal))
    &= \Prob^\sigma_{\calM^{\opt},s_0}(\F \Goal \vee \F \bad^{\opt}) \\
    &= \Prob^{\sigma'}_{\calM,s_0} (\F \Goal \vee \F \Avoid^{\opt}_\calM(\Goal)) \\
    &= 1\enspace,
  \end{align*}
  which concludes the proof.
\end{proof}

To show Theorem~\ref{thm:approx-opt} (the convergence of scheme
$\ApproxScheme_\unfold^{\opt}$ when~$\calM$ is finitely action-branching and
$\opt$-decisive), it suffices to prove the hypotheses of
Theorem~\ref{thm:genericCorrectnessTermination} (i.e., that $\lim_{n \to \infty} \Prob^{\sup}_{\calM^{\opt},s_0}(H^{\opt}_n) = 0$).
\begin{lem} \label{lem:extraction}
  If $\calM$ is finitely action-branching and $\opt$-decisive
  w.r.t.\ $\Goal$ from~$s_0$, then
  \[
    \lim_{n \to \infty} \Prob^{\sup}_{\calM^{\opt},s_0}(H^{\opt}_n) =
    0 \enspace.
  \]
\end{lem}
\begin{proof}
  This proof uses the following equality (which follows from Theorem~\ref{thm:limoptEqualsoptlim}; the sequence~$(H_n^{\opt})_n$ is non-increasing, and for all $n$, $H_n^{\opt}$ has step bound~$n$):
  \begin{align*}
    \lim_{n\to\infty} \Prob^{\sup}_{\calM^{\opt},s_0}(H_n^{\opt})
    &= \sup_{\sigma} \lim_{n\to\infty} \Prob_{\calM^{\opt},s_0}^\sigma(H_n^{\opt})\enspace.
  \end{align*}

  Given the continuity of probabilities, we have
  \[
    \sup_{\sigma} \lim_{n\to\infty} \Prob_{\calM^{\opt},s_0}^\sigma(H_n^{\opt})
    = \sup_{\sigma} \Prob_{\calM^{\opt},s_0}^\sigma(\G (\neg \Goal \wedge
    \neg \bad^{\opt}))\enspace.
  \]

  For the safety objective
  $\G (\neg \Goal \wedge \neg \bad^{\opt})$, applying
  Lemma~\ref{lem:schedulerComplexity} (item~\ref{item:infOptPurePos}) to the finitely action-branching $\calM^{\opt}$, there is a pure and positional
  scheduler $\sigma^{\!\star}$ such that
  \begin{align*}
    \sup_{\sigma} \Prob_{\calM^{\opt},s_0}^\sigma(\G (\neg \Goal \wedge
    \neg \bad^{\opt}))
    &= \Prob_{\calM^{\opt},s_0}^{\sigma^{\!\star}}(\G (\neg \Goal
    \wedge \neg \bad^{\opt}))\\
    &= 1 - \Prob_{\calM^{\opt},s_0}^{\sigma^{\!\star}}(\F \Goal \vee
    \F \bad^{\opt})\enspace.
  \end{align*}

  Moreover, $\calM$ is $\opt$-decisive w.r.t.~$\Goal$ from $s_0$, hence so is $\calM^{\opt}$ (Lemma~\ref{lem:optdecisiveMopt}). Since
  $\Avoid^{\opt}_{\calM^{\opt}}(\Goal) = \{\bad^{\opt}\}$, we have
  $\Prob_{\calM^{\opt},s_0}^{\sigma^{\!\star}}(\F \Goal \vee
    \F \bad^{\opt}) = 1$.
  Combining all the equalities above, we deduce that $\lim_{n\to\infty} \Prob^{\sup}_{\calM^{\opt},s_0}(H_n^{\opt}) = 0$.
\end{proof}

\begin{proof}[Proof of Theorem~\ref{thm:approx-opt}]
  By Lemma~\ref{lem:optdecisiveMopt}, the collapsed MDP $\calM^{\opt}$ is also $\opt$-decisive w.r.t.~$\Goal$ from~$s_0$.
  Lemma~\ref{lem:extraction} then yields
  \[
    \lim_{n\to\infty} \Prob^{\sup}_{\calM^{\opt},s_0}(H_n^{\opt}) = 0.
  \]
  Hence, the hypothesis of Theorem~\ref{thm:genericCorrectnessTermination} is satisfied.
  It follows that $\ApproxScheme_\unfold^{\opt}$ provides a converging approximation scheme for $\Prob^{\opt}_{\calM,s_0}(\F\Goal)$, and therefore $\ApproxScheme_\unfold^{\opt}$ converges on $\calM$ from~$s_0$.
\end{proof}

We now move on to the proof of Theorem~\ref{thm:approxSchemeTwoImpliesOne}, which states that the convergence of $\ApproxScheme_\sliced^{\inf}$ implies the convergence of $\ApproxScheme_\unfold^{\inf}$ under finite branching.

\approxSchemeTwoImpliesOne*
\begin{proof}
  Assume that $\ApproxScheme_\unfold^{\inf}$ does not converge on $\calM$ from $s_0$.
  This implies that there exists $\epsilon > 0$ such that $\lim_n p_n^{\inf, +} \ge \lim_n p_n^{\inf, -} + \epsilon$.
  Thanks to Theorem~\ref{thm:limoptEqualsoptlim} and the continuity of probabilities, we have $\lim_n p_n^{\inf, +} = \Prob^{\inf}_{\calM^{\inf},s_0}(\F \Goal \vee \G(\neg \Goal \wedge
  \neg \bad^{\inf}))$ and $\lim_n p_n^{\inf, -} = \Prob^{\inf}_{\calM^{\inf},s_0}(\F \Goal)$.
  Hence,
  \[\Prob^{\inf}_{\calM^{\inf},s_0}(\F \Goal \vee \G(\neg \Goal \wedge
  \neg \bad^{\inf})) \ge \Prob^{\inf}_{\calM^{\inf},s_0}(\F \Goal) + \epsilon\enspace.\]

  Fix $n\in\nats$ and $\tau\in\Sched(\calM_n^{\inf})$, and let $\sigma\in\Sched(\calM^{\inf})$ mimic $\tau$ until $\Reachable{n}$ is left. Since $\Reachable{n}$ is finite and every state in $\Reachable{n}\setminus\{\bad^{\inf}\}$ has a positive infimum reachability probability, these probabilities are uniformly bounded away from zero. Thus, under $\sigma$, a path satisfying $\G(\neg\Goal\wedge\neg\bad^{\inf})$ leaves $\Reachable{n}$ almost surely. Hence,
  \[
    \Prob^\tau_{\calM_n^{\inf},s_0}(\F(\Goal\vee s_\bot^n))
    \ge \Prob^\sigma_{\calM^{\inf},s_0}(\F\Goal\vee\G(\neg\Goal\wedge\neg\bad^{\inf}))
    \ge \Prob^{\inf}_{\calM^{\inf},s_0}(\F\Goal)+\epsilon\enspace.
  \]
  Since $n$ and $\tau$ are arbitrary, $q_n^{\inf,+}\ge \Prob^{\inf}_{\calM^{\inf},s_0}(\F\Goal)+\epsilon$, whereas Lemma~\ref{lem:slices} gives $q_n^{\inf,-}\le \Prob^{\inf}_{\calM^{\inf},s_0}(\F\Goal)$. Thus, $\ApproxScheme_\sliced^{\inf}$ does not converge on $\calM$ from $s_0$.
\end{proof}

We finally show Lemma~\ref{lem:lemmaNonFleeing}, which is the key step in the proof of Theorem~\ref{thm:nonFleeing}.
\lemmaNonFleeing*
\begin{proof}
  Towards a contradiction, assume that
  there exists $\epsilon>0$ and an infinite set $N\subseteq \nats$ such that for every $n \in N$,
  \[
    \Prob^{\sup}_{\calM_n^{\sup},s_0} (\F s_\bot^n) \ge \epsilon
    \enspace.
  \]
  Since $\calM$ is finitely branching, $\calM_n^{\sup}$ is a finite
  MDP. Hence, by Lemma~\ref{lem:finiteMDP}, for every $n \in N$,
  there is $\sigma_n \in \SchedPurePos(\calM_n^{\sup})$ such that
  $\Prob^{\sigma_n}_{\calM_n^{\sup},s_0} (\F s_\bot^n) =
  \Prob^{\sup}_{\calM_n^{\sup},s_0} (\F s_\bot^n)$. The scheduler
  $\sigma_n$ is defined on $S_n^{\sup}$, which is a finite
  set.

  Following similar but simpler arguments to those in the proof of
  Lemma~\ref{lem:uniformizingClopen}, we ``uniformise'' this family
  of schedulers as follows. We let $N_1 \subseteq N\cap\{n\ge 2\}$ be an infinite
  set of indices such that all $\sigma_n$'s with $n \in N_1$
  coincide on $\Reachable{1}$; we define $\sigma^{\!\star}$ on
  $\Reachable{1}$ in accordance with these schedulers. Then, for every
  $i \ge 1$, we let $N_{i+1} \subseteq N_i\cap\{n\ge i+2\}$ be an infinite set of
  indices such that all $\sigma_n$'s with $n \in N_{i+1}$ coincide
  on $\Reachable{i+1}$; we then define $\sigma^{\!\star}$ on
  $\Reachable{i+1}$ in accordance with these schedulers. Globally, the
  scheduler $\sigma^{\!\star}$ is defined on $\bigcup_{i \ge 1} \Reachable{i}$, and we define it arbitrarily on the other states. For every
  $i \ge 1$, let $n_i \in N_i$
  be such that $\sigma^{\!\star}$ coincides with
  $\sigma_{n_i}$ on $\Reachable{i}$. By construction of
  $\calM_{n_i}^{\sup}$,
  \begin{align*}
    \Prob^{\sigma^{\!\star}}_{\calM^{\sup},s_0} (\F (\Reachable{i+1}
    \setminus \Reachable{i})) & =
      \Prob^{\sigma_{n_i}}_{\calM_{n_i}^{\sup},s_0}
      (\F (\Reachable{i+1}
      \setminus \Reachable{i})) \\
    & \ge \Prob^{\sigma_{n_i}}_{\calM_{n_i}^{\sup},s_0}
      (\F s_\bot^{n_i}) ~ \text{since $\F s_\bot^{n_i} \subseteq \F (\Reachable{i+1}
      \setminus \Reachable{i})$} \\
    & \ge
      \epsilon \enspace.
  \end{align*}
  Since the event $\F (\Reachable{m_2+1} \setminus \Reachable{m_2})$
  is included in $\F (\Reachable{m_1+1} \setminus \Reachable{m_1})$ as
  soon as $m_2 \ge m_1$, we get that
  \[
    \Prob^{\sigma^{\!\star}}_{\calM^{\sup},s_0}
    (\Div) \ge \epsilon \enspace.
  \]
  This is a contradiction with the non-fleeing hypothesis. We conclude
  that
  \[
    \limsup_{n \to \infty} \Prob^{\sup}_{\calM_n^{\sup},s_0} (\F
    s_\bot^n) =0 \enspace,
  \]
  hence
  \[
    \lim_{n \to \infty} \Prob^{\sup}_{\calM_n^{\sup},s_0} (\F
    s_\bot^n) =0 \enspace.\qedhere
  \]
\end{proof}

\section{Undecidability of the value-\texorpdfstring{$1$}{1} problem for NPLCSs (Section~\ref{subsec:nplcs})} \label{app:nplcs}

In this section, we provide the missing proof of Theorem~\ref{thm:undec} from Section~\ref{subsec:nplcs}: we show that the \emph{value-$1$ problem} is undecidable for NPLCSs.

\undec*
\begin{proof}
We present a reduction from the \emph{boundedness} problem for lossy
channel systems~(LCSs). The latter is known to be undecidable, even
with a single channel~\cite{Mayr03}. Since the core of the paper
only introduces channel systems and their nondeterministic and
probabilistic variant, we define here lossy channel systems.

A \emph{lossy channel system} is syntactically identical to a channel
system, i.e., it is a tuple $\calS = (Q,\calC,\Mess,\ActLCS,\Delta)$. Its
semantics differs in that messages may be lost arbitrarily when performing
a step. A step in a lossy channel system is composed of the perfect
application of a rule (receive, send, or internal action, as in a channel
system), followed by message losses. The latter are formalised using the
subword ordering on channel contents. For $v \in \Mess^*$, let~${\downarrow
v}$ denote the set of subwords of $v$. For a rule
$\delta=(q,\op,q')\in\Delta$, we write
$(q,\ccont)\step{\delta}(q',\ccont')$ whenever
\begin{itemize}
  \item if $\op = c?m$, then $\ccont(c)=mv$ for some $v\in\Mess^*$,
  $\ccont'(c) \in {\downarrow v}$, and, for every $c' \neq c$,
  $\ccont'(c') \in {\downarrow \ccont(c')}$;
  \item if $\op = c!m$, then $\ccont'(c) \in
  {\downarrow(\ccont(c)m)}$ and, for every $c'\neq c$, $\ccont'(c') \in
  {\downarrow \ccont(c')}$;
  \item if $\op=\ell \in \ActLCS$, then, for every $c \in \calC$,
  $\ccont'(c) \in {\downarrow \ccont(c)}$.
\end{itemize}

A lossy channel system with a fixed initial configuration $(q_0,\emptyConfig)$ is \emph{bounded} if the set of reachable configurations from it is
finite.  Given such an LCS $\calL$ with a single channel over message
alphabet~$\Sigma$, consider the NPLCS $\calN$ represented in
Figure~\ref{fig:undec-NPLCS}, with $\lambda = \frac12$ a fixed loss rate. In that reduction, $\ell_{\mathsf{restart}}$,
$\ell_{\mathsf{try}}$, and $\ell_{\mathsf{lose}}$ all label internal
actions with no effect on the channel contents, and $? \Sigma$ is a
shortcut for multiple transitions $? a$ for every $a \in \Sigma$.

\begin{figure}[htbp]
\centering
\begin{tikzpicture}
\path (.5,0) node[j0] (q0) {$q_0$};
\path (2.5,0) node[j0] (q) {$q$};
\path (5,0) node[j0] (qtry) {$q_{\mathsf{try}}$};
\path (7.5,.5) node[j0] (goal) {$\Goal$};
\path (7.5,-.5) node[j0] (sink) {$\bad$};

\path (1.5,-2) node[j0] (cg) {$\begin{array}{c}\mathtt{Cleaning}\\\mathtt{Gadget}\end{array}$};

\path (-.15,.7) node (tag) {$\calL$};
\path (-.95,0) node (tag) {$\calN:$};

\draw [latex'-] (q0) -- +(-.8,0);

\path (cg) edge[-latex',bend left]  (q0);

\path (q) edge[-latex',bend left] node[pos=.85,right]
{$\ell_{\mathsf{restart}}$} (cg);

\path (q) edge[-latex'] node[pos=.7,above]
{$\ell_{\mathsf{try}}$} (qtry);

\path (qtry) edge[-latex',bend left] node[pos=.5,above]
{$? \Sigma$} (goal);

\path (qtry) edge[-latex',bend right] node[pos=.5,below]
{$\ell_{\mathsf{lose}}$} (sink);

\draw[densely dashed] (-.5,1) -- (3.5,1) -- (3.5,-.7) -- (-.5,-.7) --
cycle;
\end{tikzpicture}
\caption{Reduction to prove the undecidability of the value-$1$
  problem for NPLCSs. The $\ell_{\mathsf{try}}$ and $\ell_{\mathsf{restart}}$ transitions are available from every control state $q$ of $\calL$. The cleaning gadget ensures that the channel is empty when returning to $q_0$.}
\label{fig:undec-NPLCS}
\end{figure}

Compared to $\calL$, $\calN$ adds an option to attempt reaching one of the absorbing states. From
state $q_{\textsf{try}}$, the target can be reached if the channel
contents were not emptied when moving to
$q_{\textsf{try}}$. Additionally, from any state $q$, the scheduler
may decide to restart a computation and move back to the initial
state $q_0$ with empty channel contents. This is achieved by a
cleaning gadget --- detailed in~\cite[Lemma~5.2]{BBS-acmtocl07} --- which
ensures that (1) under any scheduler, whenever the cleaning gadget is exited into $q_0$, the
channel is empty, and (2) there exists a scheduler under which the cleaning gadget is exited almost surely.

Write $\calM$ for the MDP induced by $\calN$ with a fixed loss rate
$\lambda$, and let $s_0 = (q_0,\emptyConfig)$ be the initial
configuration. We claim that this construction enjoys the following
properties:
\begin{itemize}
\item if $\calL$ is bounded, there exists $p>0$ such that
$\Prob^{\sup}_{\calM,s_0} (\F \Goal) \le 1 - p <1$;
\item if $\calL$ is unbounded,
$\Prob_{\calM,s_0}^{\sup} (\F \Goal) = 1$.
\end{itemize}

Assume first that $\calL$ is bounded. In order to reach $\Goal$, a
scheduler must move to state $q_{\mathsf{try}}$ with internal action
$\ell_{\mathsf{try}}$ at some point. If all messages from $\calL$
are lost in that step, the only possibility is to move to
$\bad$. Otherwise, the scheduler chooses to move to $\Goal$.

Let $\ell$ be the maximal channel length among configurations of $\calL$ reachable from $(q_0,\emptyConfig)$, and let $p=\lambda^\ell>0$. Whenever the scheduler chooses to try
and reach~$\Goal$, the probability is at least~$p$ to move to~$\bad$. All in all, under any scheduler~$\sigma$,
$\Prob^\sigma_{\calM, s_0}(\F \Goal) \le 1 - p$, and thus
$\Prob^{\sup}_{\calM, s_0}(\F \Goal) \le 1 - p$. Note that $p$ only
depends on the LCS $\calL$ and on the loss rate~$\lambda$.

Assume now that $\calL$ is unbounded, and fix $\eta >0$. Consider a
reachable configuration~$(q,w)$, large enough to have
$\lambda^{|w|} \le \eta$. Thus from configuration $(q,w)$, when
$\ell_{\mathsf{try}}$ is played, the probability to move to $\bad$
is at most $\eta$. We define a scheduler $\sigma$ as follows. The
objective of $\sigma$ is to reach configuration $(q,w)$ and then
play action $\ell_{\mathsf{try}}$ to move to~$q_{\mathsf{try}}$. Since $(q,w)$ is reachable, there is a finite sequence
of actions whose execution reaches it with positive probability from
the initial configuration~$(q_0,\emptyConfig)$. If this path is not
realised (due to unwanted message losses), $\sigma$ changes mode and
restarts the simulation by moving to the cleaning gadget. So
defined, $\sigma$ eventually succeeds in reaching $(q,w)$ almost
surely, and thus ensures
$\Prob^\sigma_{\calM, s_0}(\F \Goal) = 1 - \lambda^{|w|} \ge 1
-\eta$. Since this holds for every~$\eta$, we obtain that~$\Prob^{\sup}_{\calM, s_0}(\F \Goal) = 1$.
\end{proof}

\end{document}